\documentclass[a4paper,english]{article}
\pdfoutput=1
\usepackage{amsmath}
\usepackage{amssymb}
\usepackage{algorithm2e}
\usepackage{mathrsfs}

\makeatletter

\usepackage[english]{babel}
\usepackage{authblk}
\usepackage{amsfonts}
\usepackage{amsthm}
\usepackage{color}
\usepackage{multicol}
\usepackage{empheq}
\usepackage{graphicx}
\usepackage{float}
\usepackage{bbm}
\usepackage[colorlinks,linkcolor=blue,anchorcolor=black,citecolor=red]{hyperref}

\usepackage{calligra}

\DeclareMathAlphabet{\mathcalligra}{T1}{calligra}{m}{n}
\DeclareFontShape{T1}{calligra}{m}{n}{<->s*[2.2]callig15}{}

\usepackage[list=true]{subcaption}
\usepackage{babel}

\newtheorem{lemma}{Lemma}

\newcommand{\R}{{\mathbb{R}}}

\DeclareMathOperator*{\argmax}{arg\,max}

\DeclareMathOperator*{\argmin}{argmin}

\let\e\varepsilon

\makeatother

\usepackage[left=26mm, right=24mm]{geometry}

\usepackage{placeins}

\newboolean{show_LOUT}
\setboolean{show_LOUT}{false}

\usepackage[T1]{fontenc}
\usepackage{eurosym}
\def\yen{{\setbox0=\hbox{Y}Y\kern-.97\wd0\vbox{\hrule height.1ex
width.98\wd0\kern.33ex\hrule height.1ex width.98\wd0\kern.45ex}}}

\usepackage{array}  
 \usepackage{multirow,bigdelim}

\begin{document}

%
%
%
%

\title{An MBO scheme for clustering and semi-supervised clustering of signed networks}

\author[a,b]{\footnote{The authors are listed in alphabetical order.}Mihai Cucuringu}
\author[b,c,d]{Andrea Pizzoferrato}
\author[e]{Yves van Gennip}
\affil[a]{Department of Statistics and Mathematical Institute, University of Oxford, Oxford, UK}
\affil[b]{The Alan Turing Institute, London, UK}
\affil[c]{Department of Mathematics, Imperial College, London, UK}
\affil[d]{School of Mathematical Sciences, Queen Mary University of London, London, UK}
\affil[e]{Delft Institute of Applied Mathematics, Delft University of Technology, Delft, NL}

\renewcommand\Authands{ and }

\maketitle
%
%
%
%
\begin{abstract}

We introduce a principled method for the signed clustering problem, where the goal is to partition a weighted undirected graph whose edge weights take both positive and negative values, such that edges within the same cluster are mostly positive, while edges spanning across clusters are mostly negative. 
%
Our  method relies on a graph-based diffuse interface model formulation utilizing the Ginzburg--Landau functional, based on an adaptation of the classic numerical Merriman--Bence--Osher (MBO) scheme for minimizing such graph-based functionals. The proposed objective function  aims to minimize the total weight of inter-cluster positively-weighted edges, while maximizing the total weight of the inter-cluster negatively-weighted edges. Our method scales to 
large sparse networks, and can be easily adjusted to incorporate labelled data information, as is often the case in the context of semi-supervised learning. We tested our method  on a number of both synthetic stochastic block models  and real-world data sets (including financial correlation matrices),
and obtained promising results that compare favourably against a number of state-of-the-art approaches from the recent literature.
\end{abstract}

%
%
\tableofcontents{}

\vspace{5mm}

%
%
\section{Introduction}
Clustering is one of the most widely used techniques in data analysis, and aims to identify groups of nodes  in a network  that  have similar behaviour or features.  It is arguably one of the first tools one uses when facing the task of understanding the structure of a network. 

Spectral clustering methods have become a fundamental tool with a broad range of applications  in applied mathematics and computer science, and  tremendous impact on  areas including scientific computing, network science, computer vision, machine learning and data mining, such as ranking and web search algorithms. On the theoretical side, understanding  the properties of  the spectrum (both eigenvalues and eigenvectors) of the adjacency matrix and its various corresponding Laplacians, is crucial for the development of efficient algorithms with performance guarantees, and leads to a very mathematically rich set of open problems.

Variational methods have a long history in mathematics and physics. In a variational model, the aim is to find minimizers of a given functional (which is often called an `energy' in recognition of the field's origins, which are rooted in physics, even in situations where this terminology is not physically justified). When formulated in a continuum setting, such models are closely related to partial differential equations (PDEs) which arise as the Euler--Lagrange equations or gradient flow equations in the model. These equations are often nonlinear and can have other interesting properties, such as nonlocality. In recent years it has been recognized that such properties can be useful for clustering problems as well \cite{merkurjev2013mbo,HuiyiTVModularity}. When attempting a hard clustering problem, i.e. a clustering problem in which each node is to be assigned to exactly one cluster, the approach is usually to relax the exact combinatorial problem. For example, in the case of \textit{spectral relaxations}, one typically leverages the spectrum of certain graph-based matrix operators, in particular their eigenvalues and especially eigenvectors.  
This introduces the question of how to translate the real-valued outcomes of the relaxed problem into a discrete cluster label. To this end, non-linearities can prove useful as they can be used to drive the real-valued approximate solution away from the decision boundaries. Moreover, what is non-local in a continuum setting, can be made local in a graph setting, by adding edges between nodes that are located far away in continuum space. Furthermore, since the equations which result from graph-based variational methods  
tend to resemble discretized PDEs, the large existing literature on numerical methods for solving PDEs often provides fast algorithms for solving such equations. Considerations of similar nature have led to a growing interest in the use of variational methods in a graph setting  \cite{bertozzi2016diffuse,merkurjev2015variational,vanGennipGuillenOstingBertozzi14,jordan1999introduction}.  

In certain applications, negative weights can be used to denote dissimilarity or distance between a pair of nodes in a network. The analysis of such signed networks has become an increasingly important research topic in the past decade, and examples include social networks that  contain both friend and foe links,   shopping bipartite networks that encode like and dislike relationships between users and products  \cite{banerjee2012partitioning}, and online news and  review websites such as Epinions \cite{epinions} and Slashdot \cite{slashdot} that allow users to approve or denounce others  \cite{leskovec2010predicting}. Such networks may be regarded as signed graphs, where a positive edge weight captures the degree of trust or similarity, while a negative edge weight captures the degree of distrust or dissimilarity. 
Further efforts for the analysis of signed graphs have lead to 
a wide range of applications including 
edge prediction~\cite{kumar2016wsn,leskovec2010predicting}, 
node classification~\cite{Bosch_2018_nodeclassification,tang2016nodeClassification},  
node embeddings~\cite{Chiang_2011_ELC,derr_2018_signedGraphConvolutionalNetwork,Kim_2018_SIDE_WWW, Wang_2017_signedNetworkEmbedding}, 
and node ranking~\cite{Chung_2013,Shahriari_2014}.
We refer the reader to~\cite{tang2015survey} for a recent survey on the topic.

Finally, another main motivation for our proposed line of work 
arises in the area of time series analysis, where clustering time series data is a task of particular importance \cite{aghabozorgi2015time}. It arises in a number of applications in a diverse set of fields ranging from biology 
(analysis of gene expression data \cite{fujita2012functional}), 
to economics (time series that capture macroeconomic variables \cite{focardi2001clustering}). 
A particularly widely used application originates from finance, where one aims to cluster  
the large number of time series corresponding to  financial instruments in the stock market \cite{ziegler2010visual, pavlidis2006financial}.
In all the above application scenarios, a popular approach in the literature is to consider the similarity measure given by the 
Pearson correlation coefficient (and related distances),  
that measures the  linear dependence between pairs of variables and takes values in $[-1,1]$.
Furthermore, the resulting empirical correlation matrix can be construed as a weighted network, where the signed edge weights capture pairwise correlations, which reduces the task of clustering the multivariate time series to that of clustering the underlying signed network.
For scalability and robustness reasons, one can further consider
thresholding on the p-value associated to each  individual sample correlation entry  \cite{Ha2015_Network_Threshold_pValue},  following tests of statistical significance, which renders the network sparse, and thus amenable to larger scale computations.
We refer the reader to the seminal work of Smith et al. \cite{Smith_2011_FMRI_networks_TimeSeries}  for a detailed survey and comparison of various methodologies for casting time series data into networks, in the context of fMRI time series. Their overall conclusion that  
correlation-based approaches are typically quite successful at estimating the connectivity of brain networks from fMRI time series data, provides sound motivation for the development of algorithms for analysis of signed networks and correlation clustering.

\textbf{Contribution}. Our main contribution is to extend the Ginzburg--Landau functional to the setting of clustering signed networks, %
and show the applicability of the MBO scheme in this context, for different types of constraints that arise in the semi-supervised setting.  
As a side contribution, we provide a more computationally efficient way to determine node-cluster membership (as an alternative to the traditional projection/thresholding step).
In addition, we show extensive numerical experiments on both synthetic and real data, showcasing the competitivity  of our approach when compared to state-of-the-art methods, across various noise and sparsity regimes.

\textbf{Paper outline}. The remainder of this paper is organized as follows. 
Section \ref{sec:related_lit} is a summary of related work from the signed clustering literature. 
Section \ref{sec:Ginzburg_Landau_overview} gives a brief overview of the  Ginzburg--Landau functional. 
Section \ref{sec:signedGL} extends the Ginzburg--Landau framework to the setting of signed graphs, by explicitly handling negative weights and modifying the MBO algorithm accordingly. 
Section \ref{sc:GLFwithcons} pertains to the semi-supervised learning setting, and 
incorporates various types of a priori  information available to the user.
Section \ref{sc:thealgo} summarizes the structure of the algorithm we propose. 
Section \ref{sec:numericalSection} contains numerical experiments under a signed stochastic block model, and includes various scenarios in the semi-supervised setting.
Section  \ref{sec:expBAmodel}  details the outcome of numerical experiments under the  Barab\'asi--Albert model. 
Section  \ref{sec:imageSeg} considers an application to image segmentation, Section \ref{sec:realData} to clustering correlation matrices arising from financial time-series data, and Section \ref{sec:migration} to computational social science using a network resulting from migration patterns between US counties.
Section \ref{sec:conclusion}  summarizes our findings  and details future research directions.
Finally, we defer to the Appendix the technical details that pertain to the projection/thresholding step.

\textbf{Notation}. We summarize our notation in Table \ref{tab:methodAbbrev}.
\begin{table} 
\begin{minipage}[b]{0.49\linewidth}
\begin{center}
\begin{tabular}{|p{1cm}|p{6cm}|}
\hline
Symbol & Description   \\
\hline
$\alpha$  &  percentage of known information\\
$A$  &  affinity matrix of the signed case\\
$A^+$ & positive part of $A$\\
$A^-$ & negative part of $A$\\
$B$  &  time evolution matrix in the MBO scheme with affinity matrix $N$\\
$\bar{B}$  &  time evolution matrix in the MBO scheme with affinity matrix $A$\\
$c_i$ & size of clusters of the Signed Stochastic Block Model\\
$d\tau$  &  time discretization parameter\\
$D$  &  degree matrix of $N$\\
$\bar{D}$  &  degree matrix of $A$\\
$\underbar{\ensuremath{e}}$  &  element of the vertices of the probability simplex\\
$\underbar{\ensuremath{e}}^\pm$  &  element of the vertices of $\Sigma_K^\pm$\\
$\varepsilon$  &  parameter of the GL functional \\
$F_\varepsilon$  &  GL functional\\
$\bar{F}_\varepsilon$  & signed GL functional\\
$I_V$ & Identity matrix with size $V$\\
$K$  &  number of clusters\\
$\Lambda$  &  eigenvalue matrix\\
$L$  &  Laplacian of $N$\\
$L^+$ & Laplacian of $A^+$\\
$L_{\text{rw}}$  &  random walk Laplacian of $N$\\
$L_{\text{sym}}$  &  symmetric Laplacian of $N$\\
$\bar{L}$  &  signed Laplacian\\
$\bar{L}_{\text{rw}}$  &  signed random walk Laplacian  \\
$\bar{L}_{\text{sym}}$  &  signed symmetric Laplacian\\
$\lambda$ & sparsity probability\\
$\lambda^+$, $\lambda^-$ & trade-off parameters for must- and cannot-links, respectively\\
$m$ & number of eigenvector considered in the expansion of $\bar{B}$\\
$n$ & step number in the MBO algorithm\\
$N$  &  affinity matrix of the positive edges case\\
$N_\tau$  &  number of repetitions of the diffusion step\\
$\mathbb{N}_+$  &  natural numbers excluding $0$.\\
$\eta$ & noise probability\\
\hline
\end{tabular}
\end{center}
\end{minipage}
\begin{minipage}[b]{0.49\linewidth}
\begin{center}
\begin{tabular}{|p{1cm}|p{6cm}|}
\hline
Symbol & Description   \\
\hline
$Q^-$ & signless Laplacian\\
$\rho$  & system density 
for one-dimensional continuous version of the GL functional \\
$\underbar{\ensuremath{\rho}}$  &  vertex set values of the discrete bi-dimensional case\\
$R^{fi}$, $R^{av}$ & weight matrix for the fidelity and avoidance contributions, respectively\\
$\mathbb{R}$  &  real numbers \\
$\mathbb{R}_+$  &  positive real numbers excluding $0$\\
$\mathbb{R}_{+0}$  &  positive real numbers including $0$\\
$S$  &  ground truth affinity matrix for the Stochastic Signed Block Model\\
$\sigma$  &  component of the probability simplex\\
$\sigma^\pm$  &  component of the $\Sigma_K^\pm$ simplex\\
$\Sigma_K$  &  probability simplex\\
$\Sigma_K^{\pm}$  &  `$2-K$' simplex (see Equation \eqref{eq:2mKsimplex})\\
$t$  &  time variable\\
$\underbar{\ensuremath{u}}$  &  column vector carrying the weight of a node on various clusters\\
$\underbar{\ensuremath{u}}^\pm$  &  column vector carrying the weight of a node on various clusters in the signed case\\
$U$  &  node-cluster association matrix\\
$U^\pm$  &  node-cluster association matrix in the signed case\\
$V$  &  number of nodes\\
$x$  &  space coordinate\\
$X$  &  eigenvectors matrix\\
$\underbar{\ensuremath{0}}_K$  &  column vector of length $K$ whose elements are all $0$\\
$0_{V\times K}$  &  matrix with $V$ rows and $K$ columns whose elements are all $0$\\
$\underbar{\ensuremath{1}}_V$  &  column vector of length $V$ whose elements are all $1$\\
$1_V$  &  square matrix of size $V$ whose elements are all $1$\\
$\circ$ & Hadamard product\\
$\left|\cdot\right|$ & absolute value\\
$\left\Vert \cdot\right\Vert _{1}$ & taxicab metric\\
$\left\Vert \cdot\right\Vert _{2}$ & Euclidean distance\\
\hline
\end{tabular}
\end{center}
\end{minipage}
\vspace{-2mm}
\caption{Summary of frequently used symbols in the paper from Section \ref{sec:Ginzburg_Landau_overview} onwards.}
\label{tab:methodAbbrev}
\end{table}

\section{Related literature}  \label{sec:related_lit}

The problem of clustering signed graphs traces its roots to the work of Cartwright and Harary from the 1950s from  social balance theory \cite{HararySigned,cartwright1956structural}. 
Consider, for example, three mutually connected nodes from a signed graph, and  define the relationship between two nodes to be \text{\it positive} if the edge in-between is $+1$ (they are ``friends''), and \text{\it negative} if the edge in-between is $-1$ (they are ``enemies''). We say that such a triad is \text{\it balanced} if the product of its signs is positive, and that the signed graph is balanced if all of its cycles are positive \cite{cartwright1956structural}, i.e., each cycle contains an even number of negative edges. 
Cartwright and Harary proved the so-called \textit{structure theorem} that if a signed graph  $G=(V,E)$ is \textit{balanced}, then the set of vertices, $V$ can be partitioned into two subsets, called plus-sets, such that all of the positive edges are between vertices within a plus-set and all negative edges are between vertices in different plus-sets.

Davis \cite{DavisWeakBalance} proposed a weaker notion of balance, \textit{weak balance theory},  which relaxed the balanced relationship by allowing an enemy of one's enemy to also be an enemy, and showed that a network with both positive and negative edges is k-weakly balanced if and only if its nodes can be clustered into $k$ groups such that edges within groups are positive and edges between groups are negative. 
For a $k$-balanced network, the problem of finding the partition of the node set can be viewed as an optimization problem: find the $k$-partition $C$ which minimizes $P_\alpha: \Phi \rightarrow \mathbb{R}$ given by
\begin{equation}
P_\alpha(C) := \alpha{\sum}_n+(1-\alpha){\sum}_p,
\end{equation}
over $\Phi$, the set of all possible $k$-partitions. Here $\sum_p$ is the number of positive arcs between plus-sets, $\sum_n$ is the number of negative arcs within plus-sets, and $\alpha \in \left(0,1\right)$ is a weighting constant  \cite{doreian1996partitioning}. If the graph is $k$-balanced the set of minimizers of $P$ will not change when $\alpha$ varies; if the graph is not $k$-balanced, the parameter $\alpha$ allows control over which kind of error is penalized more strongly, positive ties between groups or negative ties within groups.

Motivated by this characterization, the $k$-way clustering  problem in signed networks amounts to finding a partition into $k$ clusters such that most edges within clusters are positive and most edges across clusters are negative. As an alternative formulation, one may seek a partition such that the number of \textbf{violations} is minimized, i.e., negative edges within the cluster and positive edges across clusters.  In order to avoid partitions where clusters contain only a few nodes, one often prefers to incentivize  clusters of similar large size or volume, leading to balancing factors such as the denominators $x_c^T x_c$ in \eqref{obj_BRatioC} or $x_c^T \bar{D} x_c$ in \eqref{obj_BNormC},  discussed further below.

Weak balance theory has motivated the development of algorithms for clustering signed networks. Doreian and Mrvar \cite{DoreianMrvar20091} proposed a local search approach in the  spirit of the  Kernighan--Lin algorithm \cite{KernighanLin}. Yang et al. \cite{YangCheungLiu} introduced an agent-based approach by  considering a certain random walk on the graph.  Anchuri et al. \cite{anchuri2012communities}  
proposed a spectral approach to optimize modularity and other objective functions in signed networks, and demonstrated its  effectiveness in detecting communities from real world networks.
In a different line of work, known as  \textit{correlation clustering}, Bansal et al. \cite{BansalBlumCorrClust} considered the  problem of clustering signed complete graphs, proved that it is NP-complete, and proposed two approximation algorithms with theoretical guarantees on their performance.  On a related note, Demaine and Immorlica \cite{DemaineImmorlicaCorrClust} studied the same problem but for arbitrary weighted graphs, and proposed an O($\log n$) approximation algorithm based on linear programming. For correlation clustering, in contrast to k-way clustering, the number of clusters is not given in advance, and there is no normalization with respect to size or volume.

Kunegis et al. \cite{kunegis2010spectral} proposed new  spectral tools for clustering, link prediction, and visualization of signed graphs, 
by solving a signed version of the 2-way ratio-cut problem   via the signed graph Laplacian \cite{HouSignedLap}
\begin{equation}
\bar{L} := \bar{D} - A, 
\label{signed_Lap} 
\end{equation}
where $\bar{D}$ is the diagonal matrix with $\bar{D}_{ii} = \sum_{i=1}^{n} |d_{ij}|$. 
The authors proposed a unified approach 
to handle graphs with both positive and negative edge weights, motivated by a number of tools and applications, including random walks, graph clustering, 
graph visualization and electrical networks. The signed extensions of the combinatorial Laplacian  also exist for the random-walk normalized Laplacian and the symmetric graph Laplacian
\begin{equation}
\bar{L}_{\text{rw}} := I - \bar{D}^{-1} A,
\label{signed_Lap_rw}
\end{equation}
\vspace{-4mm}
\begin{equation}
\bar{L}_{\text{sym}} := I - \bar{D}^{-1/2} A \bar{D}^{-1/2},
\label{signed_Lap_sym} \\
\end{equation}
the latter of which is particularly suitable for skewed degree distributions. Here $I$ denotes the identity matrix.
Recent work, by a subset of the authors in our present paper, provided the first theoretical guarantees, under a signed stochastic block model, for the above Signed Laplacian $\bar{L}$   \cite{SPONGE}.

In related work,  Dhillon et al. \cite{DhillonBalNormCut} proposed a formulation based on the  \textit{Balanced Ratio Cut} objective
\begin{equation}
\operatorname{min}_{ \{x_1,\ldots,x_k\} \in \mathcal{I}_k } \left(  \sum_{c=1}^{k} \frac{x_c^T(D^{+} - A)x_c }{x_c^T x_c}  \right),
\label{obj_BRatioC}   
\end{equation}
and the closely related \textit{Balanced Normalized Cut} (BNC) objective for which the balance factors $ x_c^T x_c$ in the denominators 
are replaced by $x_c^T \bar{D} x_c$, 
\begin{equation}
\operatorname{min}_{ \{x_1,\ldots,x_k\} \in \mathcal{I}_k} \left(  \sum_{c=1}^{k} \frac{x_c^T(D^{+} - A)x_c }{x_c^T \bar{D} x_c}  \right).
\label{obj_BNormC}  
\end{equation}
Here $\mathcal{I}_k$ denotes the collection of all sets containing the indicator vectors of a $k$-partition of the node set $V$, i.e. if $\{x_1,\ldots,x_k\} \in \mathcal{I}_k$, then each of the $k$ vectors has entries
\begin{equation}
x_t(i) := \left\{
 \begin{array}{rl}
 1, & \; \text{if node } i \in C_t,\\
 0, & \; \text{otherwise },	\\
     \end{array}
   \right.
\label{def:indicatorFcn}
\end{equation} 
for a collection of disjoint clusters $C_t$ which partitions the node set $V$. 
The same authors  
claimed  that the approach  of Kunegis et al. \cite{kunegis2010spectral} via the Signed Laplacian faces a fundamental weakness when  directly  extending it to $k$-way clustering, and obtained better results on several real data sets via their proposed algorithm  which optimizes the BNC objective function. On the other hand, \cite{SPONGE} reported good numerical performance of the  Signed Laplacian even for higher values of $k$, aside from providing a robustness analysis of it for $k=2$.  

A rather different approach for signed graph clustering has been investigated by Hsieh et al. \cite{DhillonLowRank}, who proposed a low-rank model, based on the observation that
the problem of clustering a signed network can be interpreted as recovering the missing structure of the network. 
The missing relationships are recovered via state-of-the-art matrix completion algorithms (MC-SVP),  which first  complete the network and then recover the clusterings via the k-means algorithm using the top $k$ eigenvectors of the completed matrix. An alternative method (MC-MF),  scalable to very large signed networks, is based on a low-rank matrix factorization approach, which first completes the network using matrix factorization and derives the two low-rank factors  $U,H \in \mathbb{R}^{n \times k}$, runs k-means on both $U$ and $V$, and selects the clustering which achieves a lower BNC objective function \eqref{obj_BNormC}.

Both strong and weak balance are defined in terms of local structure at the level of triangles. As observed early in the social balance theory literature, local patterns between the nodes of a signed graph induce certain global characteristics. Chiang et al. \cite{DhillonLocalGlobal} exploited both local and global aspects of social balance theory, and proposed algorithms for sign prediction and clustering of signed networks. They defined more general measures of social imbalance based on cycles in the graph, and showed that the classic Katz measure, ubiquitous for unsigned link prediction, also has a balance theoretic interpretation in the context of signed networks. The global structure of balanced networks leads to a low-rank matrix completion approach 
via convex relaxation, which scales well to very large problems. This approach is made possible by the fact that the adjacency matrix of a complete k-weakly balanced network has rank 1 if $k \leq 2$, and rank $k$ for all $k>2$ \cite{DhillonLocalGlobal}.

The recent work of Mercado et al. \cite{Mercado_2016_Geometric} presented an extended spectral method based on the geometric mean of Laplacians.  In another recent work, a subset of the authors showed in \cite{sync_congress} that, for $k=2$, signed clustering can be formulated as an instance of the group synchronization problem over $\mathbb{Z}_2$, for which spectral, semidefinite programming relaxations, and message passing algorithms have been studied; extensions to the multiplex setting and incorporation of constraints have also been considered. 
Finally, we refer the reader to \cite{JeanGallierSurvey,fox2018investigation} for recent surveys on clustering 
signed and unsigned graphs.

%
%
\section{Ginzburg--Landau functional}  \label{sec:Ginzburg_Landau_overview}

The \textit{Ginzburg--Landau functional} (GLF) was originally introduced in the materials science literature to study phase separation, and has since also found many uses in image processing and image analysis due to its connections with the total variation functional \cite{ModicaMortola77,Modica87}. In this section, we briefly review the main idea behind the GLF together with a description of the (local) minimizer(s) of the functional in terms of solutions of the \textit{Allen--Cahn equation} (ACE). Then, we will show how this approach can be adapted to the discrete case of clustering over signed graphs.

\subsection{GLF as potential energy}
In this introductory presentation, which is based on \cite{weizhu2007dynamics, bartels2015numerical}, we restrict to the one dimensional case for simplicity. Consider the function $\rho$ in the Sobolev space $\mathscr{W}^{1,2}\left(\left[0,1\right]\right)$. The GLF $F_\varepsilon:\mathscr{W}^{1,2}\left(\left[0,1\right]\right)\mapsto \left[0,\infty\right)$ is defined as
\begin{equation}\label{eq:GLfunctional}
F_{\varepsilon}\left[\rho\right]\coloneqq\frac{\varepsilon}{2}\int_{0}^{1}\left(\frac{d\rho}{ds}\right)^{2}ds+\frac{1}{4\varepsilon}\int_{0}^{1}\left(1-\rho^{2}\right)^{2}ds ,
\end{equation}
for some $\varepsilon>0$. The second term has the shape of a double-well and is often labelled $W\left(\rho\right)\coloneqq\left(1-\rho^{2}\right)^{2}\in\mathbb{R}_+$. The GLF is interpreted as the \textit{potential energy} for the physical observable $\rho$. At this stage, $\rho$ is also understood as time dependent and we want to find solutions of
\begin{equation}\label{eq:energypot}
\frac{d\rho}{dt}\left(t,x\right)=-\frac{\delta F_\varepsilon}{\delta \rho}\left[\rho\right]\left(t\right),
\end{equation}
where $\frac{\delta F}{\delta \rho}$ is the functional derivative of $F$, which lead in their time evolution to the local minima of the GLF. A common choice for the boundary values of the problem \eqref{eq:energypot} is the \textit{von Neumann condition}
\begin{equation}
\frac{d\rho}{dx}\left(t,1\right)=\frac{d\rho}{dx}\left(t,0\right)=0\quad\text{for}\quad{t>0}.
\end{equation}
In this way, calculating the functional derivative of the right-hand side in \eqref{eq:energypot} gives rise to the ACE
\begin{equation}\label{eq:ACequation}
\frac{d\rho}{dt}\left(t,x\right)=\varepsilon\frac{d^{2}\rho}{dx^{2}}\left(t,x\right)-\frac{1}{\varepsilon}\rho\left(t,x\right)\left(\rho^{2}\left(t,x\right)-1\right),
\end{equation}
where the last term is often called $W'\left(\rho\right)\coloneqq \rho\left(\rho^{2}-1\right)\in\mathbb{R}$. Moreover, it is possible to verify by direct calculation that the GLF is decreasing in time on the solution of the ACE, that is
\begin{equation}
\frac{dF_\varepsilon\left[\rho\right]}{dt}\left(t\right)=-\int_{0}^{1}\left(\frac{d\rho}{dt}\left(t,s\right)\right)^{2}ds<0,
\end{equation}
as is expected for the evolution of a physical observable along its energy potential. The meaning of the coefficient $\varepsilon$ can be understood as follows. A stationary solution for the ACE, such that $\frac{d\rho}{dt}\left(t,x\right)=0\,\forall t\in\mathbb{R}_+$ in \eqref{eq:ACequation}, is given by
\begin{equation}\label{eq:solutionACstat}
\rho^{*}_\varepsilon\left(x\right)\coloneqq\tanh\left(\frac{x}{\varepsilon\sqrt{2}}\right).
\end{equation}
By taking the first derivative with respect to the spatial coordinate, we can see that the parameter $\varepsilon$ is responsible for the sharpness of the solution (see Figure \ref{fig:solutionACstat}). Thus, given the values of the function $\rho^{*}_{\varepsilon}$ at the boundary of its domain, the ACE can be understood as the evolution of the layer profile connecting the two boundary values of the system.
\begin{figure}
    \centering
    \includegraphics[width=0.42\textwidth]{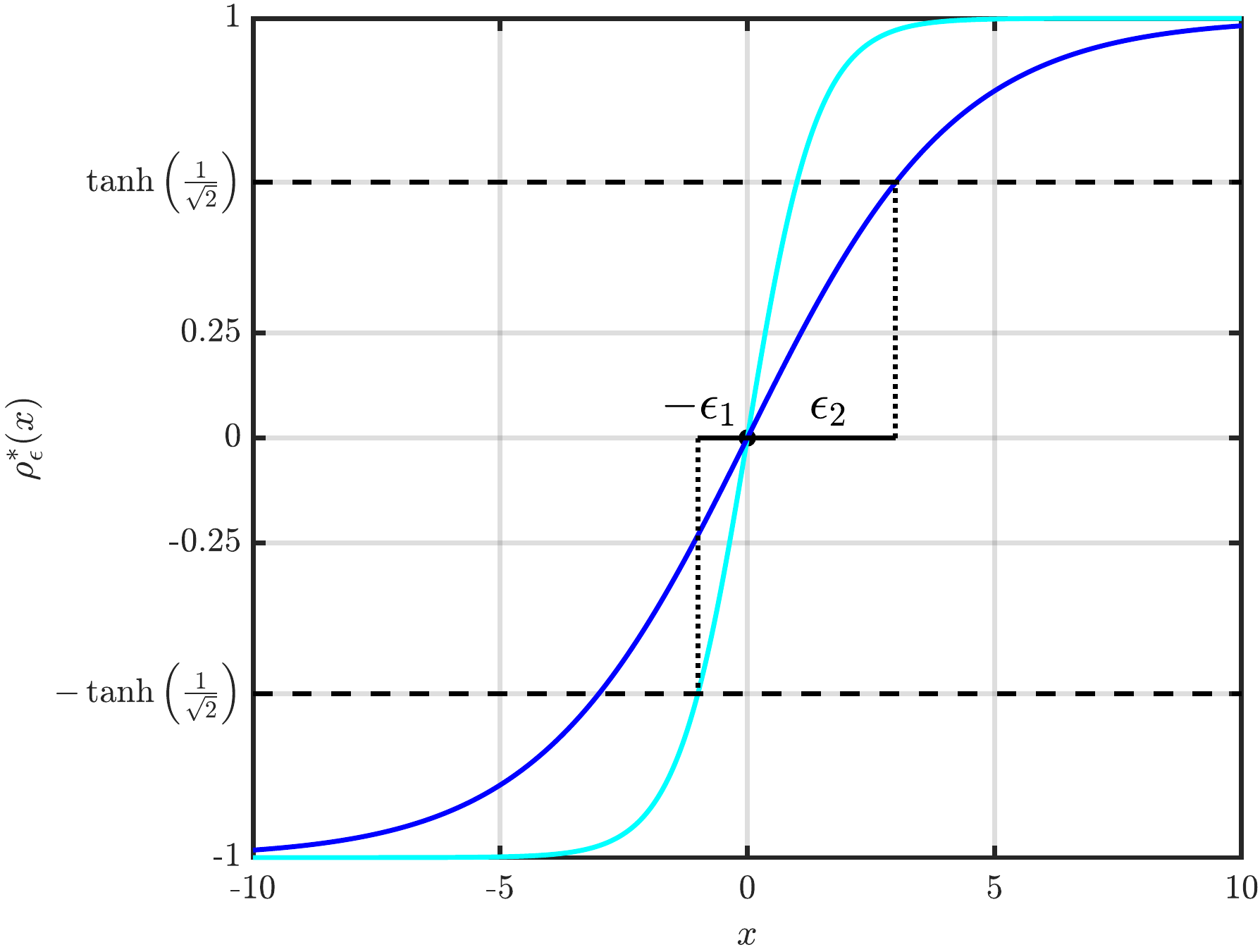}
    \caption{Plot of the stationary solution \eqref{eq:solutionACstat} of the ACE \eqref{eq:ACequation} for two parameter values $\varepsilon_2$ (blue curve) and $\varepsilon_1$ (cyan curve) with $\varepsilon_2>\varepsilon_1$. We can see that bigger values of $\varepsilon$ reduce sharpness of the curve $\rho^*_\varepsilon\left(x\right)$. Moreover, $\rho_{\varepsilon_{1}}^{*}\left(\varepsilon_{1}\right)=\rho_{\varepsilon_{2}}^{*}\left(\varepsilon_{2}\right)=\tanh\left(\frac{1}{\sqrt{2}}\right)\approx0.609$.}
    \label{fig:solutionACstat}
\end{figure}
\subsection{GLF and clustering}
Using tools from nonlocal calculus, see Section 2.3 in \cite{merkurjev2013mbo}, it is possible to write a version of the GLF where the density function is defined on a network. To fix the notation, we denote as $G\left(V,E,N\right)$ a positive-weighted undirected graph with \textit{vertex set} $V$, \textit{edge set} $E$ and \textit{affinity matrix} $N$ which is a symmetric matrix whose entries $N_{ij}$ are zero if $i=j$ or $i\not\sim j$\footnote{The symbol $i\sim j$, respectively  $i\not\sim j$, denotes that nodes $i$ and $j$ are connected, respectively not connected, by an edge.} and nonnegative otherwise.
From the affinity matrix, ones defines the diagonal  \textit{degree matrix} as
\begin{equation}\label{eq:degreeMatdef}
D_{ij}\coloneqq\left\{ \begin{array}{cc}
\begin{array}{c}
\sum_{k=1}^{V}N_{ik}\in\mathbb{R}_{+0}\\
0
\end{array} & \begin{array}{c}
\text{if }i=j\\
\text{if }i\neq j
\end{array}\end{array}\right. ,
\end{equation}
where $V$ is the cardinality of the vertex set. When the entries of the affinity matrix are positive, the associated network is often called an \textit{unsigned graph}. In the same way, the discrete equivalent of the density function is now a column vector on the vertex set, that is $\underbar{\ensuremath{\rho}}:\mathbb{N}\times\left[0,T\right]\mapsto\mathbb{R}^{V}$, written as
\begin{equation}
\underbar{\ensuremath{\rho}}\left(t\right)\coloneqq\left\{ \rho_{0}(t),\rho_{1}(t),\dotsc,\rho_{V}(t)\right\} .
\end{equation}
Following \cite{bertozzi2016diffuse}, it turns out that the network version of \eqref{eq:GLfunctional} for unsigned graphs, also called \textit{graph GLF} (GGLF), is given by
\begin{equation}\label{eq:netGLF}
F_{\varepsilon}(\underbar{\ensuremath{\rho}})\coloneqq\frac{\varepsilon}{2}\langle\underbar{\ensuremath{\rho}},L\underbar{\ensuremath{\rho}}\rangle+\frac{1}{4\varepsilon}\sum_{i=1}^{V}(\rho_{i}^{2}-1)^{2},
\end{equation}
where $L\coloneqq D-N\in\mathbb{R}^{V\times V}$ is the (combinatorial) graph Laplacian,  and $\langle \cdot, \cdot \rangle$ denotes the usual inner product on $\R^{V}$. It is important to notice that the following expression holds true  
\begin{equation}\label{eq:identityLap}
\langle\underbar{\ensuremath{\rho}},L\underbar{\ensuremath{\rho}}\rangle=\sum_{i=1}^{V}N_{ij}(\rho_{i}-\rho_{j})^{2} \geq 0.
\end{equation}
The inequality in \eqref{eq:identityLap} shows that $L$ is positive semidefinite (i.e., $L$ has non-negative, real-valued eigenvalues, including $0$) and this assures the non-negativity of the GGLF \eqref{eq:netGLF}. Again, the minimization of the GGLF can shed light on the interpretation of the terms composing the potential. The second term in $F_\e$ will force minimizers $\rho\in \R^{V}$ to have entries which are close to $-1$ or $1$ in order to keep $(\rho_i^2-1)^2$ small. The first term in $F_\varepsilon$ penalizes the assignment to different clusters of nodes that are connected by highly weighted edges. This means that minimizers of $F_\varepsilon$ are approximate indicator vectors of clusters that have few highly weighted edges between them. In the end, we interpret nodes having different sign as belonging to different clusters.

As before, minimizing the GGLF using the method of the gradient descent, we obtain the ACE for networks which we write component-wise as
\begin{equation}\label{eq:ACEnetwork}
\frac{d\rho_{i}}{dt}\left(t\right)=-\varepsilon\left[L\underbar{\ensuremath{\rho}}\right]_{i}\left(t\right)-\frac{1}{\varepsilon}\rho_{i}\left(t\right)\left(\rho_{i}^{2}\left(t\right)-1\right) ,
\end{equation}
where $\left[L\underbar{\ensuremath{\rho}}\right]_{i}$ indicates the $i$-th component of the vector $L\underbar{\ensuremath{\rho}}$.
Instead of the combinatorial graph Laplacian $L$, it is possible to use the random walk Laplacian $L_{\text{rw}} := I - D^{-1} N$ (which is positive semidefinite), or the symmetric graph Laplacian $L_{\text{sym}} := I - D^{-1/2} N D^{-1/2}$ (symmetric positive semidefinite). In fact, \cite{BertozziFlenner12} argues for the use of $L_{\text{sym}}$, but no rigorous results are known that confirm this is necessarily the best choice for each application.


\subsection{Numerical solution of the ACE}\label{sc:numsolACE}

Various methods have been employed to numerically compute (approximate) minimizers of $F_\e$. In \cite{BertozziFlenner12} a gradient flow approach is used and \eqref{eq:ACEnetwork}, interpreted as a system of ordinary differential equations (ODEs), is solved. The hope is that the gradient flow dynamics will converge to a local minimizer of $F_\e$. Under some assumptions \cite{absil2006stable}, this hope is theoretically justified \cite{Luo2017}, but numerically the nonconvexity of $F_\e$ can cause some issues. In \cite{BertozziFlenner12} such issues were tackled with a convex splitting scheme. 

Another strategy which has been developed to approximately solve the ACE for networks consists in using a variation of the \textit{Merriman--Bence--Osher} (MBO) scheme (also known as \textit{threshold dynamics}). This was first introduced in \cite{MBO1992,MBO1993} in a continuum setting as a proposed algorithm to approximate flow by mean curvature. Following \cite{merkurjev2013mbo}, we first neglect the terms $W'\left(\rho_i\right)$ and discretize \eqref{eq:ACEnetwork} in the following way
\begin{equation}\label{eq:discrMBO}
\frac{\underbar{\ensuremath{\rho}}^{n+\frac{1}{2}}-\underbar{\ensuremath{\rho}}^{n}}{d\tau}=-\varepsilon L\underbar{\ensuremath{\rho}}^{n+\frac{1}{2}} ,
\end{equation}
where $d\tau\in\mathbb{R}_+$ is the length of a small time interval to be tuned and the superscript $n$ indicates the step of the iterative procedure ($n+\frac{1}{2}$ simply specifies an intermediate step between the $n$-th step and the subsequent one which will be labelled as $n+1$). The above leads to the \textit{MBO equation} (MBOE)
\begin{equation} \label{eq:MBOE1d}
\underbar{\ensuremath{\rho}}^{n+\frac{1}{2}}=\left(I_{V}+\varepsilon d\tau L\right)^{-1}\underbar{\ensuremath{\rho}}^{n} ,
\end{equation}
where $I_{V}\in\left\{ 0,1\right\} ^{V\times V}$ is the identity matrix. After performing the eigendecomposition of the Laplacian, \eqref{eq:MBOE1d} becomes
\begin{equation}
\underbar{\ensuremath{\rho}}^{n+\frac{1}{2}}=\left[X\left(I_{V}+d\tau\varepsilon\Lambda\right)X^{T}\right]^{-1}\underbar{\ensuremath{\rho}}^{n} ,
\end{equation}
where $X$ is a matrix whose columns are the normalized eigenvectors of $L$ and $\Lambda$ is the relevant eigenvalue diagonal matrix. Since the double-well potential pushes the components of $\underbar{\ensuremath{\rho}}$ to be either $+1$ or $-1$, its role can be approximated via a thresholding step, in the following way
\begin{equation} \label{eq:MBOEthres1d}
\rho_{i}^{n+1}=T\left(\rho_{i}^{n+\frac{1}{2}}\right)\coloneqq\left\{ \begin{array}{c}
\begin{array}{cc}
+1 & \text{if }\rho_{i}^{n+\frac{1}{2}}\geqslant0\end{array}\\
\begin{array}{cc}
-1 & \text{if }\rho_{i}^{n+\frac{1}{2}}<0\end{array}
\end{array}\right. .
\end{equation}
This procedure can be iterated until a convergence criterion is met. In practice, the diffusion step \eqref{eq:MBOE1d} is repeated $N_\tau\in\mathbb{N}_+$ times before the thresholding step \eqref{eq:MBOEthres1d}. This corresponds to multiplying $\underbar{\ensuremath{\rho}}^{n}$ in \eqref{eq:MBOE1d} by $\left[X\left(I_{V}+\frac{d\tau}{N_{\tau}}\varepsilon\Lambda\right)X^{T}\right]^{-N_{\tau}}$ (note the division of $d\tau$ by $N_\tau$).

The graph version of MBO has been used successfully in recent years in various applications, both applied to $F_\e$ \cite{merkurjev2013mbo}, its multi-class extension \cite{merkurjev2014graph}, and its signless Laplacian variant \cite{KeetchvanGennip17}. In general, we notice that the choice of $\e$ also poses a problem. Smaller values of $\e$ are expected to lead to minimizers $\underbar{\ensuremath{\rho}}$ that are closer to being indicator vectors, however the discrete length scales of the problem impose a lower bound on the $\e$ values that can be resolved. The MBO scheme avoids this problem since the $W'\left(\rho_i\right)$ are substituted with a thresholding step. In this way, we can fix $\varepsilon=1$. Using the MBO scheme instead of the gradient flow does not solve all problems: It eliminates the terms in the ODEs that originated with the non-convex term in $F_\e$ and it removes $\e$ from consideration, but instead a value for $d\tau$ has to be chosen. The discrete length scales in the problem disallow values that are too small, yet $d\tau$ has to be small enough to prevent oversmoothing by the diffusion term $\left[L\underbar{\ensuremath{\rho}}\right]_{i}$.

\subsection{Generalization to multiple clusters}\label{sc:genmulcl}

The Ginzburg--Landau method has been adapted to allow for clustering or classification into $K>2$ clusters \cite{Garcia-CardonaFlennerPercus13,MerkurjevGarcia-CardonaBertozziFlennerPercus14,Garcia-CardonaMerkurjevBertozziFlennerPercus14,Garcia-CardonaFlennerPercus15}. We briefly review this line of work here, since our proposed approach for clustering signed graphs is for general $K$ as well.
This construction relies upon an \textit{assignment} or \textit{characteristic matrix} whose entries correspond to the weight for a node $i$ to belong to a cluster $k$. More precisely, consider the simplex
\begin{equation}
\Sigma_{K}\coloneqq\left\{ \left(\begin{array}{cccc}
\sigma_{1}, & \sigma_{2}, & \dotsc, & \sigma_{K}\end{array}\right)\in\left[0,1\right]^{K}:\sum_{k=1}^{K}\sigma_{k}=1\right\} ,
\end{equation}
which corresponds to row vectors whose components sum up to $1$. Then, we can associate to each node a \textit{characteristic vector}\footnote{We generally refer to the vector whose entries are the ``intensity" of belongingness to a cluster as \textit{characteristic vector} also if these coefficients are not only $0$ or $1$.}, written as a column vector
\begin{equation}
\underbar{\ensuremath{u}}_{i}\coloneqq\left(\begin{array}{cccc}
u_{i1}, & u_{i2}, & \dotsc, & u_{iK}\end{array}\right)^T\in\Sigma_{K} ,
\end{equation}
which specifies the weight for the node $i$ to belong to each of the clusters in the network. So, the characteristic matrix of all the nodes in the system is defined as
\begin{equation}\label{eq:origcharU}
U\coloneqq\left(\begin{array}{cccc}
\underbar{\ensuremath{u}}_{1}, & \underbar{\ensuremath{u}}_{2}, & \dotsc, & \underbar{\ensuremath{u}}_{V}\end{array}\right)^{T}\in\left(\Sigma_{K}\right)^{V} .
\end{equation}
In the end, as explained in \cite{Garcia-CardonaMerkurjevBertozziFlennerPercus14}, the generalization of the GLF reads as 
\begin{equation}\label{eq:multiKGLF}
F_{\varepsilon}(U)\coloneqq\frac{\varepsilon}{2}\langle U,LU\rangle+\frac{1}{2\varepsilon}\sum_{i=1}^{V}\left(\prod_{k=1}^{K}\frac{1}{4}\left\Vert \underbar{\ensuremath{u}}_{i}-\underbar{\ensuremath{e}}_{k}\right\Vert _{1}^{2}\right) ,
\end{equation}
where $\left\Vert \cdot\right\Vert _{1}$ is the taxicab metric, $\langle U,LU\rangle=\text{Tr}\left(U^{T}LU\right)$ and $\underbar{\ensuremath{e}}_{k}$ correspond to the canonical basis vector defined as
\begin{equation}
\underbar{\ensuremath{e}}_{k}\coloneqq\left(\begin{array}{cccc}
e_{1}, & e_{2}, & \dotsc, & e_{K}\end{array}\right)^T\in\left\{ 0,1\right\} ^{K}, 
\end{equation}
whose components satisfy $e_{i=k}=1$ and $e_{i\neq k}=0$. In this context, the MBO scheme for the characteristic matrix reads as
\begin{equation}\label{eq:intermediateMBO}
U^{n+\frac{1}{2}}=BU^{n},\quad\text{where}\quad B\coloneqq\left[X\left(I_{V}+\frac{d\tau}{N_\tau}\Lambda\right)X^{T}\right]^{-N_\tau}\in\mathbb{R}^{V\times V}
\end{equation}
where we have fixed $\varepsilon=1$, as explained at the end of Section \ref{sc:numsolACE}, and we have included the contribution from the diffusion iteration directly in $B$ as explained below equation \eqref{eq:MBOEthres1d}. Then, following \cite{Garcia-CardonaMerkurjevBertozziFlennerPercus14}, we project each row of $U^{n+\frac{1}{2}}$ on the simplex $\Sigma_K$. We call the corresponding vector 
\begin{equation}\label{eq:projOLD}
\underbar{\ensuremath{v}}_{i}\coloneqq\argmin_{\underbar{\ensuremath{y}}\in\Sigma_K}\left\Vert \underbar{\ensuremath{u}}_{i}^{n+\frac{1}{2}}  -\underbar{\ensuremath{y}} \right\Vert _{2}\in\left[0,1\right]^{K}.
\end{equation}
It is important to notice that in the intermediate step $\underbar{\ensuremath{u}}_{i}^{n+\frac{1}{2}}\in\mathbb{R}^{K}$. The thresholding step consists in determining the closest vertex to $\underbar{\ensuremath{v}}_{i}$ of $\Sigma_K$, that is\footnote{When $\underbar{\ensuremath{v}}_{i}$ has the same distance to multiple vertices, one of those is picked at random with uniform probability.}
\begin{equation}\label{eq:threOLD}
\underbar{\ensuremath{u}}_{i}^{n+1}=\underbar{\ensuremath{e}}_{k}\quad\text{with}\quad k=\argmin_{j=1,\dotsc,K}\left\Vert \underbar{\ensuremath{v}}_{i}  -\underbar{\ensuremath{e}}_j \right\Vert _{2}.
\end{equation}
In the end, we have $U^{n+1}=\left(\begin{array}{cccc}
\underbar{\ensuremath{u}}_{1}^{n+1}, & \underbar{\ensuremath{u}}_{2}^{n+1}, & \dotsc, & \underbar{\ensuremath{u}}_{V}^{n+1}\end{array}\right)^{T}$. In Section \ref{sc:shortMBO}, we will use a simplified version of the MBO scheme which reduces the projection \eqref{eq:projOLD} and thresholding \eqref{eq:threOLD} to a single step (the equivalence of these two versions will be proved in the Appendix \ref{app:projsimplex}).

%
%

\section{Ginzburg--Landau functional on signed graphs}  \label{sec:signedGL}

In this section, we define a new affinity matrix that also incorporates negative weights, and investigate how the degree and Laplacian matrices can be suitably modified. In this way, we also show how negative weights can be interpreted in the context of the GL functional. Moreover, we discuss how the standard MBO algorithm can be simplified.
\subsection{Signed Laplacians}\label{sc:signedLap}
In the following, we will consider connected weighted undirected simple graphs (hence with symmetric affinity matrices and no self-loops). The corresponding affinity matrix $A$ is a symmetric matrix whose entries $A_{ij}$ are zero if $i=j$ or $i\not\sim j$ and nonzero real otherwise.
This can be decomposed in two matrices $A^+\in\mathbb{R}_{+0}^{V\times V}$ and $A^-\in\mathbb{R}_{+0}^{V\times V}$ whose components are defined as
\begin{equation}
A_{ij}^{+}\coloneqq\max\left\{ A_{ij},0\right\} \in\mathbb{R}_{+0}\text{ and } A_{ij}^{-}\coloneqq-\min\left\{ A_{ij},0\right\} \in\mathbb{R}_{+0},
\end{equation}
respectively, so that
\begin{equation}\label{eq:Adecomp}
A=A^+-A^- \in\mathbb{R}^{V\times V}.
\end{equation}
In this decomposition, note that  $A^+$, respectively $A^-$, takes into account the affinity, respectively antipathy, between  nodes. Along  the same lines as \cite{kunegis2010spectral}, we define the \textit{signed degree matrix} as
\begin{equation}\label{eq:signedDegmatrix}
\bar{D}_{ij}\coloneqq\left\{ \begin{array}{cc}
\begin{array}{c}
\sum_{k=1}^{V}\left|A_{ik}\right|\in\mathbb{R}_{+0}\\
0
\end{array} & \begin{array}{c}
\text{if }i=j\\
\text{if }i\neq j
\end{array}\end{array}\right. ,
\end{equation}
where $\left|A_{ik}\right|$ is the absolute value of $A_{ik}$. As for the affinity matrix, $\bar{D}$ can be decomposed into $D^{+}\in\mathbb{R}_{+0}^{V\times V}$ and $D^{-}\in\mathbb{R}_{+0}^{V\times V}$, with components
\begin{equation}
D_{ij}^{+}\coloneqq\sum_{k=1}^{V}A_{ik}^{+}\in\mathbb{R}_{+0}\text{ and }\quad D_{ij}^{-}\coloneqq\sum_{k=1}^{V}A_{ik}^{-}\in\mathbb{R}_{+0},
\end{equation}
so that
\begin{equation}\label{eq:Ddecomp}
\bar{D}=D^{+}+D^{-}\in\mathbb{R}_{+0}^{V\times V}.
\end{equation}
Following \cite{kunegis2010spectral}, the \textit{signed Laplacian} is defined as
\begin{equation}\label{eq:signLap}
\bar{L}\coloneqq\bar{D}-A\in\mathbb{R}^{V\times V},
\end{equation}
together with its variants
\begin{equation}\label{eq:signLapRW}
\bar{L}_{\text{rw}}\coloneqq\bar{D}^{-1}\bar{L}=I_{V}-\bar{D}^{-1}A\in\mathbb{R}^{V\times V}
\end{equation}
and
\begin{equation}\label{eq:signLapSym}
\bar{L}_{\text{sym}}\coloneqq\bar{D}^{-\frac{1}{2}}\bar{L}\bar{D}^{-\frac{1}{2}}=\bar{D}^{\frac{1}{2}}\bar{L}_{\text{rw}}\bar{D}^{-\frac{1}{2}}=I_{V}-\bar{D}^{-\frac{1}{2}}A\bar{D}^{-\frac{1}{2}}\in\mathbb{R}^{V\times V}.
\end{equation}
Kunegis et al. \cite{kunegis2010spectral} showed that the signed Laplacian $\bar{L}$ is positive semidefinite, and furthermore, unlike the ordinary Laplacian $L$, the signed Laplacian  $\bar{L}$ is  \textit{strictly positive-definite} for certain  graphs, including most real-world networks. It is known that (see for instance \cite{golub2012matrix,gallier2014elementary}) if a matrix $Y_1$ is positive semidefinite then $Y_2\coloneqq X Y_1 X^T$ is also positive semidefinite provided that $X$ is invertible. We can apply this to show from the definition \eqref{eq:signLapSym} that $\bar{L}_{\text{sym}}$ is also positive semidefinite since $\bar{D}$ is diagonal with rank $V$. Positive semidefiniteness of $\bar{L}_{\text{rw}}$ follows from the second equality in \eqref{eq:signLapSym} since it shows that $\bar{L}_{\text{rw}}$ and $\bar{L}_{\text{sym}}$ are similar matrices.

In particular, for \eqref{eq:signLap} we can see that a natural decomposition stemming from \eqref{eq:Adecomp} and \eqref{eq:Ddecomp} occurs,
\begin{equation}
\bar{L}=\bar{D}-A=\left(D^{+}-A^{+}\right)+\left(D^{-}+A^{-}\right)\eqqcolon L^{+}+Q^{-},
\end{equation}
where we stress that $L^{+}\in\mathbb{R}^{V\times V}$ and $Q^{-}\in\mathbb{R}_{+0}^{V\times V}$. The latter is known in the literature as \textit{signless Laplacian} \cite{haemers2004enumeration} and, to the best of our knowledge, first appeared, although without this nomenclature, in \cite{desai1994characterization}.

To understand the meaning of the different terms contributing to the signed Laplacian, first consider the following quadratic form
\begin{equation}
\underbar{\ensuremath{x}}^{T}Q^{-}\underbar{\ensuremath{x}}=\underbar{\ensuremath{x}}^{T}D^{-}\underbar{\ensuremath{x}}+\underbar{\ensuremath{x}}^{T}A^{-}\underbar{\ensuremath{x}}.
\end{equation}
for some $\underbar{\ensuremath{x}}\in\mathbb{R}^V$. Using the definition of $D^-$ together with the fact that $A^-$ is symmetric with null diagonal, we have for the first term\footnote{Here, we use the following shorthand notation $\sum_{j<i}\equiv\sum_{\substack{j=1\\j< i}}^{V}$, together with $\sum_{j>i}\equiv\sum_{\substack{j=1\\j> i}}^{V}$ and $\sum_{j\neq i}\equiv\sum_{j<i}+\sum_{j>i}$.}
\begin{alignat}{2}
\underbar{\ensuremath{x}}^{T}D^{-}\underbar{\ensuremath{x}}=&\sum_{i=1}^{V}D_{ii}^{-}x_{i}^{2}=\sum_{i=1}^{V}\left(\sum_{j=1}^{V}A_{ij}^{-}\right)x_{i}^{2}=\sum_{i=1}^{V}\left(A_{ii}^{-}+\sum_{j\neq i}A_{ij}^{-}\right)x_{i}^{2}=\sum_{i=1}^{V}\left(\sum_{j<i}A_{ij}^{-}+\sum_{j>i}A_{ij}^{-}\right)x_{i}^{2}\nonumber\\
=&\sum_{i=1}^{V}\sum_{j<i}A_{ij}^{-}x_{i}^{2}+\sum_{i=1}^{V}\sum_{j>i}A_{ij}^{-}x_{i}^{2}=\sum_{i=1}^{V}\sum_{j<i}A_{ji}^{-}x_{i}^{2}+\sum_{i=1}^{V}\sum_{j>i}A_{ij}^{-}x_{i}^{2}\nonumber\\
=&\sum_{i=1}^{V}\sum_{j>i}A_{ij}^{-}x_{j}^{2}+\sum_{i=1}^{V}\sum_{j>i}A_{ij}^{-}x_{i}^{2},
\end{alignat}
and for the second
\begin{equation}
\underbar{\ensuremath{x}}^{T}A^{-}\underbar{\ensuremath{x}}=\sum_{i=1}^{V}A_{ii}^{-}x_{i}^{2}+\sum_{i=1}^{V}\sum_{j>i}A_{ij}^{-}x_{i}x_{j}+\sum_{i=1}^{V}\sum_{j<i}A_{ij}^{-}x_{i}x_{j}=2\sum_{i=1}^{V}\sum_{j>i}A_{ij}^{-}x_{i}x_{j}.
\end{equation}
Putting all the contributions together gives
\begin{equation}\label{eq:Qminmin}
\underbar{\ensuremath{x}}^{T}Q^{-}\underbar{\ensuremath{x}}=\sum_{i=1}^{V}\sum_{j>i}A_{ij}^{-}x_{i}^{2}+\sum_{i=1}^{V}\sum_{j>i}A_{ij}^{-}x_{j}^{2}+2\sum_{i=1}^{V}\sum_{j>i}A_{ij}^{-}x_{i}x_{j}=\sum_{i=1}^{V}\sum_{j>i}A_{ij}^{-}\left(x_{i}+x_{j}\right)^{2}\geqslant0.
\end{equation}
The minimum of \eqref{eq:Qminmin} is achieved when $x_i=-x_j$. 
In the simplified case in which the node vector is limited to dichotomic values such as $\underbar{\ensuremath{x}}\in\left\{ -1, 1\right\} ^{V}$, one can envision the system to be made up of $K=2$ clusters and the value of $x_i$ (being only $+1$ or $-1$) specifies the node-cluster association. In addition, with this particular choice,  \eqref{eq:Qminmin} counts the number of times, weighted for $A^{-}_{ij}$, two nodes were placed incorrectly in the same cluster (which we want to minimize). We note that choosing as dichotomic values $\left\{0, 1\right\}$ would reach a different conclusion since the result of the minimization would be $x_i=x_j=0$, that is, the nodes would belong to the same cluster. In regard to $L^{+}$ (being the Laplacian of the affinity matrix $A^+$), using either $\left\{ -1, 1\right\}$ or $\left\{0, 1\right\}$ leads to the same interpretation as for \eqref{eq:identityLap}. A more rigorous point of view based on the TV measure can be found in \cite{KeetchvanGennip17} and references therein.

The elements of the canonical basis, which correspond to pure node-cluster associations, can be suitably transformed to take values  $\left\{ -1, 1\right\}$,  via  the following expression 
\begin{equation}\label{eq:basetrasf}
\underbar{\ensuremath{e}}_{k}^{\pm}\coloneqq2\underbar{\ensuremath{e}}_{k}-\underbar{1}_K\in\left\{ -1;1\right\} ^{K}.
\end{equation}
where the associated simplex is
\begin{equation}\label{eq:2mKsimplex}
\Sigma_{K}^{\pm}\coloneqq\left\{ \left(\begin{array}{cccc}
\sigma_{1}^{\pm}, & \sigma_{2}^{\pm}, & \dotsc, & \sigma_{K}^{\pm}\end{array}\right)\in\left[-1;1\right]^{K}:\sum_{k=1}^{K}\sigma_{k}^{\pm}=2-K\right\}.
\end{equation}
In addition, since each row of $U$ from \eqref{eq:origcharU} is a linear combination of $\underbar{\ensuremath{e}}_{k}$ whose coefficients sum up to $1$, the characteristic matrix in the new representation can be obtained as
\begin{equation}\label{eq:Utrasf}
U^{\pm}\coloneqq2U-1_{V}\in\left(\Sigma_{K}^{\pm}\right)^{V}.
\end{equation}
In the end, by plugging $\bar{L}$ in place of $L$ in \eqref{eq:multiKGLF}, we have that the \textit{signed GL} functional is given by
\begin{alignat}{2}\label{eq:signedGL}
\bar{F}_{\varepsilon}(U^{\pm})\coloneqq &\frac{\varepsilon}{8}\langle U^{\pm},\bar{L}U^{\pm}\rangle+\frac{1}{2\varepsilon}\sum_{i=1}^{V}\left(\prod_{k=1}^{K}\frac{1}{16}\left\Vert \underbar{\ensuremath{u}}_{i}^{\pm}-\underbar{\ensuremath{e}}_{k}^{\pm}\right\Vert _{1}^{2}\right)\\
= &\frac{\varepsilon}{8}\langle U^{\pm},L^{+}U^{\pm}\rangle+\frac{\varepsilon}{8}\langle U^{\pm},Q^{-}U^{\pm}\rangle+\frac{1}{2\varepsilon}\sum_{i=1}^{V}\left(\prod_{k=1}^{K}\frac{1}{16}\left\Vert \underbar{\ensuremath{u}}_{i}^{\pm}-\underbar{\ensuremath{e}}_{k}^{\pm}\right\Vert _{1}^{2}\right),
\end{alignat}
where the different normalization of the terms comes from \eqref{eq:basetrasf} and \eqref{eq:Utrasf}. In other words, through the signed GL functional we want to minimize the presence of positive links between clusters as well as negative links within each one. We can verify that the following equality  holds
\begin{equation}
\bar{F}_{\varepsilon}(U^{\pm})-\frac{\varepsilon}{8}\langle U^{\pm},Q^{-}U^{\pm}\rangle=F_{\varepsilon}(U),
\end{equation}
where, on the right-hand side, the Laplacian of the GL functional is $L^+$, and we used the fact that $L^{+}1_{V}=0$ together with known properties of the trace of the product of two matrices.

\subsection{MBO scheme for signed graphs}\label{sc:shortMBO}
The MBO scheme arising from the signed GL has the same structure as in \eqref{eq:intermediateMBO}, that is
\begin{equation}\label{eq:SIGNintermediateMBO}
\left(U^{\pm}\right)^{n+\frac{1}{2}}=\bar{B}\left(U^{\pm}\right)^{n}\quad\text{where}\quad \bar{B}\coloneqq\left[\frac{1}{4}\bar{X}\left(I_{V}+\frac{d\tau}{N_\tau}\bar{\Lambda}\right)^{-1}\bar{X}^{T}\right]^{N_{\tau}}\in\mathbb{R}^{V\times V},
\end{equation}
together with $\bar{X}$ and $\bar{\Lambda}$ which are the eigenvector and eigenvalue matrices of $\bar{L}$, respectively\footnote{Here we have used another definition for $\bar{B}$ which is different from the definition of $B$ in \eqref{eq:intermediateMBO}. In fact, since $\bar{L}$ is real and symmetric its eigenvector matrix is orthogonal, that is $\bar{X}^{T}=\bar{X}^{-1}$. So, we have that $\left[\bar{X}\left(I_{V}+\frac{d\tau}{N_{\tau}}\bar{\Lambda}\right)\bar{X}^{T}\right]^{-N_{\tau}}=\left[\left(I_{V}+\frac{d\tau}{N_{\tau}}\bar{\Lambda}\right)\bar{X}^{T}\right]^{-N_{\tau}}\bar{X}^{-N_{\tau}}=\left(\bar{X}^{T}\right)^{-N_{\tau}}\left(I_{V}+\frac{d\tau}{N_{\tau}}\bar{\Lambda}\right)^{-N_{\tau}}\bar{X}^{-N_{\tau}}=\left[\bar{X}\left(I_{V}+\frac{d\tau}{N_{\tau}}\bar{\Lambda}\right)^{-1}\bar{X}^{T}\right]^{N_{\tau}}$.}.
As explained in the Appendix \ref{app:projsimplex}, the succession of projection (which in \cite{Garcia-CardonaMerkurjevBertozziFlennerPercus14} is performed using the algorithm developed in \cite{chen2011projection}) and thresholding steps can be reduced to determine the entry of the largest component of each row of $\left(U^{\pm}\right)^{n+\frac{1}{2}}$, so that\footnote{As before, when $\left(\underbar{\ensuremath{u}}_{i}^{\pm}\right)^{n+1}$ has the same distance to multiple vertices, one of those is picked at random with uniform probability.}
\begin{equation}\label{eq:SIGNargmax}
\left(\underbar{\ensuremath{u}}_{i}^{\pm}\right)^{n+1}=\underbar{\ensuremath{e}}_{k}^{\pm}\quad\text{with}\quad k=\underset{j=1,\dotsc,K}{\mathrm{argmax}}\left\{ \left(u_{ij}^{\pm}\right)^{n+\frac{1}{2}}\right\} .
\end{equation}
However, we note that the calculation of the $\text{argmax}$ in \eqref{eq:SIGNargmax} is independent of the representation chosen for the characteristic matrix. In fact, using \eqref{eq:Utrasf} we can see that
\begin{equation}
\left(U^{\pm}\right)^{n+\frac{1}{2}}=2\bar{B}U^{n}-\bar{B}1_{V}. 
\end{equation}
Since $\bar{B}1_{V}$ has constant rows, the location of the largest entry in a row of $\left(U^{\pm}\right)^{n+\frac{1}{2}}$ is the same as in $\bar{B}U^{n}$, that is from \eqref{eq:SIGNargmax}
\begin{equation}
k =\underset{j=1,\dotsc,K}{\mathrm{argmax}}\left\{ 2\sum_{h=1}^{V}\bar{b}_{ih}u_{hj}^{n}-\sum_{h=1}^{V}\bar{b}_{ih}\right\} =\underset{j=1,\dotsc,K}{\mathrm{argmax}}\left\{ 2\sum_{h=1}^{V}\bar{b}_{ih}u_{hj}^{n}\right\}.
\end{equation}
since $\sum_{h=1}^{V}\bar{b}_{ih}$ is independent of $j$ by definition. The factor $2$ can also be neglected since it does not change the actual node-cluster association. We indicate this equivalence with the symbol
\begin{equation}\label{eq:repequivMBO}
\left(U^{\pm}\right)^{n+\frac{1}{2}}\approx\bar{B}U^{n}.
\end{equation}
In other words, after applying $\bar{B}$, the position of the largest entry of $\left(\underbar{\ensuremath{u}}_{i}^{\pm}\right)^{n+1}$, resulting from the argmax in \eqref{eq:SIGNargmax}, is the same as the one of $\left(\underbar{\ensuremath{u}}_{i}\right)^{n+1}$. In this way, we can keep using a sparse node-cluster association matrix by picking $U$ instead of $U^\pm$. 
%
%

%
%
\section{GLF with constraints}\label{sc:GLFwithcons}
Now, we extend the GLF to three possible types of constraints. Here, we distinguish between two different types of constraints: \textit{soft}, when the final result of the algorithm does not necessarily satisfy the original constraint, and \textit{hard}, when the output matches the constraint.
\subsection{Must-link and cannot-link}
In practice, the affinity matrix usually results from measurements of a system's properties (see Section \ref{sec:imageSeg} for example). However, in the same spirit as in \cite{cucuringu2016simple}, additional information in the form of both positive (must-link) and negative (cannot-link) edge weights, can also be known and could have a different relative strength with respect to measurements. In quantitative terms, we will embed must-links in $A^+$ and cannot-links in $A^{-}$ and generate the new degree matrices $D^+$ and $D^-$ accordingly. To this end, we define
\begin{equation}
\mathcal{A}^{+}\coloneqq A^{+} + \lambda^+ M\in\mathbb{R}_{+0}^{V\times V},
\end{equation}
where $M\in\mathbb{R}_{+0}^{V\times V}$ is the \textit{must-link matrix}, and $\lambda^+ \in \mathbb{R}_{+} $ is a trade-off parameter. 
The role of $M$  consists in adding links to nodes which should be belong to the same cluster or to increase the weight of a link to reinforce node-cluster associations. In general, it is a symmetric sparse matrix with positive weights (whose values can also be different among each couple of nodes). We indicate the resulting degree matrix for $\mathcal{A}^+$, defined in the same fashion as in \eqref{eq:signedDegmatrix}, with $\mathcal{D}^+\in\mathbb{R}_{+0}^{V\times V}$. Similarly 
\begin{equation}
\mathcal{A}^{-}\coloneqq A^{-} + \lambda^- C\in\mathbb{R}_{+0}^{V\times V},
\end{equation}
where $C\in\mathbb{R}_{+0}^{V\times V}$ is the \textit{cannot-link matrix}, and $\lambda^-   \in \mathbb{R}_{+} $ the trade-off parameter. 
In the same fashion as $M$, the matrix $C$ can include links to nodes which should not belong to the same cluster or to decrease even more the weight of a link pushing node couples to belong to different clusters. In this case as well, $C$ is a symmetric sparse matrix with positive weights, where the greater the weight the more the two nodes are incentivised to belong to different clusters. Also here we have $\mathcal{D}^-\in\mathbb{R}_{+0}^{V\times V}$ as the degree matrix of $\mathcal{A}^-$. In the end, we have that the GLF in the presence of must- and cannot-links is given as
\begin{equation}
\bar{\mathcal{F}}_{\varepsilon}(U^\pm)=\frac{\varepsilon}{8}\langle U^\pm,\mathcal{L}^{+}U^\pm\rangle+\frac{\varepsilon}{8}\langle U^\pm,\mathcal{Q}^{-}U^\pm\rangle+\frac{1}{2\varepsilon}\sum_{i=1}^{V}\left(\prod_{k=1}^{K}\frac{1}{16}\left\Vert \underbar{\ensuremath{u}}_{i}^\pm-\underbar{\ensuremath{e}}_{k}^\pm\right\Vert _{1}^{2}\right),
\end{equation}
where $\mathcal{L}^{+}\coloneqq\mathcal{D}^{+}-\mathcal{A}^{+}$, $\mathcal{Q}^{-}\coloneqq\mathcal{D}^{-}+\mathcal{A}^{-}$ so that $\mathcal{L}\coloneqq \mathcal{L}^+ +\mathcal{Q}^-$. As in \eqref{eq:SIGNintermediateMBO}, we now need to determine the matrix $\bar{\mathcal{B}}$ resulting from the eigendecomposition of $\mathcal{L}$.
\subsection{Fidelity and avoidance terms}
Soft constraints can also be implemented at the level of the characteristic matrix mainly in two ways. First, as a \textit{fidelity term} (see for instance \cite{Garcia-CardonaMerkurjevBertozziFlennerPercus14}) which attracts the dynamics of the ACE towards a particular characteristic matrix $\hat{U}$. At the level of the GLF this translates into a term which increases the overall energy when the node-cluster association in $\hat{U}$ is not matched. Second, as an \textit{avoidance term} which pushes the ACE dynamics to avoid particular node-cluster combinations expressed in a characteristic matrix $\tilde{U}$. Again, for the GLF this means an energy increase when the wrong node-cluster association is met. We notice that while the first constraint points the system to a particular combination and increasing the energy every time this condition is not met, the second does the opposite, that is only one combination increases the energy and leaves all the others as equally possible. The GLF for an affinity matrix with positive edges only reads as 
\begin{equation}
G_{\varepsilon}^{R}(U,\hat{U},\tilde{U})\coloneqq F_{\varepsilon}(U)+\frac{1}{2}\sum_{i=1}^{V}R^{fi}_{ii}\left\Vert \underbar{\ensuremath{u}}_{i}-\underbar{\ensuremath{\hat{u}}}_{i}\right\Vert _{2}^{2}+\sum_{i=1}^{V}R^{av}_{ii}\left\Vert \underbar{\ensuremath{u}}_{i}\circ\underbar{\ensuremath{\tilde{u}}}_{i}\right\Vert _{1}
\end{equation}
where $\circ$ is the Hadamard product and $R^{fi}$ (and $R^{av}$) is a diagonal matrix whose element $R_{ii}^{fi}$ ($R_{ii}^{av}$) is positive if and only if there exists a $j\in \{1, \ldots, V\}$ such that $\hat{U}_{ij} \neq 0$ ($\tilde{U}_{ij} \neq 0$) and is zero otherwise. At the level of the MBOE, following \cite{Garcia-CardonaMerkurjevBertozziFlennerPercus14} the first step in the procedure is given by
\begin{equation}\label{eq:Uhalf}
U^{n+\frac{1}{2}}=B\left[U^{n}-d\tau R^{fi} \left(U^{n}-\hat{U}\right)-d\tau R^{av}\tilde{U}\right].
\end{equation}
We can obtain an equivalent expression for the signed GLF using \eqref{eq:Utrasf}. The result is  
\begin{equation}
\bar{G}_{\varepsilon}^{R}(U^{\pm},\hat{U}^{\pm},\tilde{U}^{\pm})\coloneqq \bar{F}{}_{\varepsilon}(U^{\pm})+\frac{1}{8}\sum_{i=1}^{V}R^{fi}_{ii}\left\Vert \underbar{\ensuremath{u}}_{i}^{\pm}-\underbar{\ensuremath{\hat{u}}}_{i}^{\pm}\right\Vert _{2}^{2}+\frac{1}{4}\sum_{i=1}^{V}R^{av}_{ii}\left\Vert \underbar{\ensuremath{u}}_{i}^{\pm}\circ\underbar{\ensuremath{\tilde{u}}}_{i}^{\pm}+\underbar{\ensuremath{u}}_{i}^{\pm}+\underbar{\ensuremath{\tilde{u}}}_{i}^{\pm}+\underbar{\ensuremath{1}}_{V}\right\Vert _{1}.
\end{equation}
Now the MBOE \eqref{eq:SIGNintermediateMBO} becomes
\begin{equation}
\left(U^{\pm}\right)^{n+\frac{1}{2}}=\bar{B}\left\{ \left(U^{\pm}\right)^{n}-\frac{d\tau}{4}\left[R^{fi}\left(\left(U^{\pm}\right)^{n}-\hat{U}^{\pm}\right)+R^{av}\left(\tilde{U}^{\pm}+1_{V}\right)\right]\right\},
\end{equation}
and after the transformation $\eqref{eq:Utrasf} $ becomes
\begin{equation}\label{eq:eqbeforefinal}
\left(U^{\pm}\right)^{n+\frac{1}{2}}=2\left\{\bar{B}U^{n}-\frac{1}{2}\bar{B}1_{V}-\frac{d\tau}{4}\bar{B}\left[R^{fi}\left(U^{n}-\hat{U}\right)+R^{av}\tilde{U}\right]\right\},
\end{equation}
The factor $\frac{1}{4}$ can be embedded in the definitions of $R^{fi}$ and $R^{av}$. As explained in Section \ref{sc:shortMBO}, since $1_V$ has constant rows, the term $\bar{B}1_{V}$ does not modify the result of the $\text{argmax}$ problem. Moreover, since the global factor is positive, the presence of the coefficient $2$ in front of the brackets in \eqref{eq:eqbeforefinal} does not affect which entry is the largest and it can be neglected. So, we can consider only the following terms
\begin{equation}
\left(U^{\pm}\right)^{n+\frac{1}{2}}\approx\bar{B}\left\{ U^{n}-d\tau\left[R^{fi}\left(U^{n}-\hat{U}\right)+R^{av}\tilde{U}\right]\right\} .
\end{equation}

\subsection{Anchors}\label{sc:Anchors}
In general, \textit{anchors} correspond to a subset of the network nodes whose values are fixed once and for all.    
For example, in the context of sensor network localization (that aims to find coordinates of sensors in $\mathbb{R}^d$ from given pairwise distance measurements; typically $d=\{2,3\}$), anchors are sensors  that are aware of their location, and anchor-based algorithms make use of their existence when computing the coordinates of the remaining sensors \cite{asap2d}. 
In the context of the \textit{molecule problem} from structural biology applications \cite{asap3d} (motivated by NMR), there are molecular fragments (i.e., embeddings of certain subgraphs)  whose local configuration is known in advance up to an element of the special Euclidean group SE(3),  rather than the Euclidean group E(3). In other words, in the terminology of signed clustering for $k=2$, there exists a subset of nodes (anchors) whose corresponding ground truth assigned value $\{-1, 1\} \in \mathbb{Z}_2$ (i.e., cluster membership) is known a priori, and can be leveraged to find the values of the remaining non-anchor  nodes.

In the context of multiple clusters, anchors are understood as those nodes whose cluster membership is known. From the point of view of the MBO, we can consider anchors as fixed points in the dynamics, that is from \eqref{eq:discrMBO}\footnote{For this presentation we use the notation for the unsigned case,  but the arguments apply to $\bar{L}$ as well.}
\begin{equation}
\frac{\underbar{\ensuremath{u}}_{i}^{n+\frac{1}{2}}-\underbar{\ensuremath{u}}_{i}^{n}}{d\tau}=0.
\end{equation}
Anchors are implemented in the following way. Suppose we know that node $z$ always belongs to cluster $k$, that is 
\begin{equation}
U^{n}=\left(\begin{array}{cccccccc}
\underbar{\ensuremath{u}}_{1}^{n}, & \underbar{\ensuremath{u}}_{2}^{n}, & \dotsc, & \underbar{\ensuremath{u}}_{z-1}^{n}, & \underbar{\ensuremath{e}}_{k}, & \underbar{\ensuremath{u}}_{z+1}^{n}, & \dotsc, & \underbar{\ensuremath{u}}_{V}^{n}\end{array}\right)^{T}\,\text{ for all } n,
\end{equation}
including $n=0$ and the intermediate steps $n+\frac{1}{2}$. To include this a priori information in the MBO scheme, it is sufficient to use at the level of \eqref{eq:intermediateMBO} an evolution matrix (to be used in place of the original $B$) which leaves the anchor node unaffected by the dynamics. The latter can be defined as
\begin{equation}\label{eq:Banch}
B_{ij}^{a}\coloneqq\left\{ \begin{array}{cc}
B_{ij} & \text{if }i\neq z\\
0 & \text{if }i=z\text{ and }j\neq z\\
1 & \text{if }i=j=z
\end{array}\right. .
\end{equation}
In general, to take into account the presence of a second anchor node it is sufficient to modify all the relevant entries of $B^a_{ij}$ in the same fashion as in equation \eqref{eq:Banch}. This procedure can be repeated for all the anchor nodes.

%
%
\section{The algorithm}\label{sc:thealgo}
This section is devoted to describing the general structure of the MBO scheme. First, we will show the basic implementation and then we will explain how to modify the relevant structures to incorporate different types of constraints.\\

\begin{algorithm}[H]
\rule{6.5cm}{0.5mm}\\
 \textbf{Require}\\
$V$,  $K$,  $d\tau$, $N_\tau$, $m$, $\bar{L}$, $\epsilon$\;
\rule{6.5cm}{0.2mm}\\
 \textbf{Initialize}\\
 $\bar{B}_{m}\leftarrow\bar{X}_{m}\left(I_{m}+d\tau\bar{\Lambda}_{m}\right)^{-1}\bar{X}_{m}^{T}$\;
 $U^0=0_{V\times K}$\;
\For{$i\leftarrow 1$ \KwTo $V$}{
$k=randperm(K,1)$\;
$u_{ik}^{0}=1$\;
}
\rule{6.5cm}{0.2mm}\\
\textbf{MBO Scheme}\\
$n\leftarrow0$\;
 \While{Stop criterion not satisfied}{
\For{$s\leftarrow 1$ \KwTo $N_\tau$}{
$U^{n+\frac{1}{2}}\leftarrow\bar{B}_{m} U^{n}$\;
$U^{n}\leftarrow U^{n+\frac{1}{2}}$\;
}
\For{$i\leftarrow 1$ \KwTo $V$}{
$\underbar{\ensuremath{u}}_{i}^{n+1}=\underbar{\ensuremath{0}}_{K}$\;
$k=\argmax_{j=1,\dotsc,K}\left\{u_{ij}^{n+\frac{1}{2}}\right\}$\;
$u_{ik}^{n+1}=1$\;
}
$U^n\leftarrow U^{n+1}$
}
\rule{6.5cm}{0.2mm}\\
\textbf{Output} $\underbar{\ensuremath{v}}^{*}\in\left\{ 1;\dotsc;K\right\} ^{V}$\\
\For{$i\leftarrow 1$ \KwTo $V$}{
$v^*_{i}=\argmax_{j=1,\dotsc,K}\left\{ u_{ij}^{n}\right\}$\;
}
\rule{6.5cm}{0.5mm}\\
\medskip
\caption{Structure of the MBO scheme for signed Laplacians in the absence of constraints. The algorithm requires values specifying the affinity matrix and some parameters for best code-tuning which are explained in the main text of the present section. The output of the code is a vector $\underbar{\ensuremath{v}}^{*}$ which indicates the cluster (numbered from $1$ to $K$) each node belongs to. The command $randperm(K,1)$ generates a vector of length $K$ whose entries are natural numbers uniformly distributed in the interval $\left[1,K\right]$, i.e., it associates to each node a random cluster membership.}
\label{algo:basicMBO}
\end{algorithm}

\medskip
 
The MBO scheme for signed Laplacians we propose is given in Algorithm \ref{algo:basicMBO}. In addition to the graph variables (i.e., the number of nodes, clusters and the signed Laplacian), we also need to specify   
\begin{itemize}
\itemsep0mm
\item  $ d\tau\in\mathbb{R}_+$ which is the length of the interval of the time discretization procedure as in \eqref{eq:discrMBO}, 
\item   $N_\tau\in\mathbb{N}$ which is the number of times we repeat the diffusion step in \eqref{eq:MBOE1d},
\item  $ m \in \mathbb{N}$ which is the number of eigenvectors and eigenvalues considered in the eigendecomposition of $\bar{B}$ in \eqref{eq:SIGNintermediateMBO}  (we refer the reader to the next paragraph for a theoretical motivation of this choice). 
\end{itemize}

The corresponding eigenvector and eigenvalue matrices (thus having $m$ eigenvectors and eigenvalues) are indicated as $\bar{X}_m$ and $\bar{\Lambda}_m$, respectively. The parameter $\epsilon$ is defined later in \eqref{eq:errorparamep}. In regard to $d\tau$, in Section \ref{sc:bestpract} we show, empirically on a synthetic data set, that the outcomes of algorithm \ref{algo:basicMBO} are independent of the chosen value as long as it remains small. In \cite{Garcia-CardonaMerkurjevBertozziFlennerPercus14},  it is reported that the optimal value for the number of diffusion step iterations is $N_\tau=3$. Finally, the number of eigenvectors considered is chosen to be equal  to the number of clusters set a priori since, at a heuristic level, the eigenvalues empirically exhibit a big jump after the $K$-th term (although this may not always hold true, see \cite{von2007tutorial}). On a related note, we refer the reader to the celebrated result of Lee, Gharan, and Trevisan on higher-order Cheeger inequalities for undirected graphs with $k\geq 2$ clusters \cite{lee2014multiway}, which motivates the use of the top $k$ eigenvectors and generalizes initial results on $k=2$  \cite{Trevisan09}.
Moreover, the version of the algorithm we propose does takes into account the projection/thresholding step in a different fashion, as explained in Section \ref{sc:shortMBO}, which allows us to shorten the running time. To include various constraints, it suffices to modify only a few relevant parts of the code in algorithm \ref{algo:basicMBO}. In particular, we consider 
\begin{itemize}
\item \textbf{Must-link and cannot-link soft constraints} - 
here, we just need to use $\mathcal{B}$ in place of $\bar{B} $ in \eqref{eq:repequivMBO},
\item \textbf{Fidelity and avoidance terms} - the evolution equation given by $\bar{B}$ now includes the additional terms as in \eqref{eq:Uhalf},
\item \textbf{Anchor nodes} - the matrix $\bar{B}$ is modified, as explained in Section \ref{sc:Anchors}, to include the fixed nodes in the MBO dynamics.
\end{itemize}
In all cases, as a stop criterion we use, as in \cite{Garcia-CardonaMerkurjevBertozziFlennerPercus14},  
\begin{equation}\label{eq:errorparamep}
\frac{\max_{i}\left\Vert \underbar{\ensuremath{u}}_{i}^{n+1}-\underbar{\ensuremath{u}}_{i}^{n}\right\Vert _{2}^{2}}{\max_{i}\left\Vert \underbar{\ensuremath{u}}_{i}^{n+1}\right\Vert _{2}^{2}}<\epsilon\in\mathbb{R}_+.
\end{equation}
All the constraints can also be combined together since they act at different levels in the scheme. However, we notice that including anchors automatically rules out possible soft constraints acting on the same node.

%
%
\section{Numerical experiments for a signed stochastic block model} \label{sec:numericalSection}

This section is devoted to present the performance of the MBO scheme for signed Laplacians on synthetic data sets. This approach is useful since it allows a direct comparison of the results with the ground truth. Here, we consider a \textit{signed stochastic block model} (SSBM) followed by the definition of the cost function used to compare the output of the MBO algorithm with the ground truth.

Given a set of numbers $\left\{ c_{1},c_{2},\dotsc,c_{K}\right\} \in\mathbb{N}$ representing the size of $K$ clusters of $V$ nodes, the corresponding affinity matrix for the fully connected case is given in block form as
\begin{equation}\label{eq:grmatrix}
S\coloneqq\left(\begin{array}{cccc}
1_{c_{1}} & -1_{c_{1}\times c_{2}} & \dotsc & -1_{c_{1}\times c_{K}}\\
-1_{c_{2}\times c_{1}} & 1_{c_{2}} & \dotsc & -1_{c_{2}\times c_{K}}\\
\vdots & \vdots & \ddots & \vdots\\
-1_{c_{K}\times c_{1}} & -1_{c_{K}\times c_{2}} & \dotsc & 1_{c_{K}}
\end{array}\right)
\end{equation}
where $1_{c_{i}}$ is a square matrix of size $c_i$ whose entries are all $1$ and $-1_{c_{i}\times c_{j}}\in\left\{ -1\right\} ^{c_{i}\times c_{j}}$ is a matrix with $c_i$ rows and $c_j$ columns whose elements are all $-1$ (for consistency, we have that $\sum_{k=1}^{K}c_{k}=V$). Here, we see that all the nodes belonging to the same (respectively different) cluster have a $+1$ (respectively $-1$) link. The corresponding output vector (see Algorithm \ref{algo:basicMBO}) for the ground truth is given as
\begin{equation}\label{eq:grvector}
\underbar{\ensuremath{s}}=\left(\begin{array}{cccc}
c_{1}\underbar{\ensuremath{1}}_{c_{1}}^T, & c_{2}\underbar{\ensuremath{1}}_{c_{2}}^T, & \dotsc, & c_{K}\underbar{\ensuremath{1}}_{c_{K}}^T\end{array}\right)^T\in\left\{ 1,\dotsc,K\right\} ^{V},
\end{equation}
where each entry corresponds to a node and its value to the cluster (indexed from $1$ to $K$) the node belongs to, and its characteristic matrix is
\begin{equation}\label{eq:grcharmat}
U_{s}=\left(\begin{array}{cccc}
\underbar{1}_{c_{1}} & 0 & \dotsc & 0\\
0 & \underbar{1}_{c_{2}} & \dotsc & 0\\
\vdots & \vdots & \ddots & \vdots\\
0 & 0 & \dotsc & \underbar{1}_{c_{K}}
\end{array}\right)\in\left\{ 0,1\right\} ^{V\times K}.
\end{equation}

An instance of the SSBM is obtained from \eqref{eq:grmatrix} by introducing \textit{sparsity} (understood as null entries) and \textit{noise} (which corresponds to a change of the sign of the elements $S_{ij}$) with probabilities $\lambda$ and $\eta$, respectively. More specifically, an edge is present independently with probability $\lambda$, while the sign of an existing edge is flipped with probability $\eta$. 
In quantitative terms, an SSBM matrix $A$, which we fix to be symmetric and null-diagonal, is generated from $S$ via the following probabilistic  mixture model. Consider a row index $i=1,\dotsc,V$ and a column index $j=i,\dotsc,V$ (the $j$ index starts from $i$ so to consider only the upper triangular part in the definition of the elements of $A$), we define the affinity matrix as
\begin{equation}\label{eq:genaffmatrix}
A_{ij}= \left\{
     \begin{array}{rll}
   S_{ij}  & \;\; \text{ if } i< j,  	& \text{with probability } \left(1-\eta \right)   \lambda 	\\
 - S_{ij}  & \;\; \text{ if } i< j, 	& \text{with probability }  \eta   \lambda	\\
		     0   & \;\; \text{ if } i< j, 	& \text{with probability } 1-\lambda	\\
		     		     0   & \;\; \text{ if } i= j, 	& \text{with probability } 1	\\
     \end{array}
   \right. \text{ with } A_{ji}=A_{ij}.
\end{equation}
Then, the corresponding Laplacian can be computed as detailed in Section \ref{sc:signedLap}.

As explained in Section \ref{sc:thealgo}, the output of the MBO algorithm is a vector $\underbar{\ensuremath{v}}^{*}$ which holds the information related to the node-cluster association in the same spirit as \eqref{eq:grvector}. Therefore, it is sensible to compare $\underbar{\ensuremath{v}}^{*}$ with $\underbar{\ensuremath{s}}$ to measure the ability of MBO to retrieve the ground truth clustering. 
To quantify success in terms of  robustness to noise and sampling sparsity, 
we use the popular \textit{adjusted rand index} (ARI) defined as follows \cite{ARI_JMLR_Gates_Ahn}. Consider the vector $\underbar{\ensuremath{v}}^{*}$ and its associated nodes partition $\mathscr{P}\left(\underbar{\ensuremath{v}}^{*}\right)=\left\{ \mathscr{P}_{1}\left(\underbar{\ensuremath{v}}^{*}\right),\mathscr{P}_{2}\left(\underbar{\ensuremath{v}}^{*}\right),\dotsc,\mathscr{P}_{K}\left(\underbar{\ensuremath{v}}^{*}\right)\right\} $ into $K$ clusters, where $\mathscr{P}_i\left(\underbar{\ensuremath{v}}^{*}\right)$ is a set whose elements are those nodes belonging to cluster $i$. Also for the ground truth vector, it is possible to construct a partition $\mathscr{P}\left(\underbar{\ensuremath{s}}\right)$ which is in general different from $\mathscr{P}\left(\underbar{\ensuremath{v}}^{*}\right)$. Then, the ARI is defined as
\begin{equation}
\text{ARI}\left(\underbar{\ensuremath{v}}^{*},\underbar{\ensuremath{s}}\right)\coloneqq\frac{\sum_{ij}\binom{p_{ij}}{2}-\left[\sum_{i}\binom{a_{i}}{2}\sum_{j}\binom{b_{j}}{2}\right]/\binom{V}{2}}{\frac{1}{2}\left[\sum_{i}\binom{a_{i}}{2}+\sum_{j}\binom{b_{j}}{2}\right]-\left[\sum_{i}\binom{a_{i}}{2}\sum_{j}\binom{b_{j}}{2}\right]/\binom{V}{2}}\in\mathbb{R},
\end{equation}
where $a_i$, $b_j$ and $p_{ij}$ are extracted from the \textit{contingency table}
\begin{equation}
\begin{array}{c|cccccc|c}
{{}\atop \underbar{\ensuremath{v}}^{*}}\!\diagdown\!^{\underbar{\ensuremath{s}}} & \mathscr{P}_{1}\left(\underbar{\ensuremath{s}}\right) & \mathscr{P}_{2}\left(\underbar{\ensuremath{s}}\right) & \ldots & \mathscr{P}_{j}\left(\underbar{\ensuremath{s}}\right) & \ldots & \mathscr{P}_{K}\left(\underbar{\ensuremath{s}}\right) & \sum_{j=1}^{K}p_{ij}\\
\hline \mathscr{P}_{1}\left(\underbar{\ensuremath{v}}^{*}\right) & p_{11} & p_{12} & \ldots & p_{1j} & \ldots & p_{1K} & a_{1}\\
\mathscr{P}_{2}\left(\underbar{\ensuremath{v}}^{*}\right) & p_{21} & p_{22} & \ldots & p_{2j} & \ldots & p_{2K} & a_{2}\\
\vdots & \vdots & \vdots & \ddots & \vdots & \ddots & \vdots & \vdots\\
\mathscr{P}_{i}\left(\underbar{\ensuremath{v}}^{*}\right) & p_{i1} & p_{i2} & \ldots & p_{ij} & \ldots & p_{iK} & a_{i}\\
\vdots & \vdots & \vdots & \ddots & \vdots & \ddots & \vdots & \vdots\\
\mathscr{P}_{K}\left(\underbar{\ensuremath{v}}^{*}\right) & p_{K1} & p_{K2} & \ldots & p_{Kj} & \ldots & p_{KK} & a_{K}\\
\hline \sum_{i=1}^{K}p_{ij} & b_{1} & b_{2} & \ldots & b_{j} & \ldots & b_{K}
\end{array},
\end{equation}
given
\begin{equation}
p_{ij}\coloneqq\left|\mathscr{P}_{i}\left(\underbar{\ensuremath{v}}^{*}\right)\cap \mathscr{P}_{j}\left(\underbar{\ensuremath{s}}\right)\right|\in\mathbb{N}
\end{equation}
the number of nodes in common between $\mathscr{P}_{i}\left(\underbar{\ensuremath{v}}^{*}\right)$ and $\mathscr{P}_{j}\left(\underbar{\ensuremath{s}}\right)$.

For the numerical tests, the different types of constraints described in Section \ref{sc:GLFwithcons} can be included by extracting, at random, information from the ground truth. More specifically,   
$M$ and $C$ can be obtained from $S$ for the must/cannot-link case, and $\hat{U}$ and $\tilde{U}$ from $U_s$ for the fidelity/avoidance term etc.  We indicate the percentage of information overlapping with the ground truth using the parameter $\alpha\in\left[0,1\right]$. Since we are incorporating constraints using ground truth information, it is clear that the results will be closer to the ground truth for increasing $\alpha$. The purpose of this set of experiments is simply to provide a proof of concept for our approach.

In summary, the numerical experiments we consider are generated using the following set of parameters summarized in Table \ref{tb:onlytab}, some of which are also explained in Section \ref{sc:thealgo}.
\begin{table}[H]
\centering 
\begin{tabular}{ |p{3cm}|p{4cm}|}
 \hline
Parameter  & Value\\
 \hline
Nodes & $V=1200$\\
Clusters & $K=2$, $3$, $4$, $5$, $8$, $10$\\
Size & Equal clusters size $V/K$\\
$d\tau$ & 0.1\\
$N_\tau$ & 3\\
$m$ & Equal to no. of clusters  $K$ \\
$\epsilon$ & $10^{-7}$\\
No. Experiments & $100$\\
Sparsity & $\lambda\in\left[0.02,0.1\right]$\\
Noise & $\eta\in\left[0.05,0.45\right]$\\
Known Information & $\alpha\in\left[0.05,0.9\right]$\\
 \hline
\end{tabular}
\caption{Parameter  values used to test the MBO scheme in the SSBM experiments.} 
\label{tb:onlytab}
\end{table}
Unless stated otherwise, the initial condition we use is the clustering induced by a segmentation based on a bottom eigenvector, that is an eigenvector associated with the smallest nonzero positive eigenvalue\footnote{Since the (signed) Laplacians we use are all positive semidefinite (and not identically zero) \cite{kunegis2010spectral}, a nonzero positive eigenvalue can always be found. When the geometric multiplicity of this eigenvalue is greater than one, we pick at random one of the elements of its eigenspace.}, of the Laplacian matrix (each experiment is specified by a different Laplacian matrix which will be used for its initial condition as well as the MBO evolution). In particular, we consider the interval given by the smallest and largest entries in the aforementioned eigenvector, and create $K$ equally-sized bins spanning this interval. Finally, each node will be allocated in its corresponding bin according to its value in the eigenvector. The resulting characteristic matrix is used as the initial condition. 
\subsection{Best practice for MBO}\label{sc:bestpract}
As explained in Section \ref{sc:signedLap}, one may use different versions of the Laplacians given by \eqref{eq:signLap}, \eqref{eq:signLapRW} and \eqref{eq:signLapSym}.  In particular,  in the context of spectral clustering of (unsigned) graphs, $L_{\text{rw}}$ and $L_{\text{sym}}$ typically provide better results compared to $L$, especially in the  context of skewed degree distributions
\cite{von2007tutorial,dasgupta2004spectral}.
We observe a similar behaviour for their respective signed counterparts 
$\bar{L}_{\text{rw}}$, $\bar{L}_{\text{sym}}$ and $L$, as can be seen in Figure \ref{fig:bestLap}. We remark that in subsequent experiments, we will only use $\bar{L}_{\text{sym}}$ as it achieves the highest accuracy. Moreover, we have tested the additional degree of freedom resulting from the time discretization procedure by running MBO with different orders of magnitude for the parameter $d\tau$. Results are detailed in Figure \ref{fig:bestTau}, and the overlapping curves illustrate that the results are not sensitive to the choice of the parameter $d\tau$. 
\begin{figure}[!ht]
\centering
\subcaptionbox[]{  %
}[ 0.32\textwidth ]
{\includegraphics[width=0.35\textwidth] {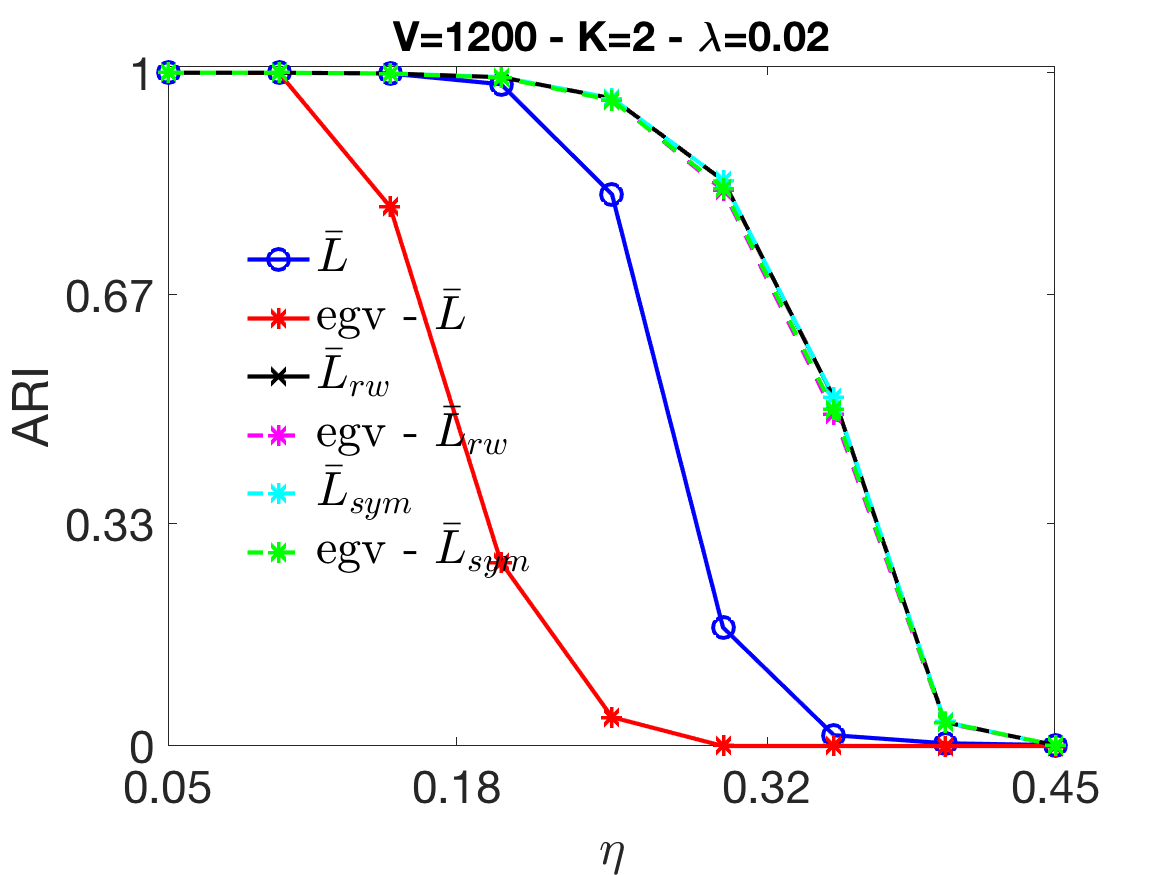} }
\subcaptionbox[]{  %
}[ 0.32\textwidth ]
{\includegraphics[width=0.35\textwidth] {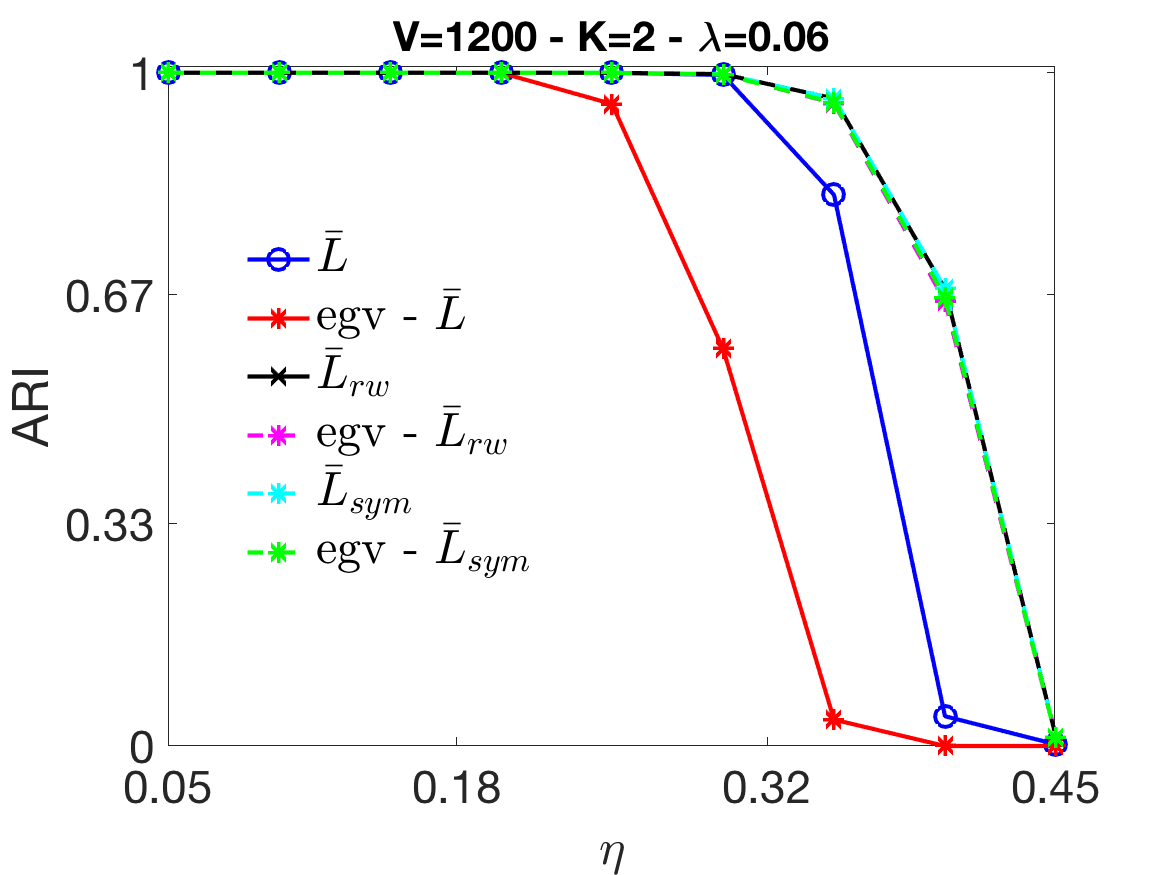} }
\subcaptionbox[]{  %
}[ 0.32\textwidth ]
{\includegraphics[width=0.35\textwidth] {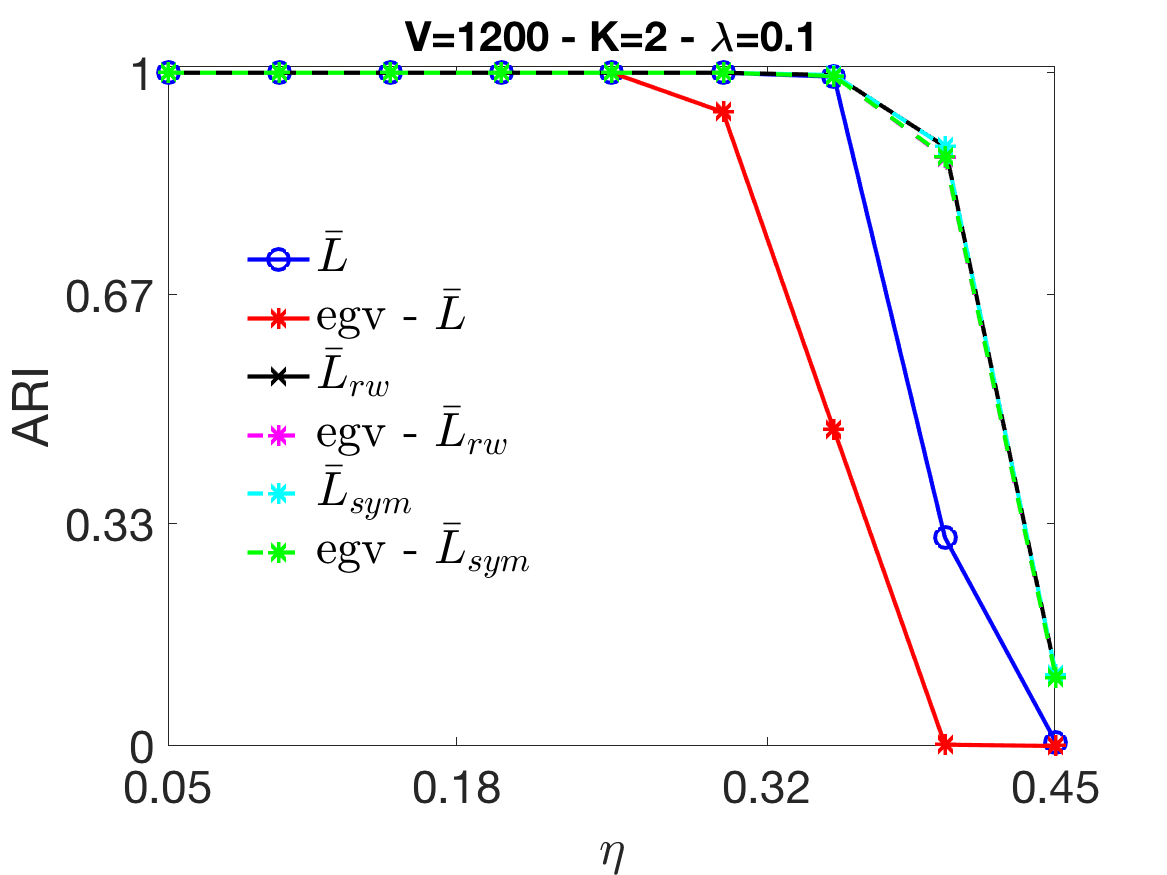} }
\subcaptionbox[]{  
}[ 0.32\textwidth ]
{\includegraphics[width=0.35\textwidth] {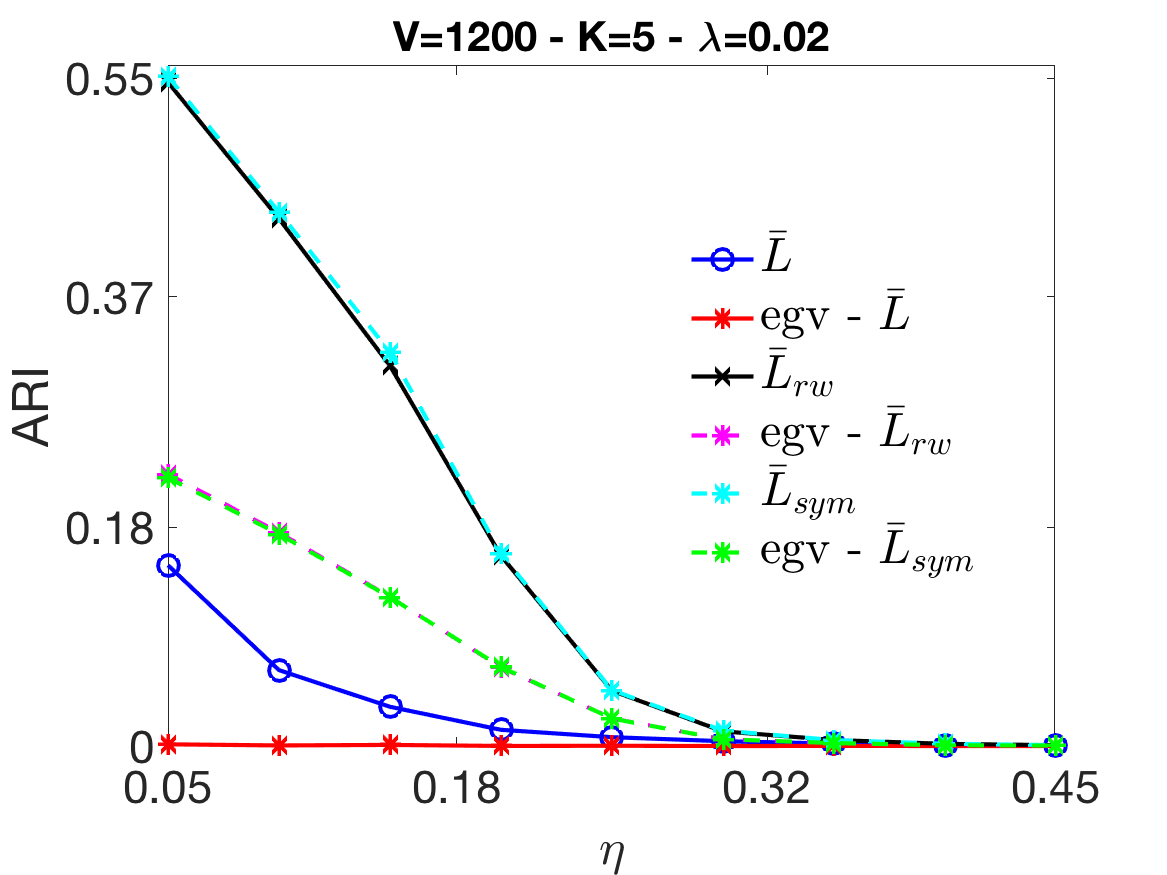} }
%
\subcaptionbox[]{  
}[ 0.32\textwidth ]
{\includegraphics[width=0.35\textwidth] {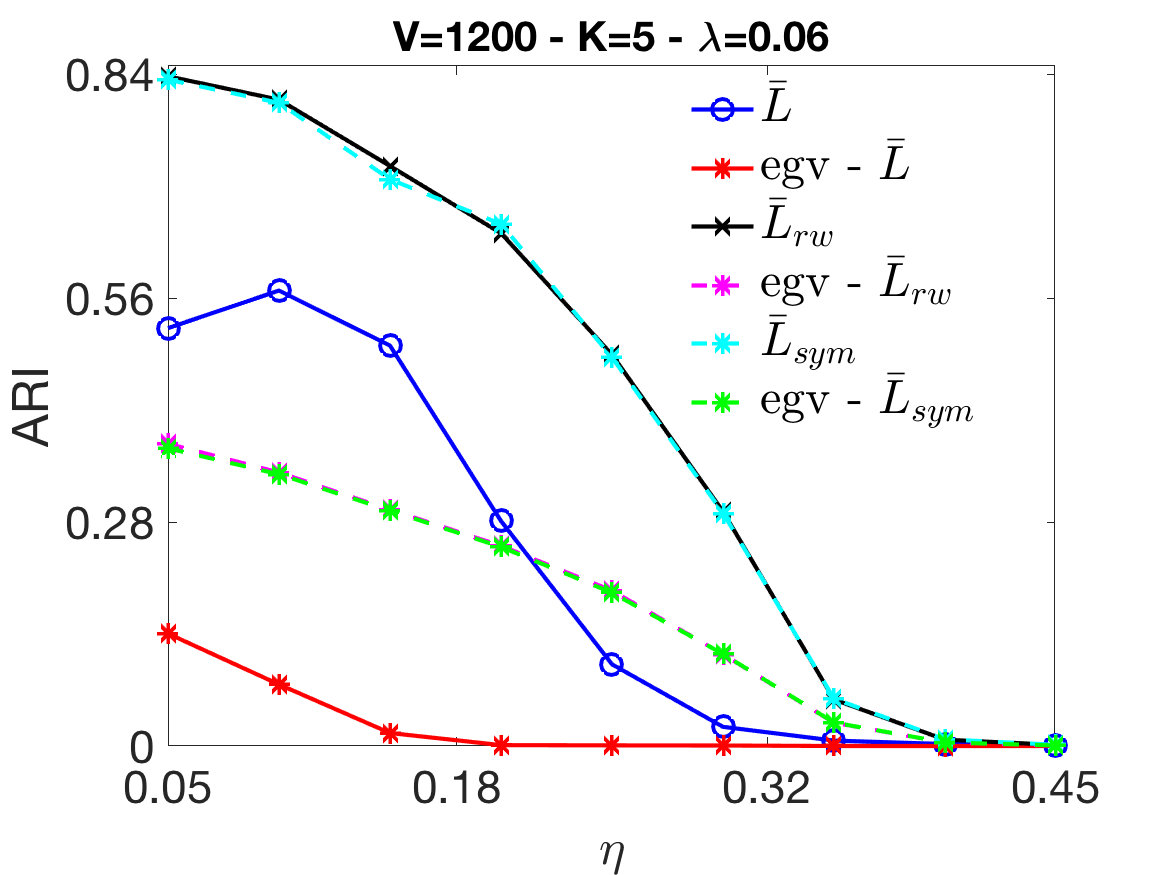} }
%
\subcaptionbox[]{  
}[ 0.32\textwidth ]
{\includegraphics[width=0.35\textwidth] {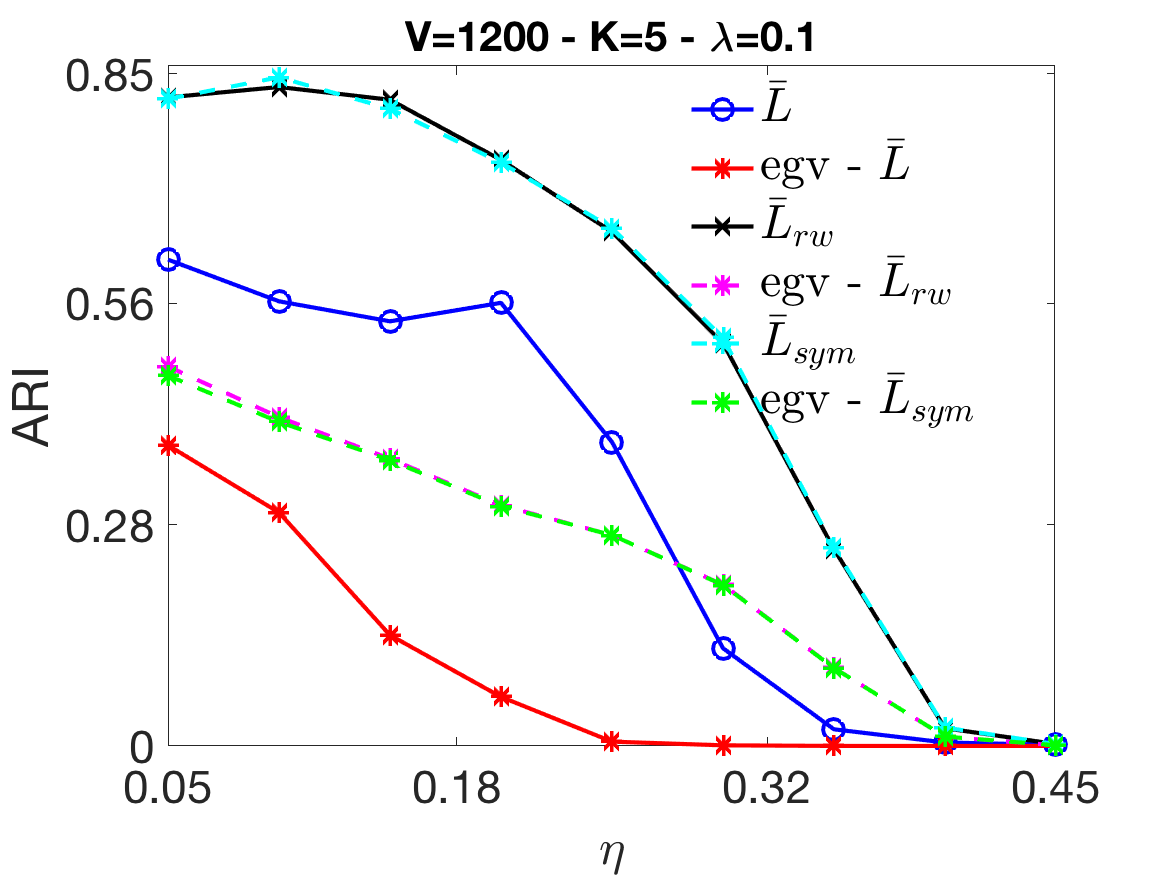} }
\subcaptionbox[]{  
}[ 0.32\textwidth ]
{\includegraphics[width=0.35\textwidth] {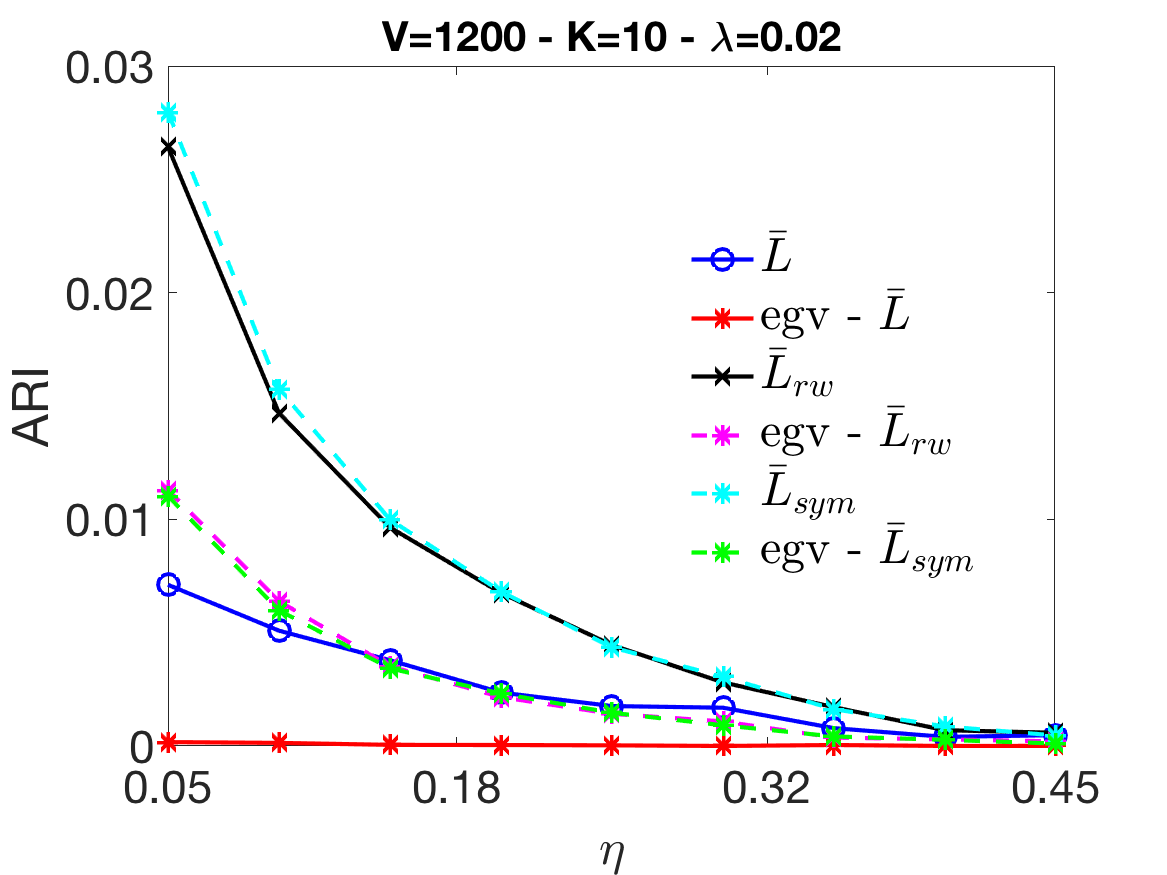} }
%
\subcaptionbox[]{  
}[ 0.32\textwidth ]
{\includegraphics[width=0.35\textwidth] {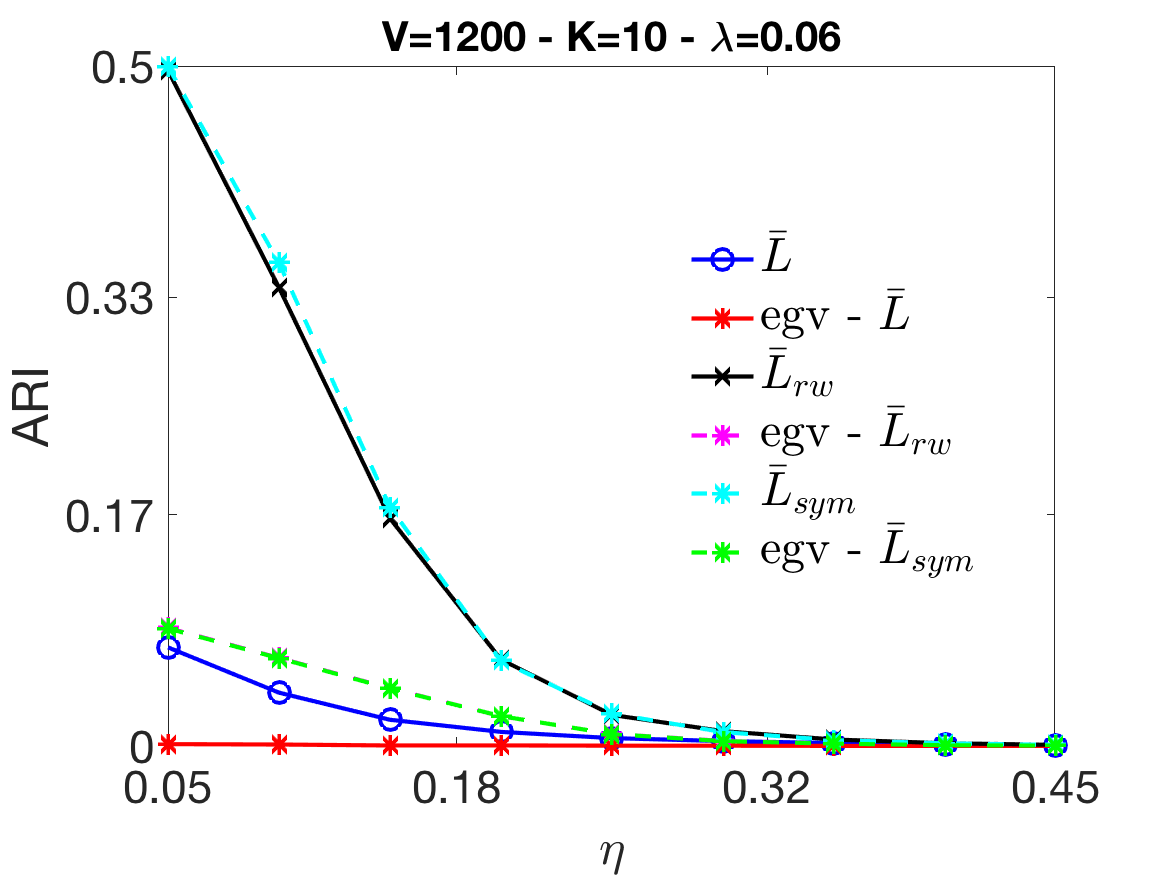} }
%
\subcaptionbox[]{  %
}[ 0.32\textwidth ]
{\includegraphics[width=0.35\textwidth] {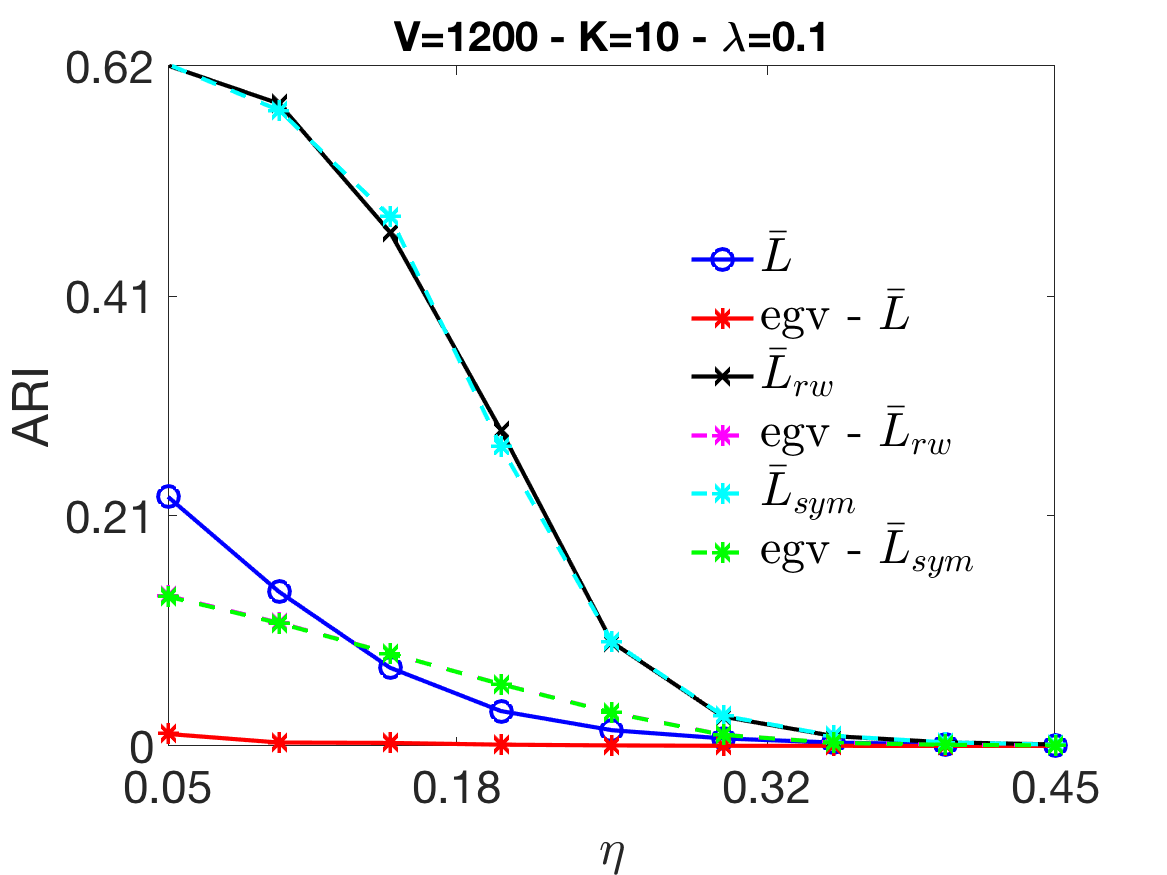} }
\captionsetup{width=0.95\linewidth}
\caption[Short Caption]{ These plots show the capability of MBO to recover the ground truth node-cluster association using the Laplacians \eqref{eq:signLap}, \eqref{eq:signLapRW} and \eqref{eq:signLapSym} at different regimes of noise, sparsity and number of clusters. As explained in the main text in Section~\ref{sec:numericalSection}, the algorithms are initialized by a partition into $K$ sets based on an eigenvector associated with the smallest positive eigenvalue. For this reason, we have included the ARI value of the initial condition to see how MBO improves on its starting point. The first row shows results for $K=2$. As we can see, MBO gives better results than the initial condition for $\bar{L}$, while $\bar{L}_{\text{rw}}$ and $\bar{L}_{\text{sym}}$ do not improve over their eigenvectors but they give better results than $\bar{L}$. For $K=5$ and $K=10$, MBO always gives better results than the initial condition for all values of $\lambda$ and $\eta$.
}
\label{fig:bestLap}
\end{figure}
\begin{figure}[!ht]
\centering
\subcaptionbox[]{  
}[ 0.32\textwidth ]
{\includegraphics[width=0.35\textwidth] {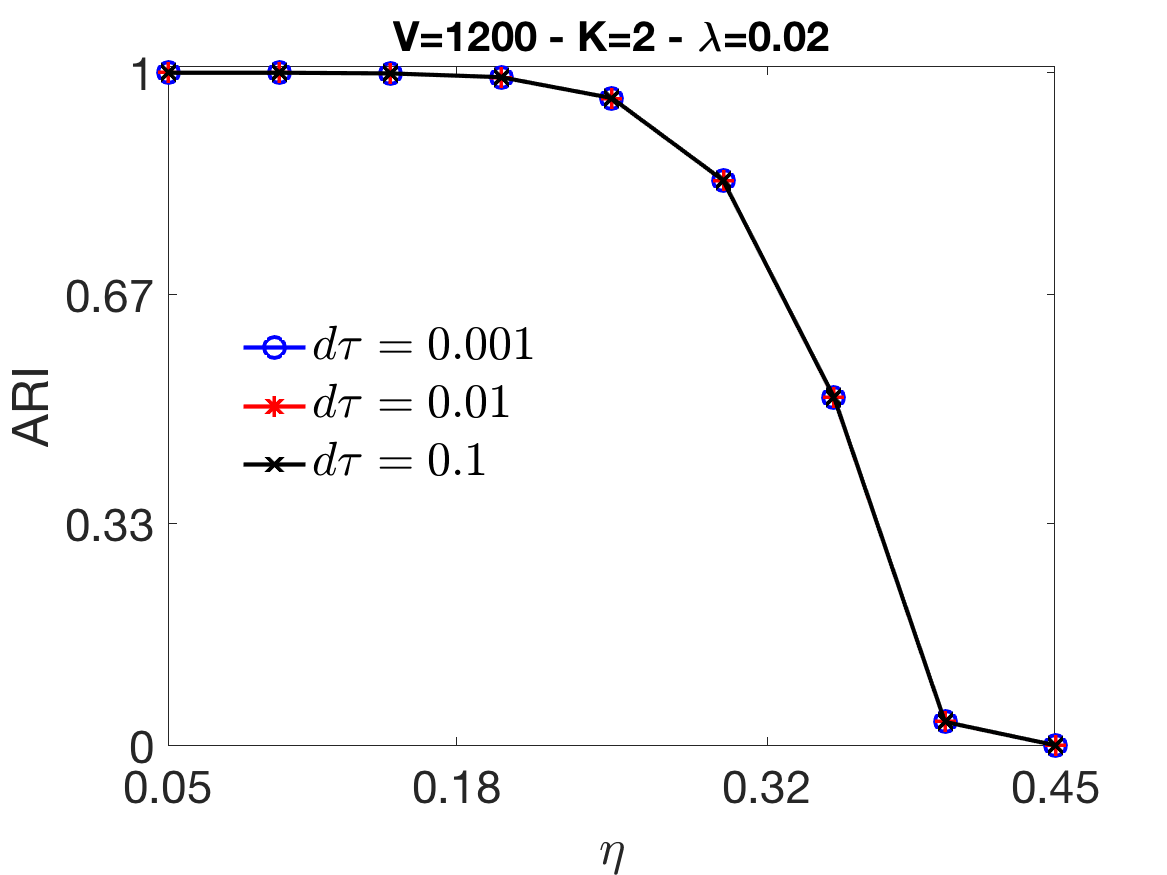} }
%
\subcaptionbox[]{  
}[ 0.32\textwidth ]
{\includegraphics[width=0.35\textwidth] {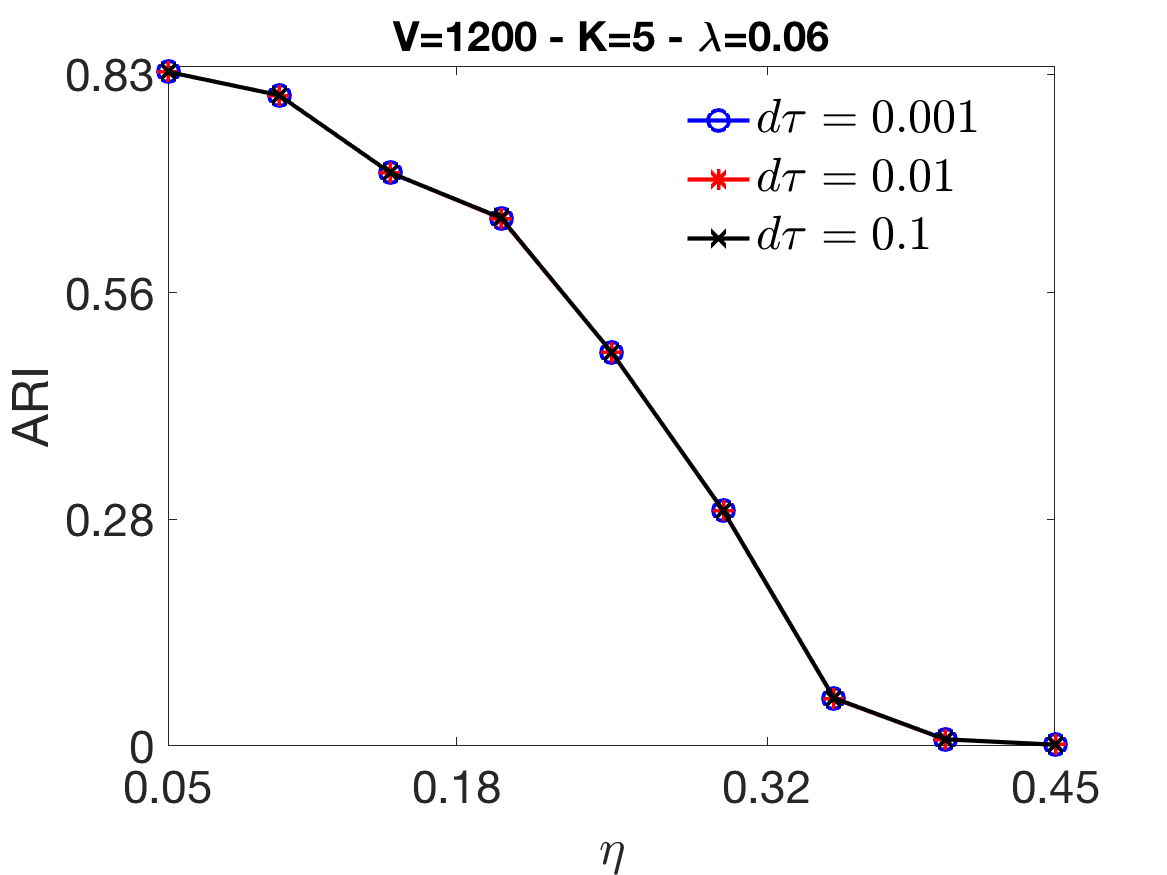} }
%
\subcaptionbox[]{  
}[ 0.32\textwidth ]
{\includegraphics[width=0.35\textwidth] {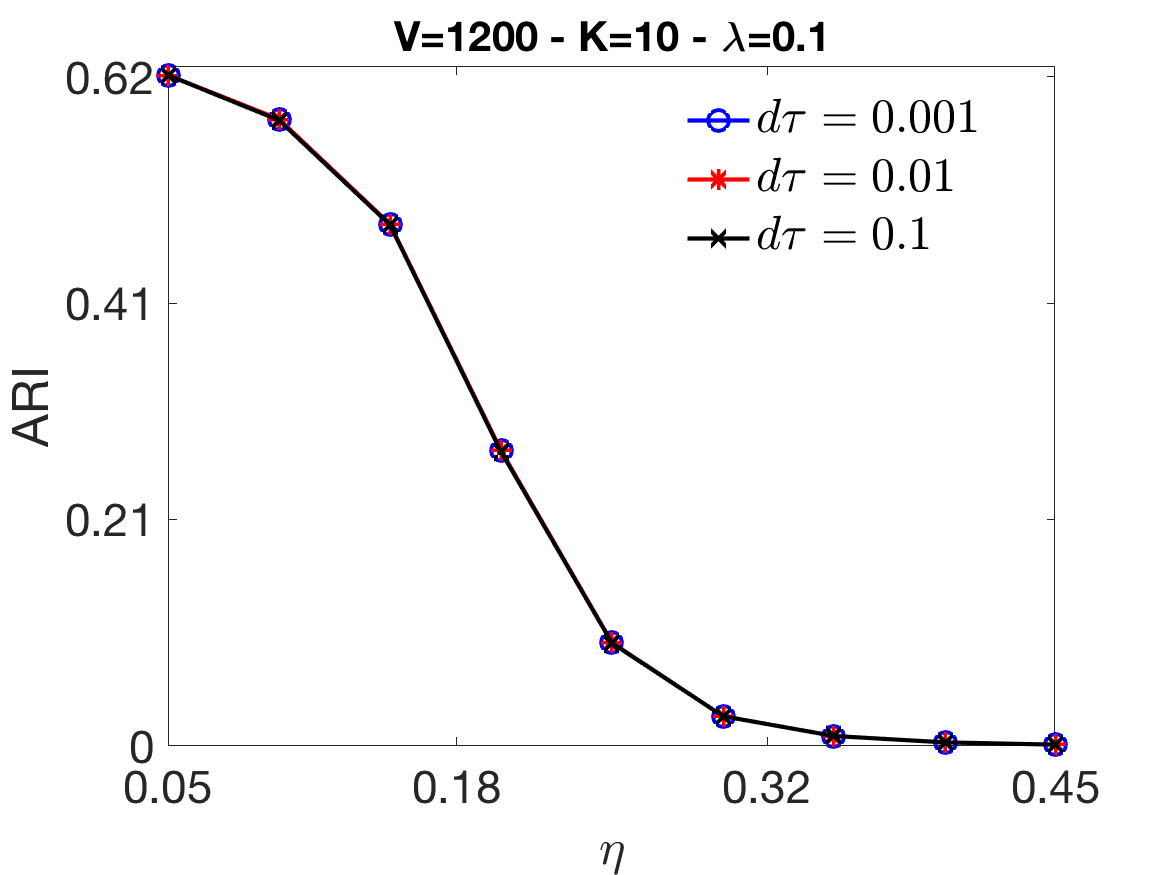} }
%
\captionsetup{width=0.95\linewidth}
\caption[Short Caption]{ The plots show different instances of MBO performed using different choices for the time discretization but with the same Laplacian $L_{\text{sym}}$. Here we can see that there is no difference in the choice of $d\tau$ in the whole range of values considered in Table \ref{tb:onlytab}.
}
\label{fig:bestTau}
\end{figure}
\subsection{Comparison to other methods}
The performance of our MBO algorithm is compared to that of the Balanced Normalized Cut (BNC) \cite{DhillonBalNormCut}, and with the results achieved by standard Kmeans$++$ clustering \cite{kmeans_pp, kmeanspp_algoMAT}. In this respect, as can be seen in Figure \ref{fig:comparison}, MBO always provides better performance than Kmeans$++$,  and improves upon BNC when sparsity and noise are high.
%
\begin{figure}[!ht]
\centering
\subcaptionbox[]{  
}[ 0.32\textwidth ]
{\includegraphics[width=0.35\textwidth] {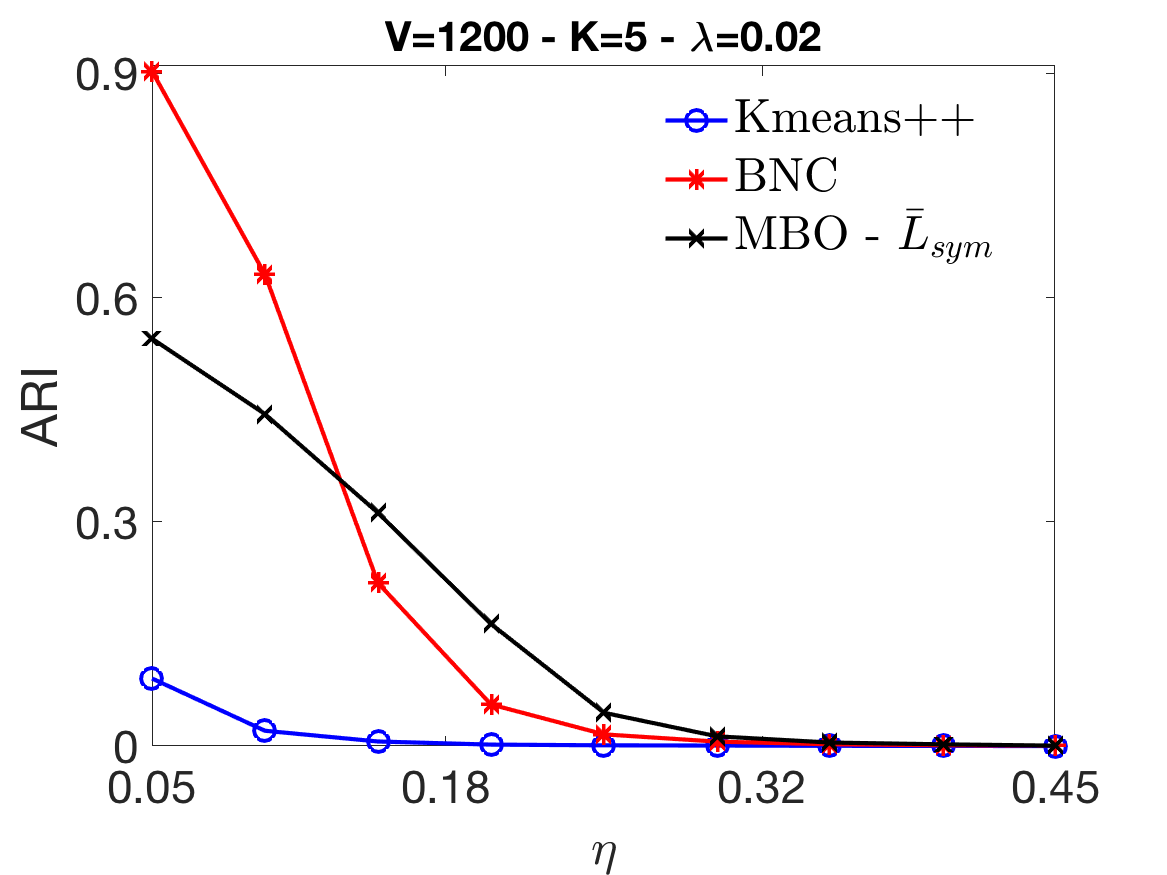} }
%
\subcaptionbox[]{  
}[ 0.32\textwidth ]
{\includegraphics[width=0.35\textwidth] {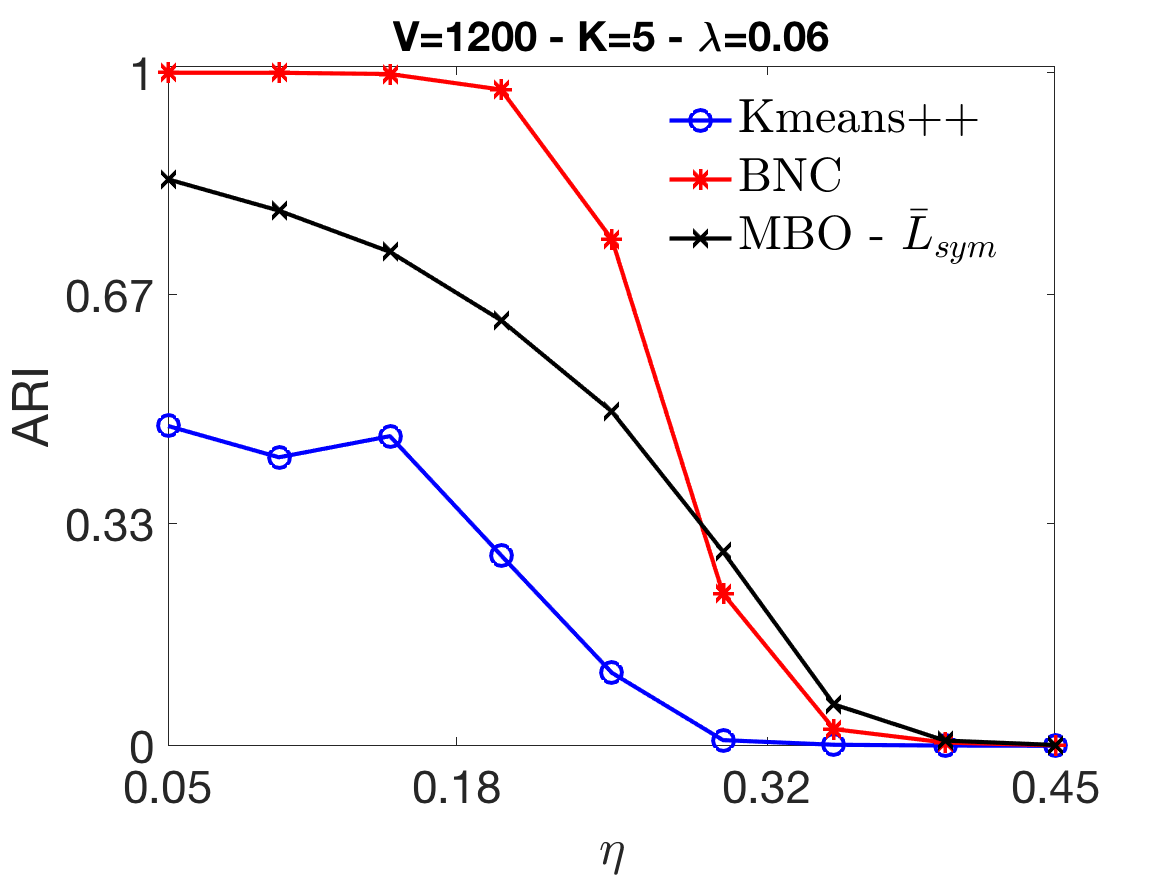} }
%
\subcaptionbox[]{  
}[ 0.32\textwidth ]
{\includegraphics[width=0.35\textwidth] {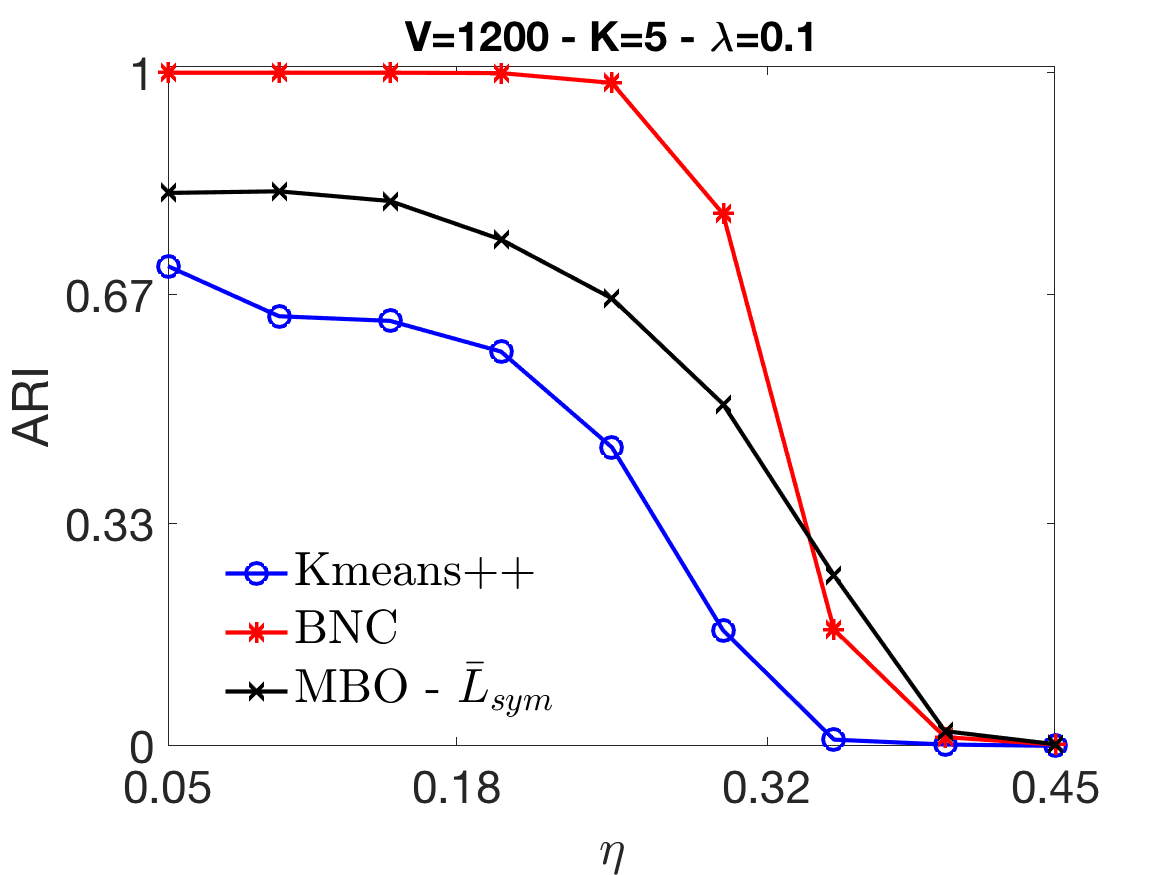} }
\subcaptionbox[]{  
}[ 0.32\textwidth ]
{\includegraphics[width=0.35\textwidth] {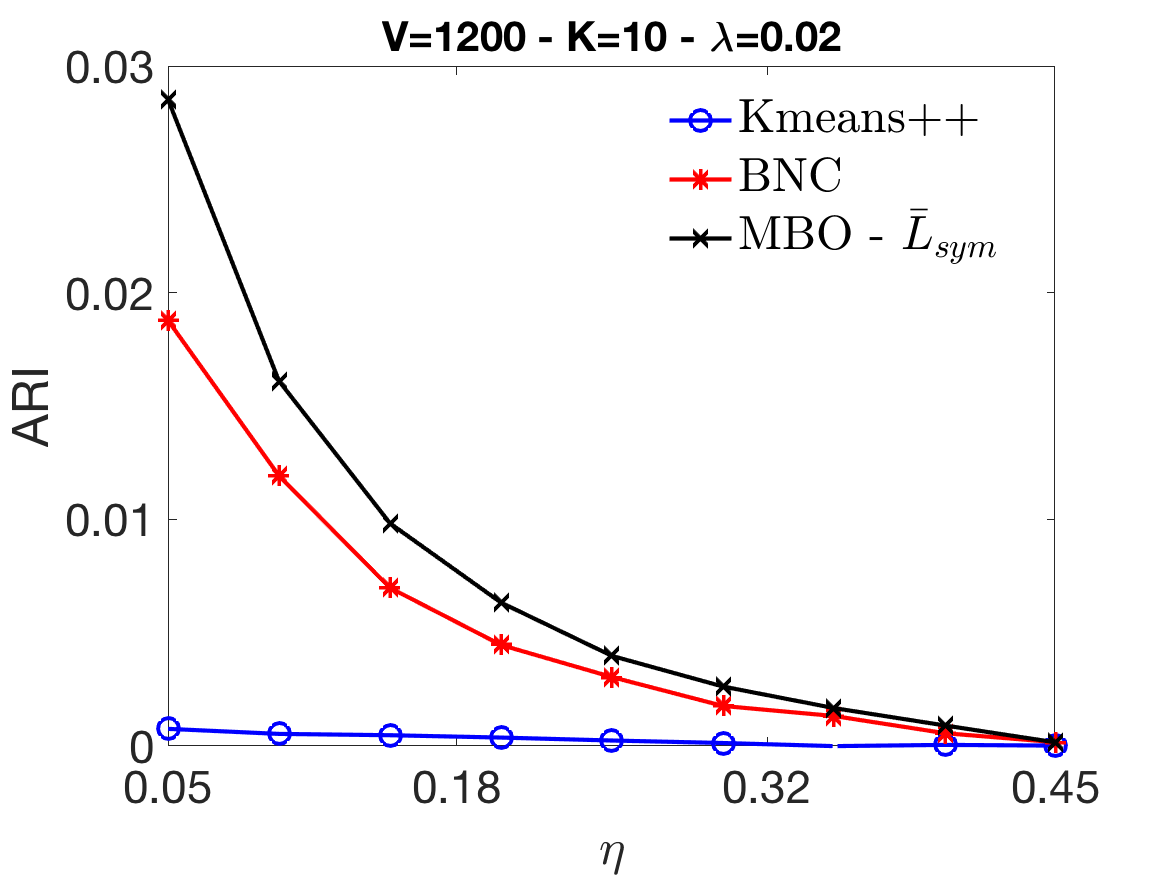} }
%
\subcaptionbox[]{  
}[ 0.32\textwidth ]
{\includegraphics[width=0.35\textwidth] {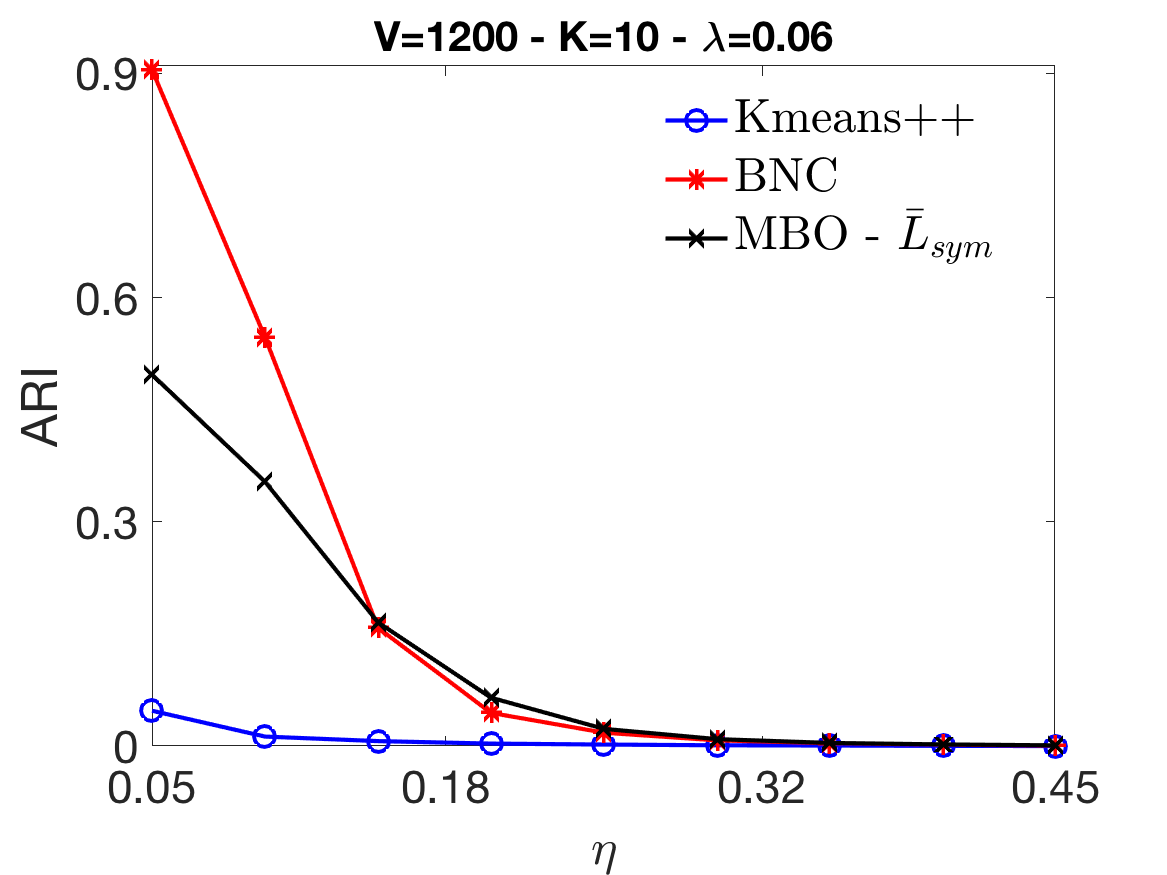} }
%
\subcaptionbox[]{  
}[ 0.32\textwidth ]
{\includegraphics[width=0.35\textwidth] {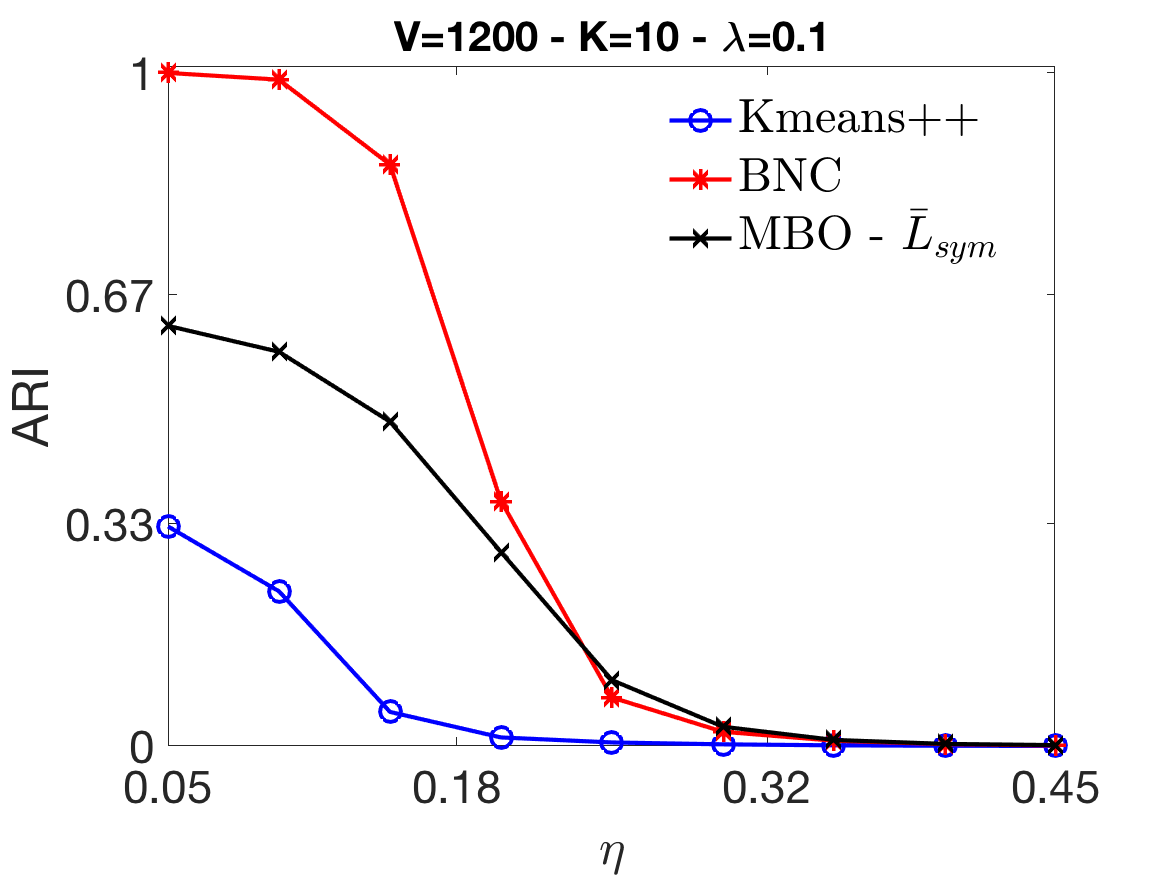} }
%
\captionsetup{width=0.95\linewidth}
\caption[Short Caption]{We have compared the node-cluster association recovery ability of MBO for signed Laplacians with two popular methods from the literature, Kmeans$++$ \cite{kmeans_pp} and BNC \cite{DhillonBalNormCut}. From the plots, it is possible to see that Kmeans$++$ always performs worse than the other methods. MBO delivers better results than BNC for high noise level in the sparsity regime considered.
}
\label{fig:comparison}
\end{figure}

\subsection{Must-link and cannot-link}
In this section, we detail how MBO can be applied in the context of constrained clustering. We have run the MBO algorithm for $\mathcal{L}=\mathcal{L}^+$, corresponding to the must-link scenario, and $\mathcal{L}=\bar{L}+\mathcal{Q}^-$, corresponding to the case in which only cannot-links are included as additional information in the definition of the Laplacian. In practice, we have considered the symmetric version of $\mathcal{L}$ in the same fashion as in \eqref{eq:signLapSym}. In the plots shown in Figure \ref{fig:constclust},  we see that, for high sparsity and noise, the constrained clustering improves upon the standard signed Laplacian in both cases. In addition, the two approaches can be combined to obtain better results. In this case, we have included the known information up to $\alpha=0.1$ to not 
modify considerably  the sparsity of the affinity matrix.
\begin{figure}[!ht]
\centering
\subcaptionbox[]{  
}[ 0.32\textwidth ]
{\includegraphics[width=0.35\textwidth] {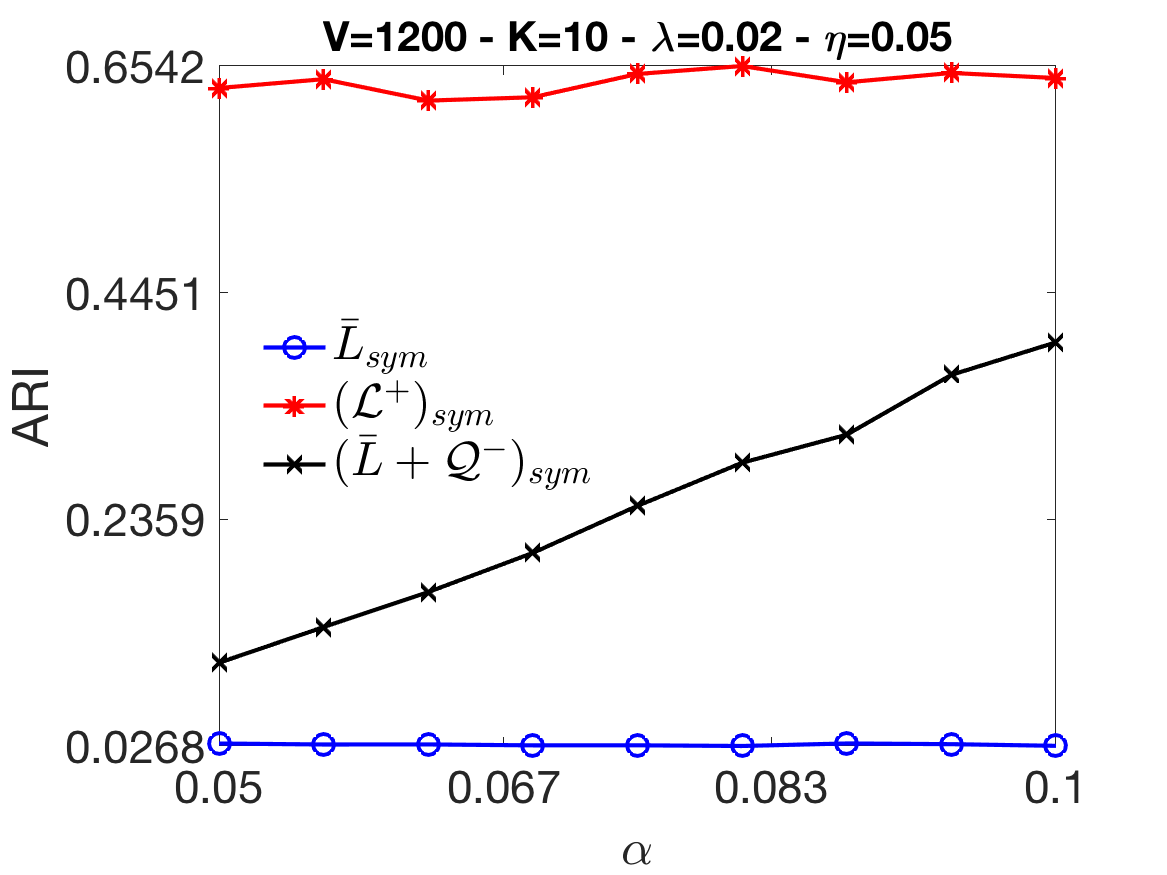} }
%
\subcaptionbox[]{  
}[ 0.32\textwidth ]
{\includegraphics[width=0.35\textwidth] {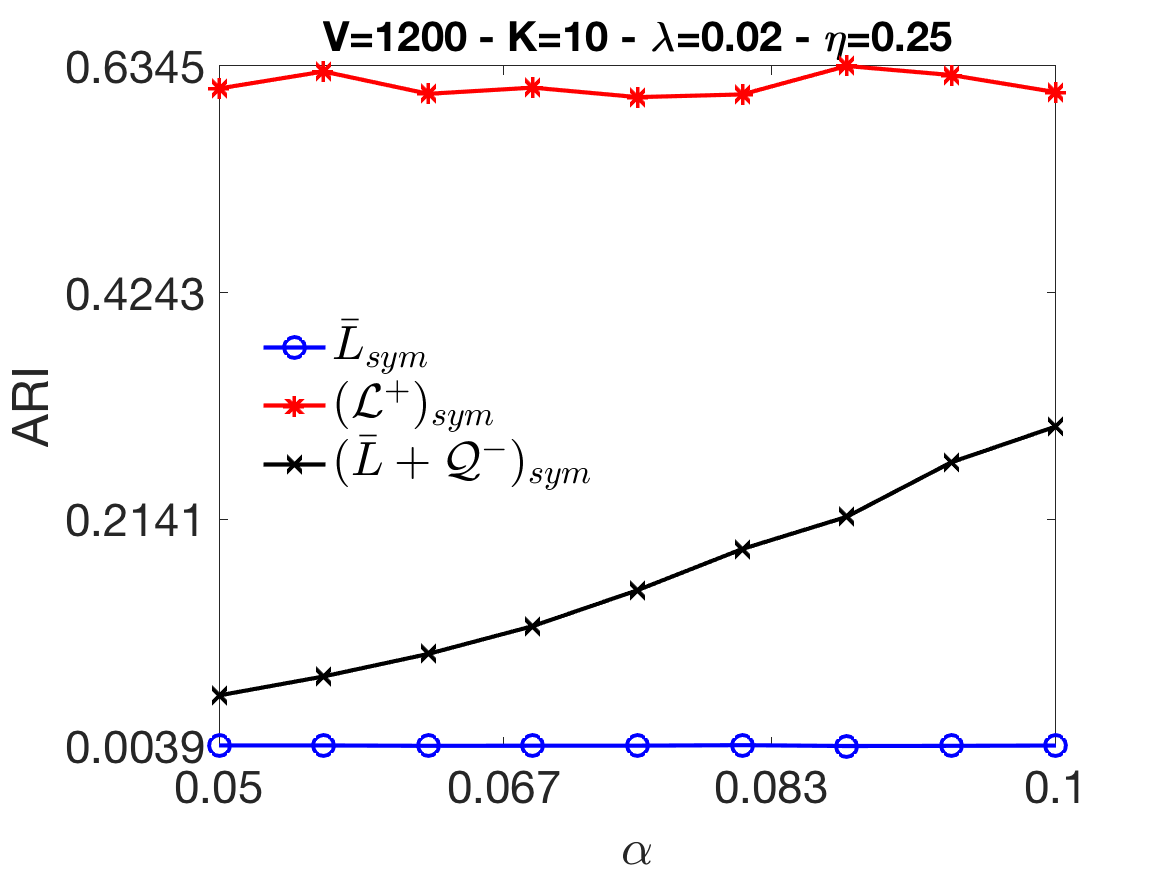} }
%
\subcaptionbox[]{  
}[ 0.32\textwidth ]
{\includegraphics[width=0.35\textwidth] {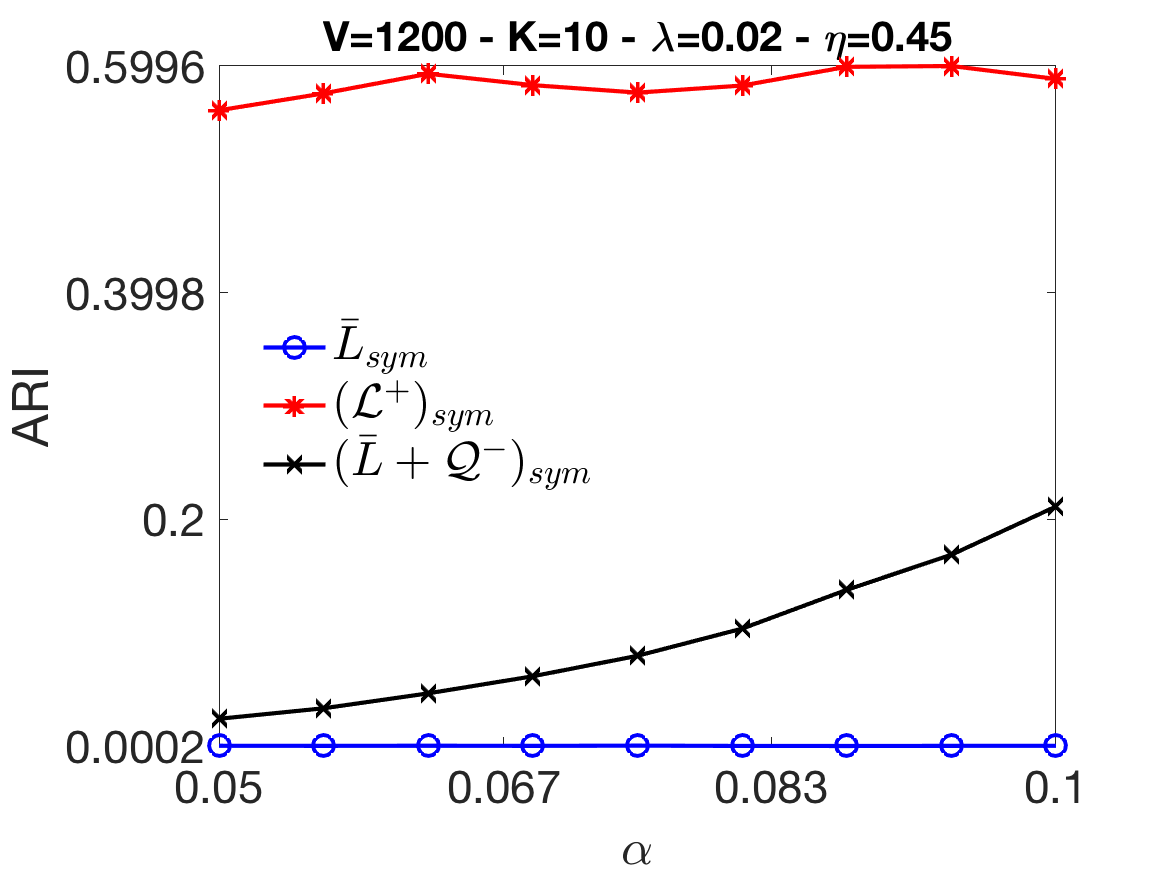} }
\subcaptionbox[]{  
}[ 0.32\textwidth ]
{\includegraphics[width=0.35\textwidth] {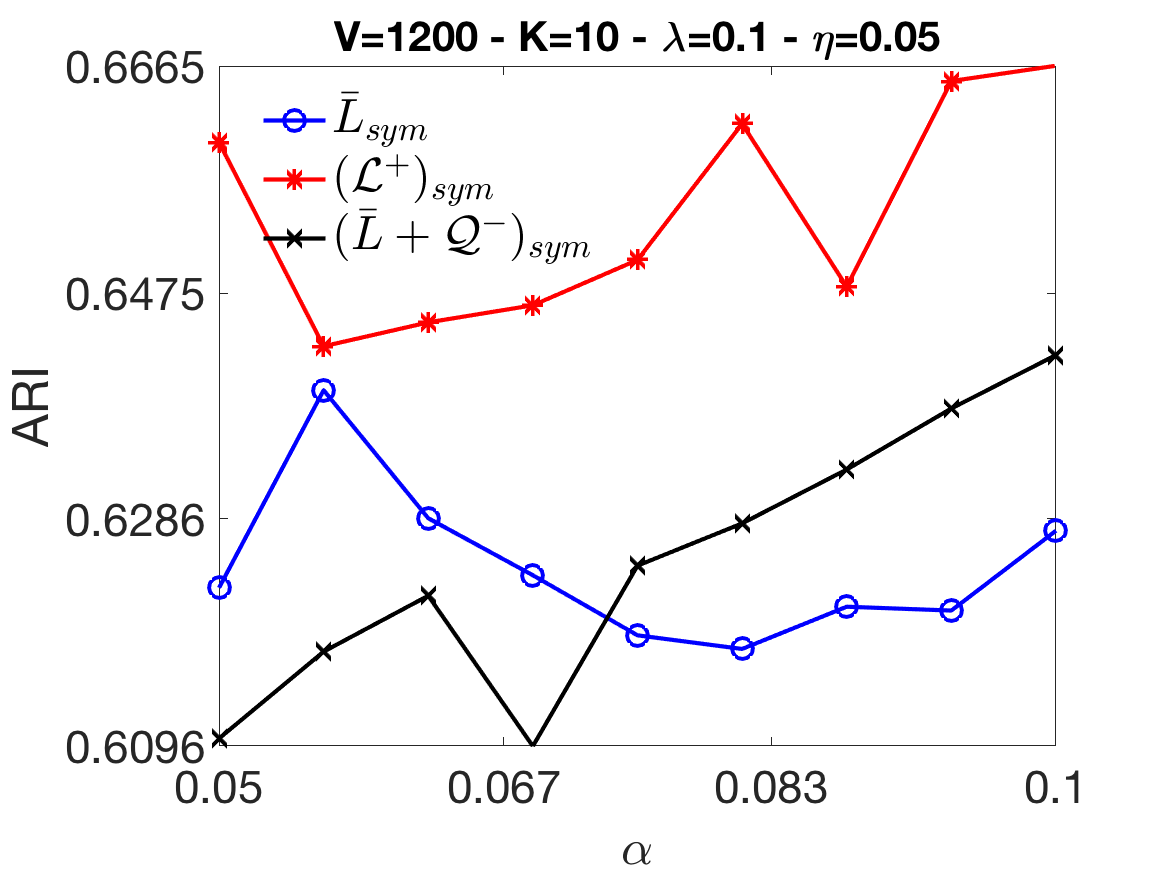} }
%
\subcaptionbox[]{  
}[ 0.32\textwidth ]
{\includegraphics[width=0.35\textwidth] {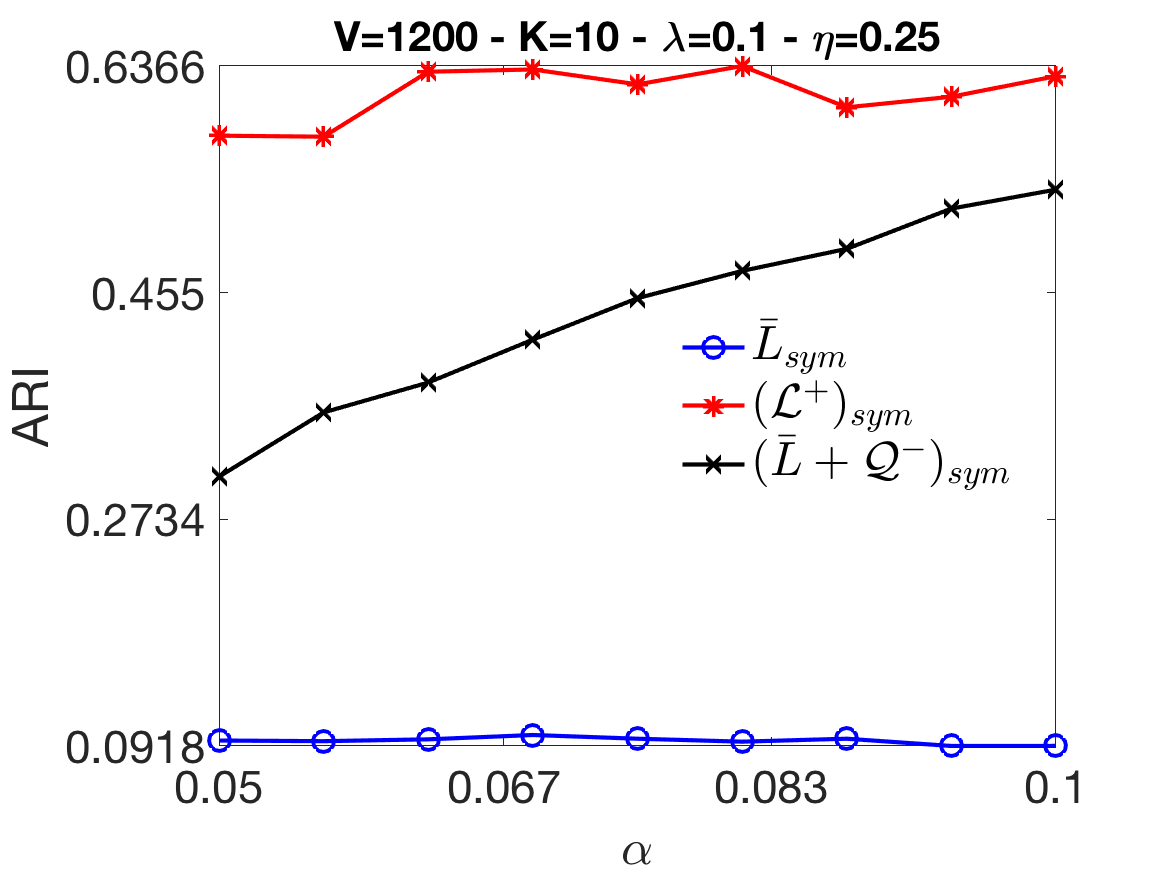} }
%
\subcaptionbox[]{  
}[ 0.32\textwidth ]
{\includegraphics[width=0.35\textwidth] {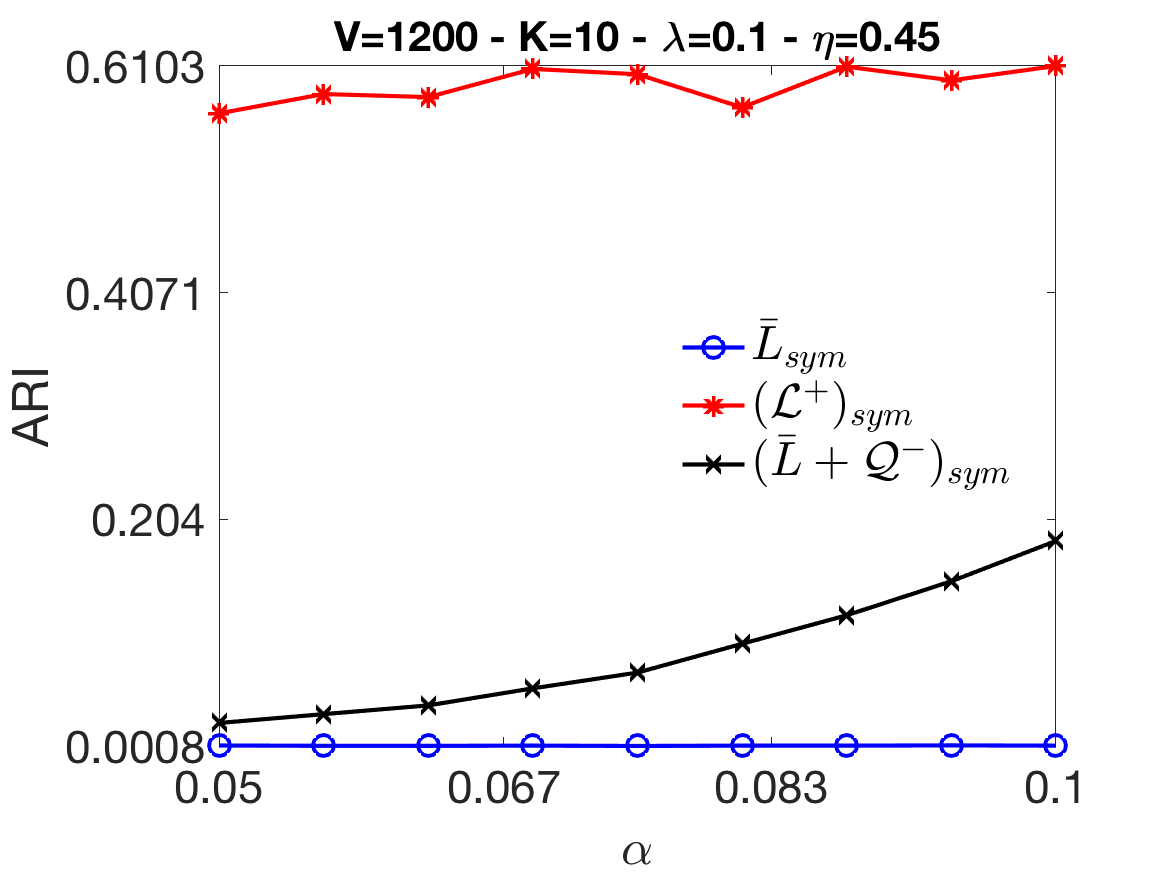} }
%
\captionsetup{width=0.95\linewidth}
\caption[Short Caption]{We compare MBO run over an instance of the affinity matrix generated from \eqref{eq:genaffmatrix}, against the same matrix where additional information extracted from the ground truth has been included in the form of must-links and cannot-links. This corresponds  to adding soft constraints to the system in the form of edge weights. The parameter $\alpha$ is the fraction of links of the ground truth used to define the must- and cannot-link matrices. In this case, we have considered $\alpha$ up to $0.1$ to avoid large variations in the sparsity of the matrix compared to the original one. For each Laplacian, we have considered its symmetric version  \eqref{eq:signLapSym}. The algorithms are initialized using different eigenvectors since they are specified by different Laplacians. At any value of $\alpha$, adding must-links and cannot-links significantly improves the recovery of the regular MBO for most sparsity and noise levels.
}
\label{fig:constclust}
\end{figure}

\subsection{Fidelity and avoidance terms}
The second soft constraint we include in the MBO scheme is at the level of the characteristic matrix. The fidelity term pushes the characteristic matrix to be similar to $\hat{U}$ and increases the energy for every row $U$ which is different from $\hat{U}$. Instead, for the avoidance term the energy of the GL functional is increased only when there is an exact match between two rows of $U$ and $\tilde{U}$. This means that, in general, the effect of the avoidance term will be less influential since it provides an energy increase fewer times. This can be seen in Figure \ref{fig:fidelavoid}, where the black curve is always below  the red one. We add that, in order to make the effect of the fidelity and avoidance terms more evident, we fix the nonzero diagonal elements of  
the matrices $R$ and $Q$ in \eqref{eq:Uhalf} equal to $30$.  

\begin{figure}[!ht]
\centering
\subcaptionbox[]{  
}[ 0.32\textwidth ]
{\includegraphics[width=0.35\textwidth] {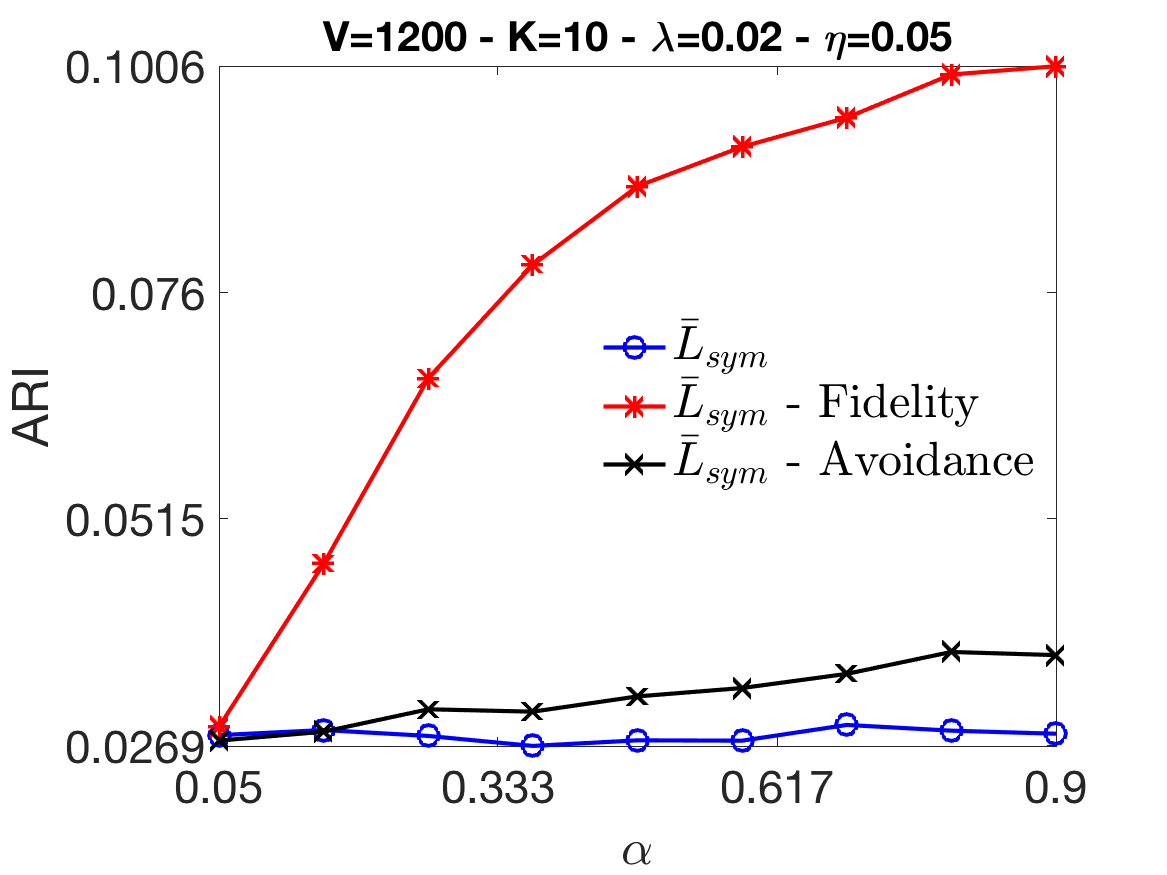} }
%
\subcaptionbox[]{  
}[ 0.32\textwidth ]
{\includegraphics[width=0.35\textwidth] {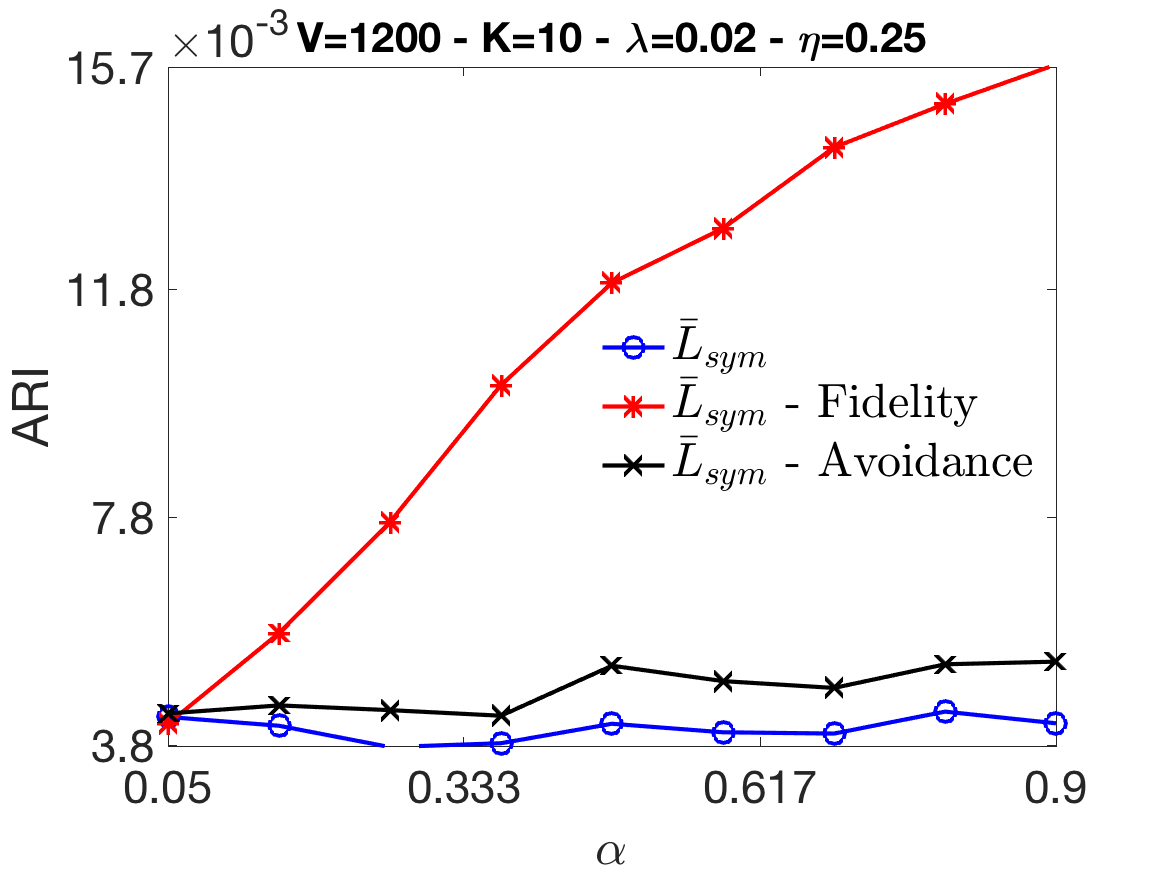} }
%
\subcaptionbox[]{  
}[ 0.32\textwidth ]
{\includegraphics[width=0.35\textwidth] {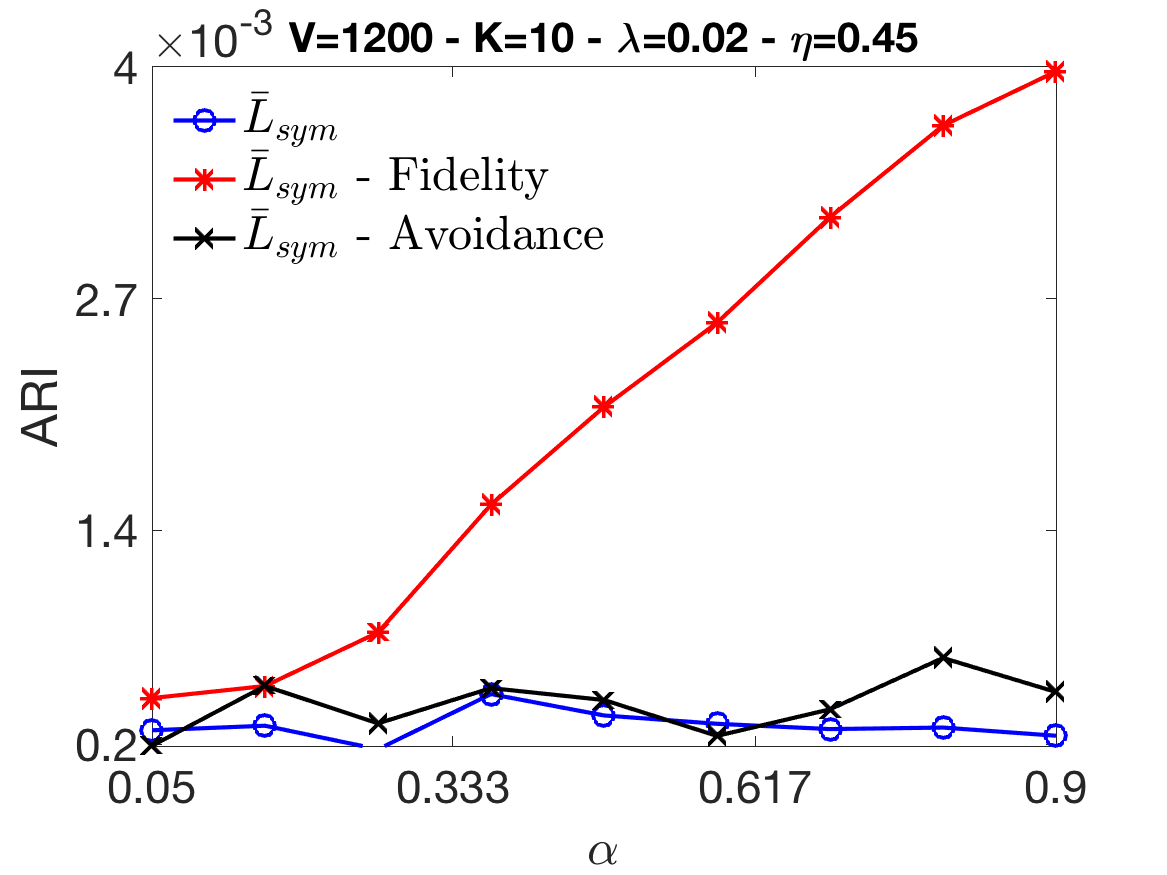} }
\subcaptionbox[]{  
}[ 0.32\textwidth ]
{\includegraphics[width=0.35\textwidth] {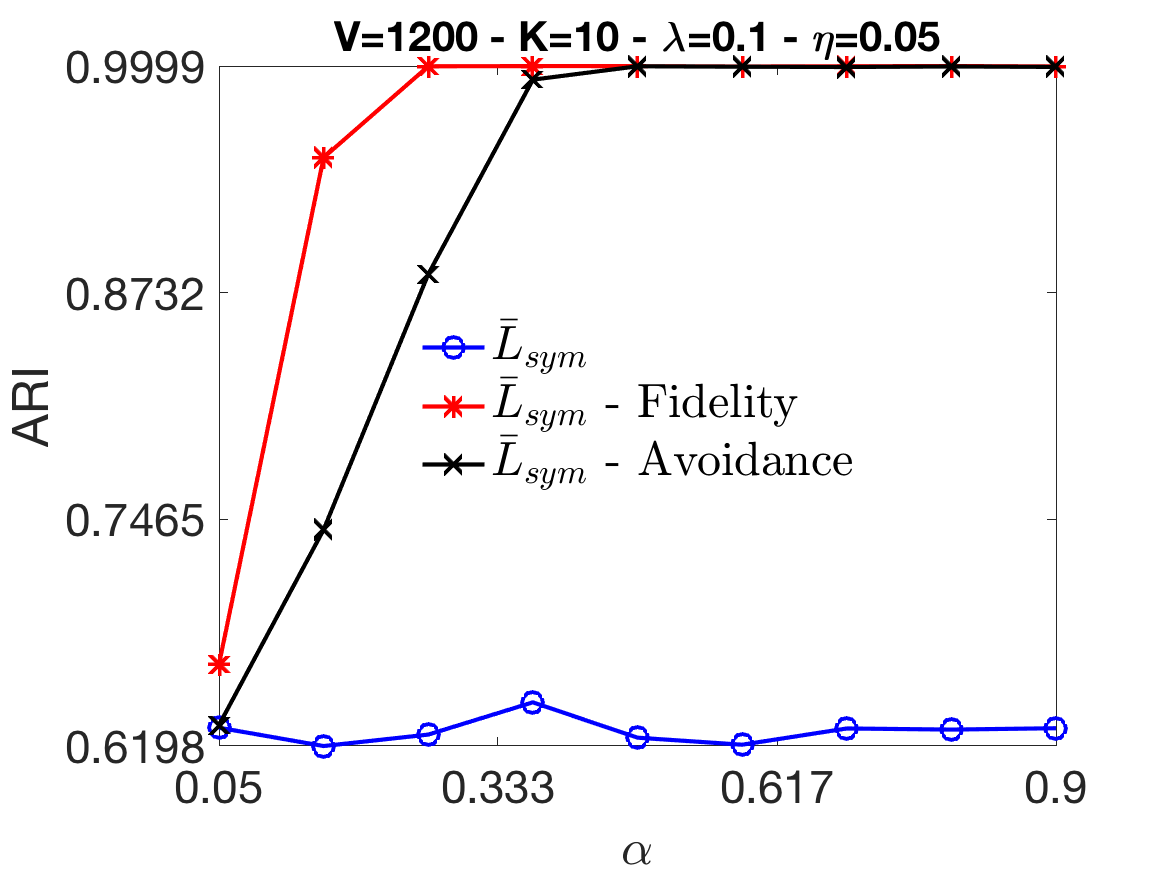} }
%
\subcaptionbox[]{  
}[ 0.32\textwidth ]
{\includegraphics[width=0.35\textwidth] {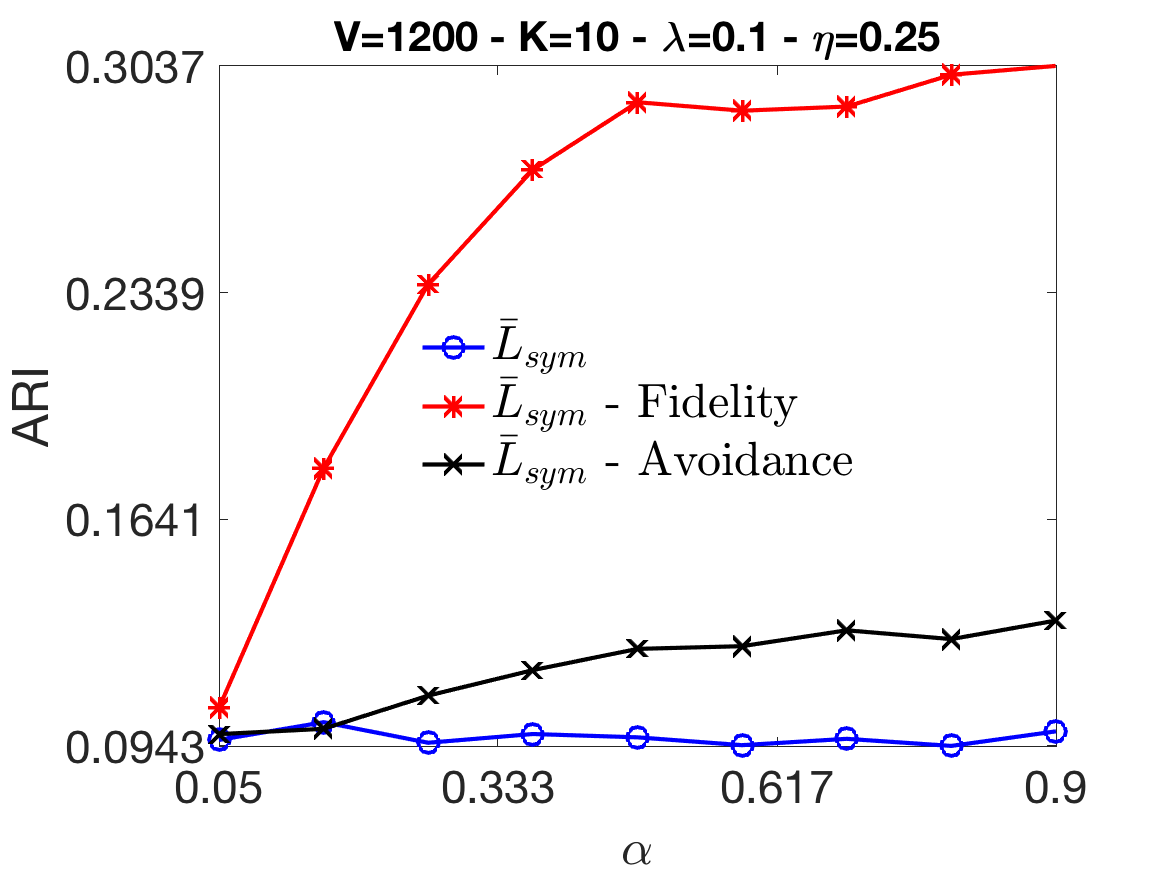} }
%
\subcaptionbox[]{  
}[ 0.32\textwidth ]
{\includegraphics[width=0.35\textwidth] {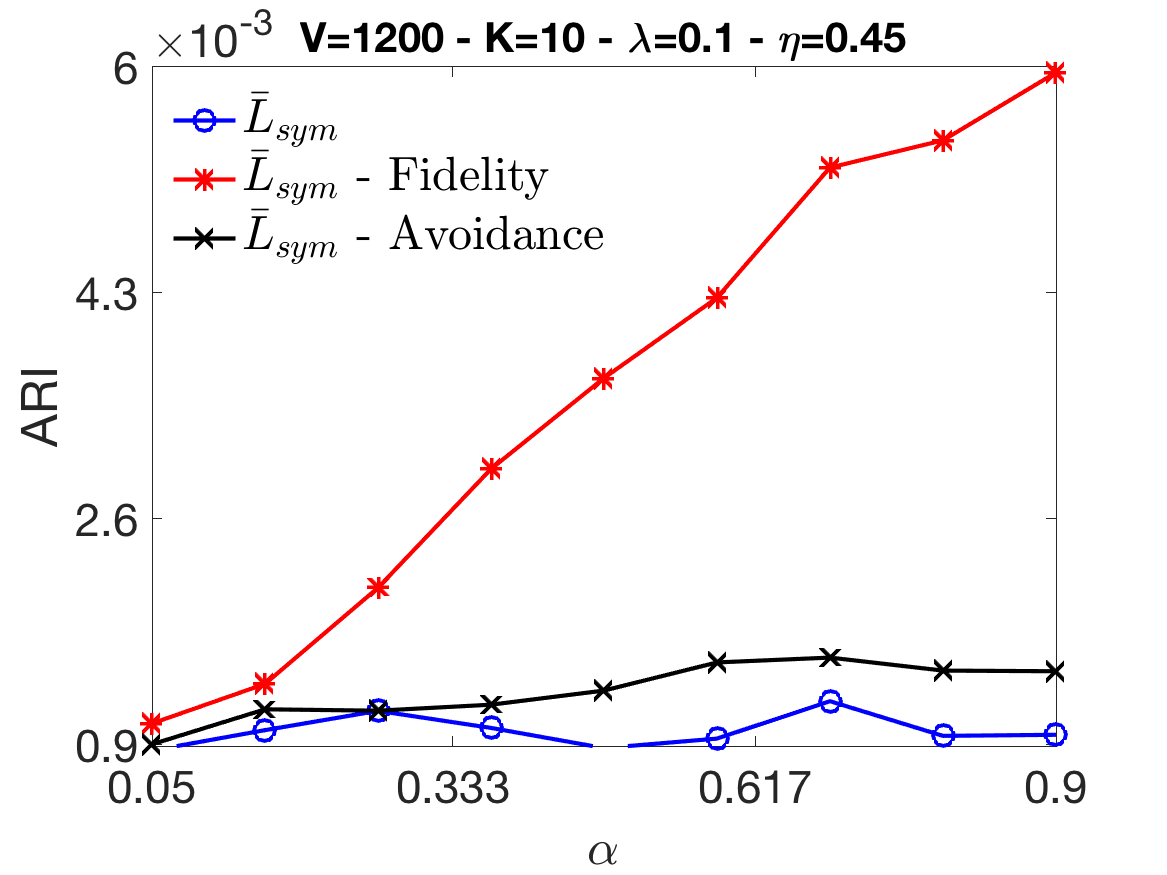} }
%
\captionsetup{width=0.95\linewidth}
\caption[Short Caption]{We consider two soft constraints at the level of the characteristic matrix $U$ as expressed in \eqref{eq:Uhalf}. The fidelity/avoidance constraints are extracted at random from the ground truth for each experiment. For the fidelity/avoidance case, the relevant entries of the initial characteristic matrix are modified to be consistent with the constraints. As observed, adding the two constraints improves the results compared to those of the regular MBO procedure for all values of $\alpha$. We remind that the matrices $R$ and $Q$ in \eqref{eq:Uhalf} are multiplied by a coefficient equal to $30$ in order to magnify the effect of the constraint (without varying the relative order of position of the three curves).
}
\label{fig:fidelavoid}
\end{figure}

\subsection{Anchors}
In this setting, we are introducing in the MBO scheme a hard constraint, as we need to  a priori modify the matrix responsible for the evolution of $U$, as detailed in Section \ref{sc:Anchors}. As shown in Figure \ref{fig:anchors}, anchors always improve upon the unconstrained signed Laplacian, as expected.
\begin{figure}[!ht]
\centering
\subcaptionbox[]{  
}[ 0.32\textwidth ]
{\includegraphics[width=0.35\textwidth] {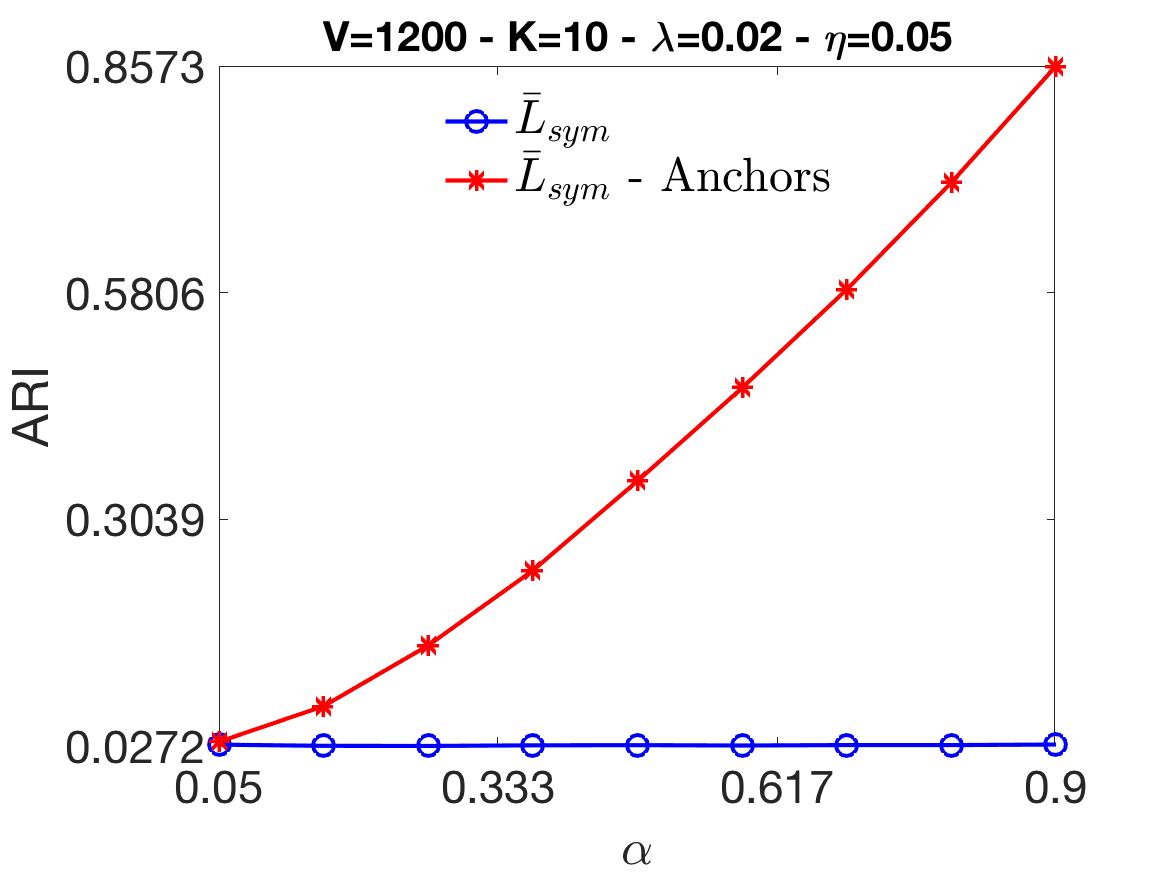} }
%
\subcaptionbox[]{  
}[ 0.32\textwidth ]
{\includegraphics[width=0.35\textwidth] {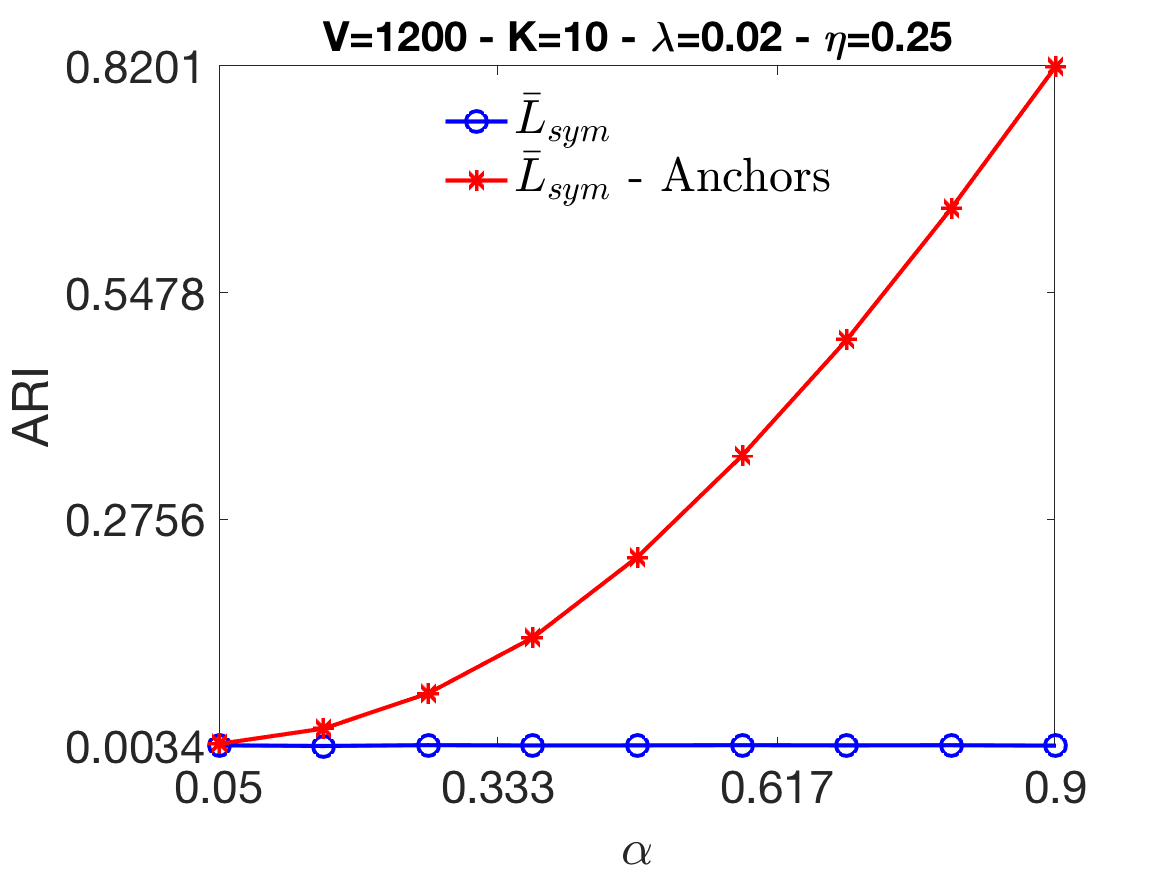} }
%
\subcaptionbox[]{  
}[ 0.32\textwidth ]
{\includegraphics[width=0.35\textwidth] {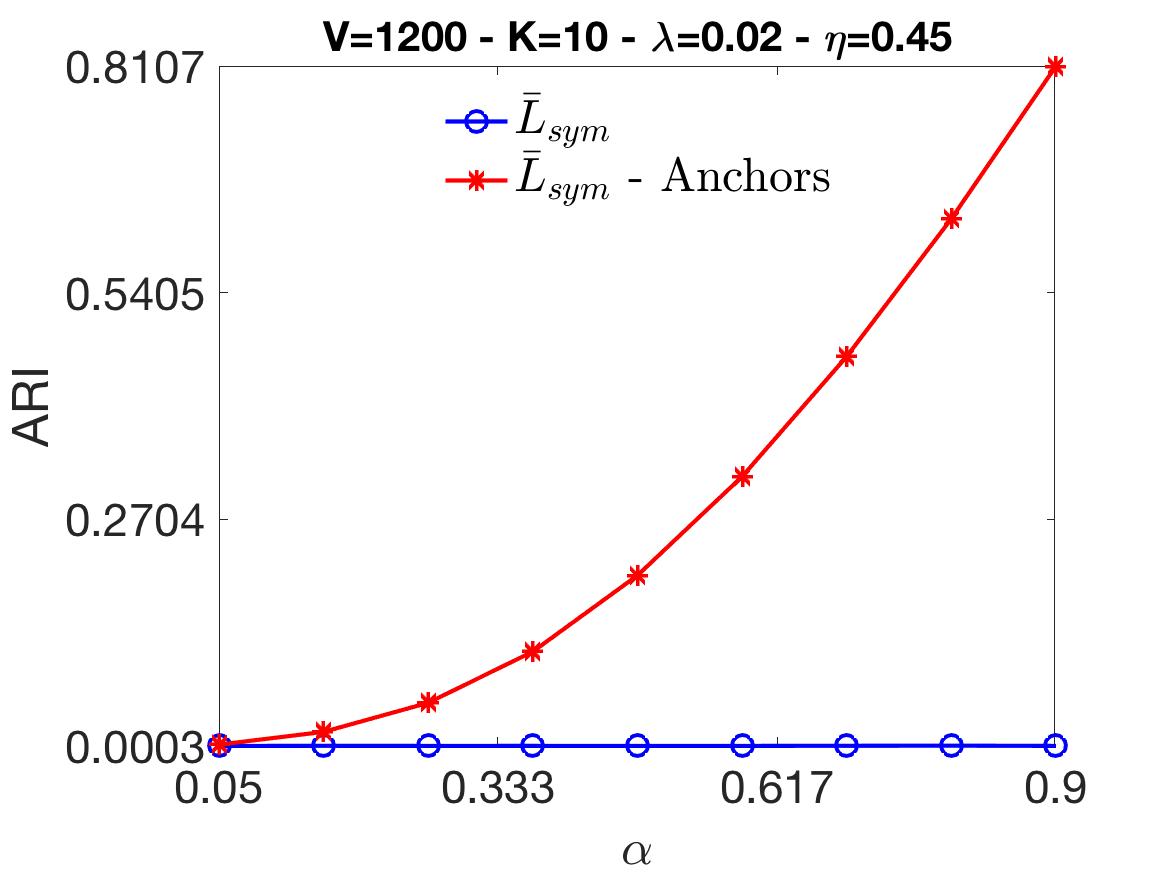} }
\subcaptionbox[]{  
}[ 0.32\textwidth ]
{\includegraphics[width=0.35\textwidth] {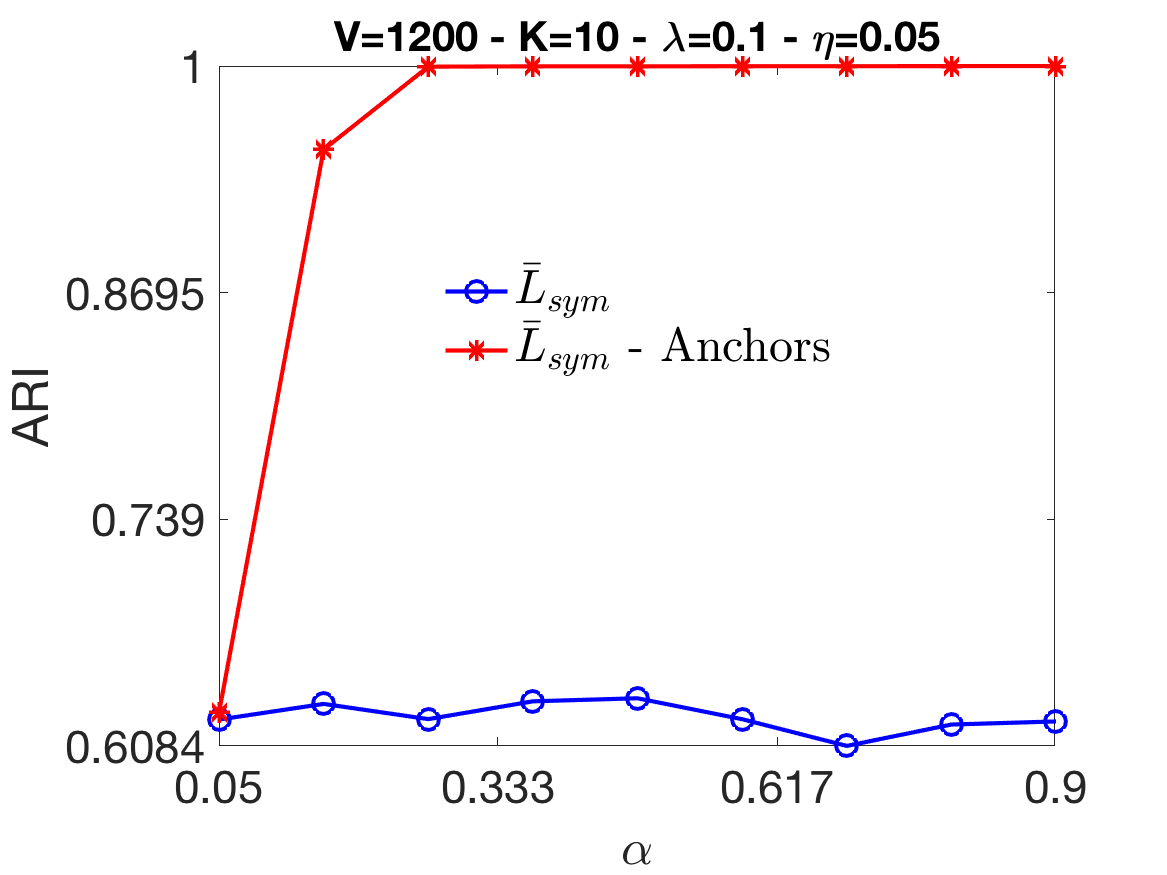} }
%
\subcaptionbox[]{  
}[ 0.32\textwidth ]
{\includegraphics[width=0.35\textwidth] {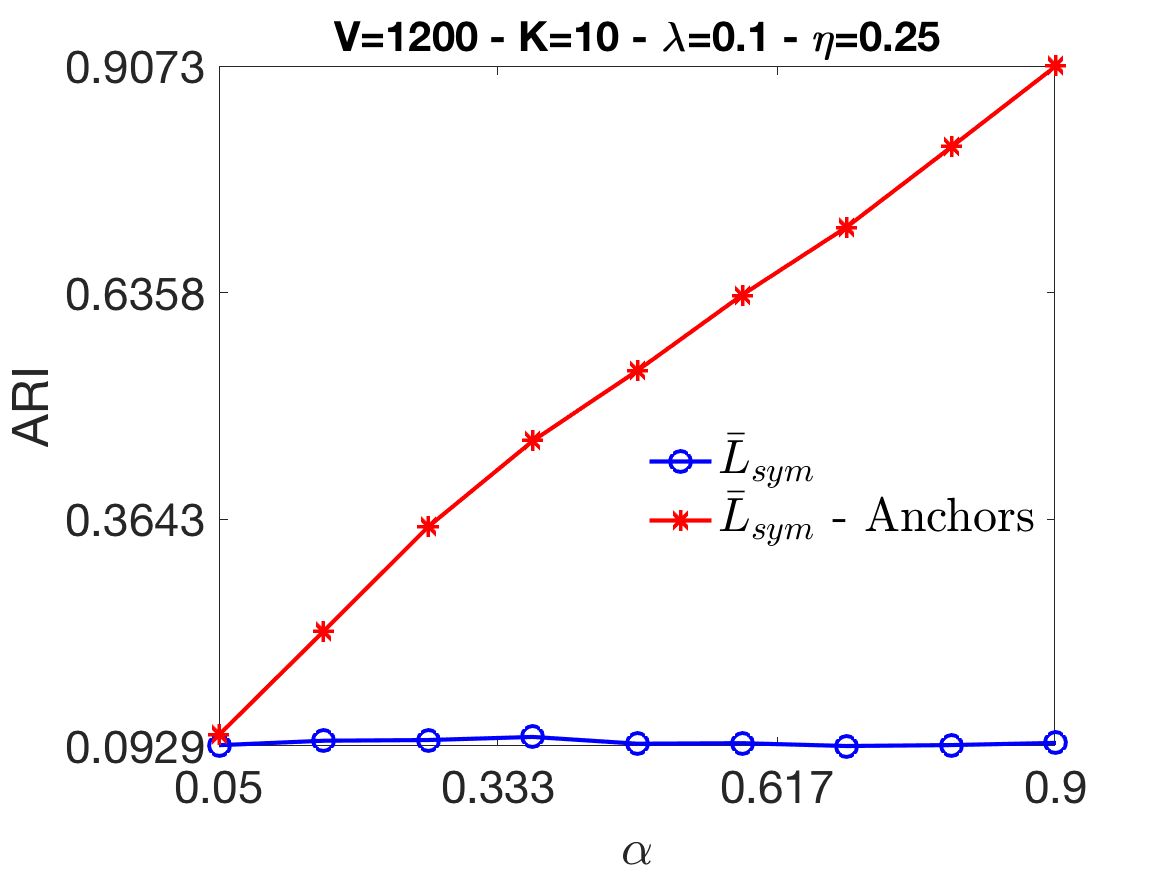} }
%
\subcaptionbox[]{  
}[ 0.32\textwidth ]
{\includegraphics[width=0.35\textwidth] {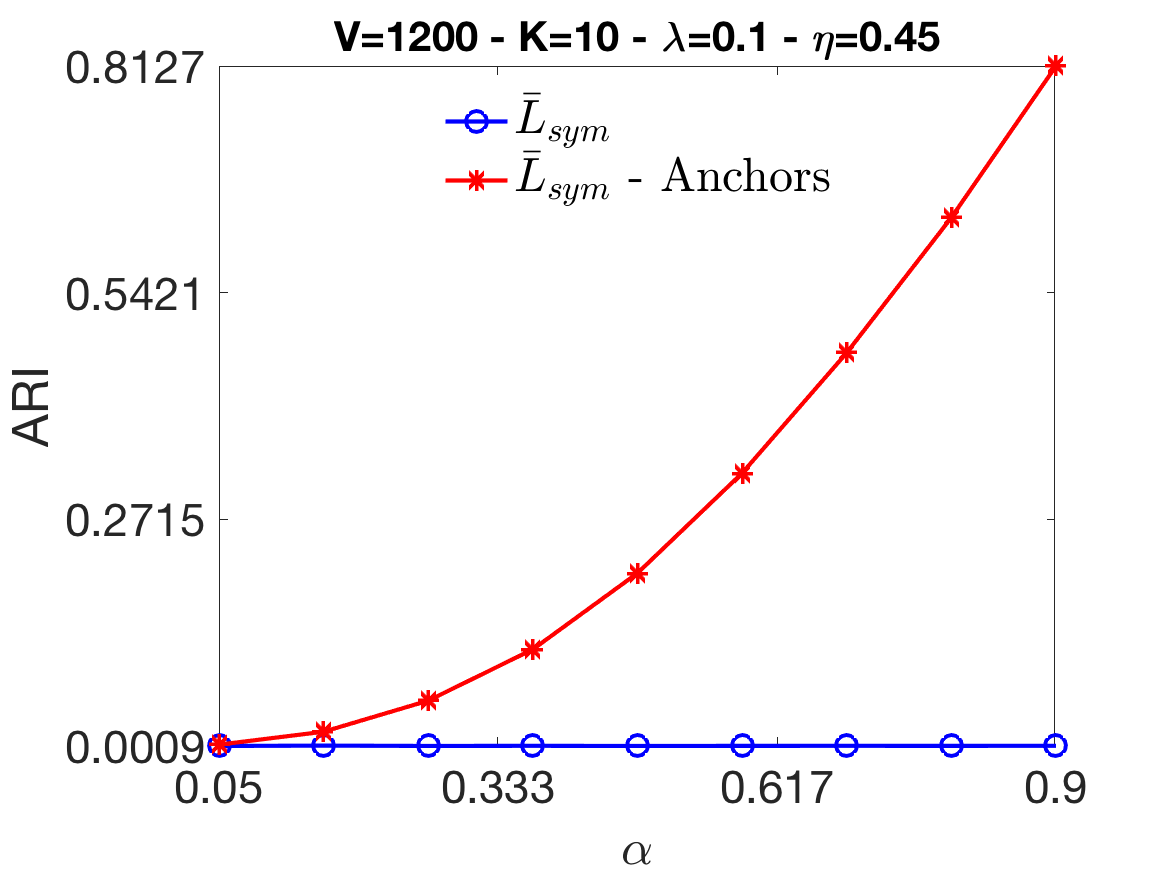} }
%
\captionsetup{width=0.95\linewidth}
\caption[Short Caption]{ The plots show the behaviour of the MBO scheme in the presence of hard constraints in the form of anchors as specified in Section \ref{sc:Anchors}. The anchors are extracted at random from the ground truth for each experiment. For the anchor case, the relevant entries of the initial characteristic matrix are modified to be consistent with the constraints. As expected, for any value of $\alpha$, adding fixed points in the dynamics improves the recovery of the regular MBO for most sparsity and noise levels.
}
\label{fig:anchors}
\end{figure}

%

%
%

\section{Numerical experiments for the Barab\'asi--Albert model}   \label{sec:expBAmodel}
In this section, we apply the MBO scheme to signed networks for which the underlying  measurement graph (indicating the presence or absence of an edge) is generated from the \textit{Barab\'asi--Albert} (BA) model \cite{Barabasi99emergenceScaling}. In a nutshell, the algorithm to generate the BA network begins with a connected network with $V_0$ nodes. At each time step, a new node is added to the graph and is connected to $\nu\leqslant V_0$ pre-existing nodes. The probability that a link from the new node connects to a generic node $i$  is given by 
\begin{equation}
p_i=\frac{D_{ii}}{\sum_j D_{jj}}
\end{equation}
where the sum indexed by $j$ runs over all pre-existing nodes. We call $H\in\left\{ 0,1\right\} ^{V\times V}$ the resulting BA network which can be used to build the relevant signed affinity matrix in the following way. As before, consider a row index $i=1,\dotsc,V$ and a column index $j=i,\dotsc,V$ (the $j$ index starts from $i$ so to consider only the upper triangular part in the definition of the elements of $A$), we define the signed BA network as
\begin{equation}\label{eq:baAffMat}
A_{ij}= \left\{
     \begin{array}{rll}
   S_{ij}H_{ij}  & \;\; \text{ if } i< j,  	& \text{with probability } \left(1-\eta \right) 	\\
 - S_{ij}H_{ij}  & \;\; \text{ if } i< j, 	& \text{with probability }  \eta   	\\
		     		     0   & \;\; \text{ if } i= j, 	& \text{with probability } 1	\\
     \end{array}
   \right. \quad\text{ with } A_{ji}=A_{ij},
\end{equation}
where the sparsity of the BA network is included via $H_{ij}$ (the noise probability $\eta$ is understood as in \eqref{eq:genaffmatrix}) and $S_{ij}$ are the elements of $S$ in equation \eqref{eq:grmatrix}. We compared the numerical performance of MBO with that of BNC and Kmeans$++$, in the same spirit as in Section \ref{sc:bestpract}. The results are given in Figure \ref{fig:barabalbert}, and are consistent with those obtained for the SSBM model. 
\begin{figure}[!ht]
\centering
\subcaptionbox[]{  
}[ 0.32\textwidth ]
{\includegraphics[width=0.35\textwidth] {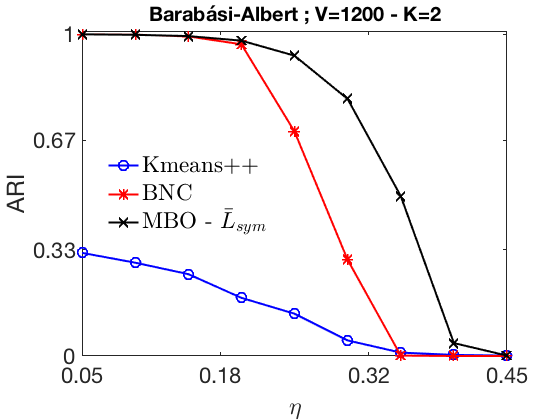} }
%
\subcaptionbox[]{  
}[ 0.32\textwidth ]
{\includegraphics[width=0.35\textwidth] {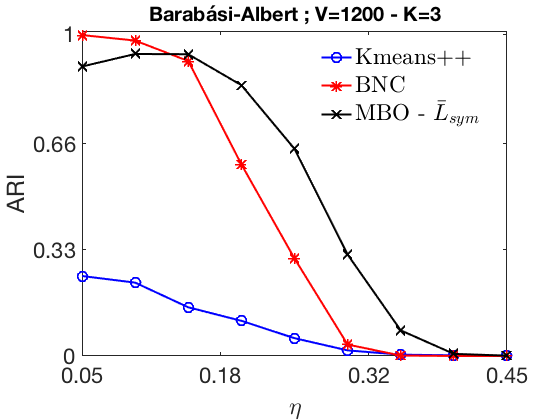} }
%
\subcaptionbox[]{  
}[ 0.32\textwidth ]
{\includegraphics[width=0.35\textwidth] {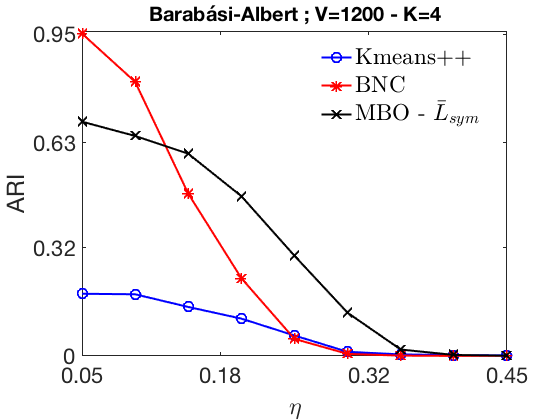} }
\subcaptionbox[]{  
}[ 0.32\textwidth ]
{\includegraphics[width=0.35\textwidth] {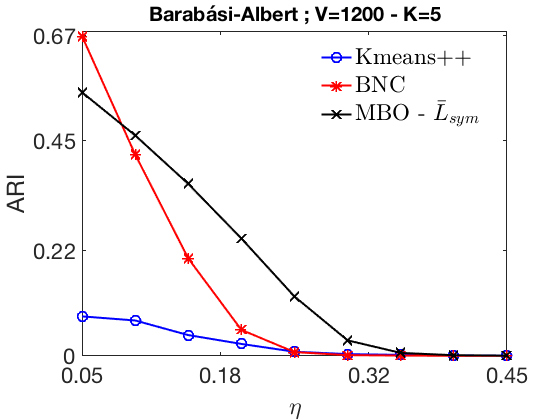} }
%
\subcaptionbox[]{  
}[ 0.32\textwidth ]
{\includegraphics[width=0.35\textwidth] {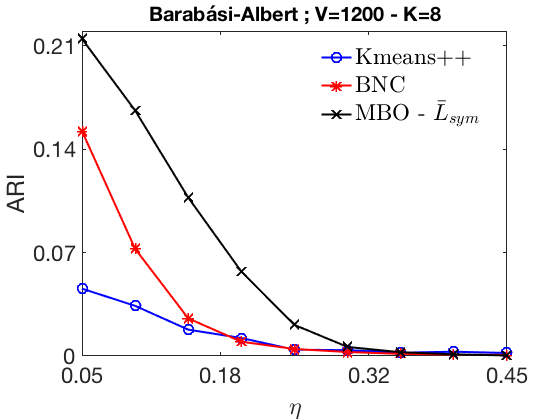} }
%
\subcaptionbox[]{  
}[ 0.32\textwidth ]
{\includegraphics[width=0.35\textwidth] {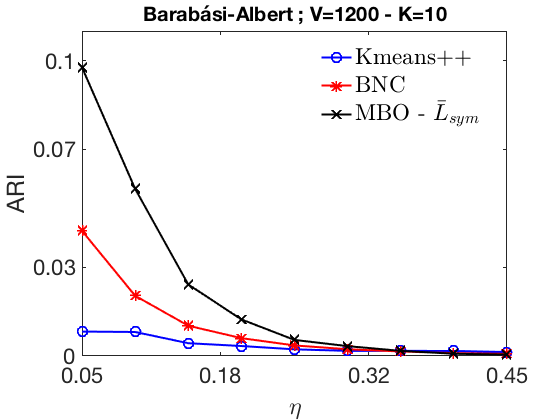} }
%
\captionsetup{width=0.95\linewidth}
\caption[Short Caption]{ These plots compare the ground truth retrieval capabilities of different algorithms using as affinity matrix a signed version of the Barab\'asi--Albert model, as explained in \eqref{eq:baAffMat}. As we can see, MBO dominates in the high noise regime for every sparsity level and number of clusters $K\in\{2,3,4\}$ we consider.  In addition, we observe that for larger number of clusters, MBO performs better than the BNC.
}
\label{fig:barabalbert}
\end{figure}

%
%
\section{Image segmentation}  \label{sec:imageSeg}

This section considers a proof-of-concept application of the Ginzburg--Landau functional under different types of constraints, as introduced in Section \ref{sc:GLFwithcons}, in the context of image segmentation. As an example, we consider a picture of Cape Circeo in Sabaudia (LT, Italy) \cite{circeolatoday}, Figure \ref{fig:sab_constclust} (a), in which it is possible to distinguish four main clusters which are the sand, sea, mountain, and sky. 

The network associated to the picture has been obtained in the following way. First, each pixel has a one-to-one correspondence with a node in the resulting graph. Second, each pixel is linked in the network only to a subset of its nearest-neighbouring pixels (in particular, we have considered each pixel to be connected to all the others which are present within a radius of $5$, as specified in the MATLAB Package \textit{The graph analysis toolbox: Image processing on arbitrary graphs}, see \cite{grady2003graph, gradyphdthesis}). Third, the weights of the affinity matrix corresponding to the picture can be obtained by considering the weight of the link between two pixels as a function of their color difference. In fact, considering the colors red, green, and blue (RGB) as the basis of a 3-dimensional vector space, each color can be described in terms of its RGB components in the same way as a vector can be characterized as the vector sum of the canonical basis elements weighted by some coefficients. In this way, we can naturally define the distance of two colors as a function of the Euclidean distance of their vectors in the RGB space. 
We used the distance function  
\begin{equation}\label{eq:distcolor}
d\left(i,j\right)\coloneqq
\begin{dcases}
e^{-b\left\Vert col\left(i\right)-col\left(j\right)\right\Vert _{2}}\;\in\left(0,1\right] & i\neq j\\
0 & i=j
\end{dcases},
\end{equation}
where $b$ is an external parameter (which we fix to $b=14$) and $col\left(i\right)$ is the corresponding color vector of pixel $i$ (note that the norm in \eqref{eq:distcolor} is not squared). The set of all the differences $\left\Vert col\left(i\right)-col\left(j\right)\right\Vert$ in a picture is normalized according to their \textit{z-score} using the MATLAB command \textit{normalize} \cite{normMatlab}. The MBO schemes used for this image segmentation problem run using $K=4$ clusters as an input parameter.
For constrained clustering, we have considered both must-links and cannot-links, as highlighted in Figure \ref{fig:sab_constclust}. 
One iteration (multiple iterations do not change the result) of the MBO scheme yields the clustering highlighted in Figure \ref{fig:sab_constclust_res} (a). Afterwards, the resulting $U$ has been used as the initial condition for two different constrained clustering optimizations, whose results are shown in Figure \ref{fig:sab_constclust_res}  (b) and (c) for the must-links and cannot-links, respectively. It can be inferred that the cluster recovery works well, except for a small area in the middle of the picture; however, we remark that even in the original picture (that is Figure \ref{fig:sab_constclust} (a) without blue and red square) it is difficult to the naked eye to understand whether that patch (that is the bit of sea to the bottom right of the mountain) corresponds to sea or mountain, just by looking at its color. 

\begin{figure}[!ht]
\centering
\subcaptionbox[]{  Known information
}[ 0.32\textwidth ]
{\includegraphics[width=0.35\textwidth] {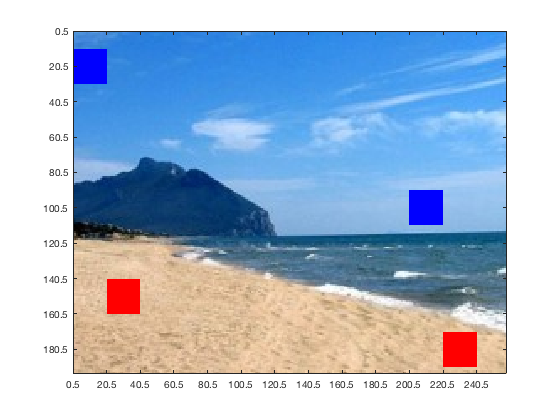} }
%
\subcaptionbox[]{  Must-links 
}[ 0.32\textwidth ]
{\includegraphics[width=0.35\textwidth] {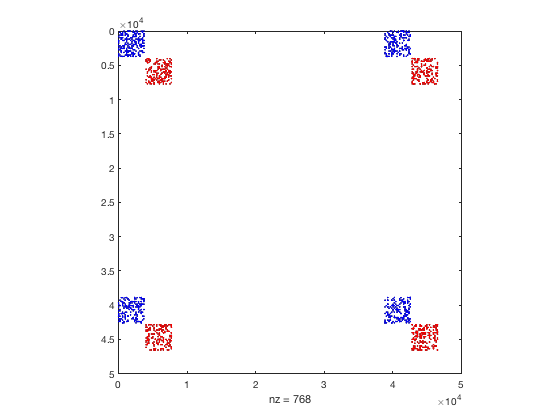} }
%
\subcaptionbox[]{  Cannot-links 
}[ 0.32\textwidth ]
{\includegraphics[width=0.35\textwidth] {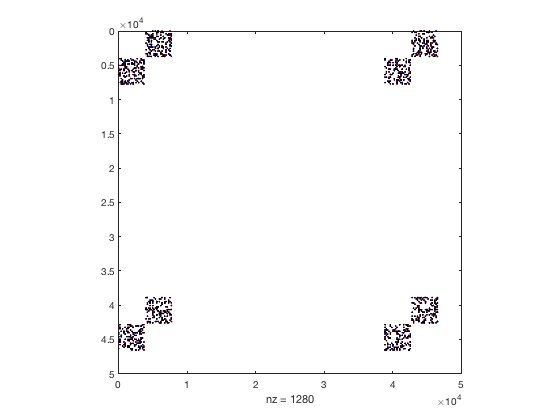} }
\captionsetup{width=0.95\linewidth}
\caption[Short Caption]{ From a picture of Cape Circeo in Sabaudia (LT, Italy) (a) With the same color we highlight the patches in the image which we would like to belong to the same cluster as a result of the MBO scheme. In terms of must-links, this means that all the nodes corresponding to pixels in the blue squares should be pairwise positively connected, and so should those corresponding to the red squares. For example, the top left blue patch in (a) corresponds to the top left blue patch in (b); the other blue patch in (a) corresponds to the bottom right patch in (b); the remaining blue patches in (b) pertain to the cross interactions between the pixels in the blue patches in (a).  Similar correspondences hold for the red patches. 
This can be seen in the affinity matrix in picture (b) where we have used the same color scheme as in (a). Instead, each blue/red node pair should be negatively connected since they belong to different clusters, as shown in figure (c). For the present case, we have weighted both must- and cannot-links with a prefactor of $2$. The reason why squares in (b) are not fully filled with color as in (a) is because, to reduce computational time when calculating eigenvectors, we have removed half of the must-link edges at random and half of the cannot-link edges at random (white color corresponds to missing edges).
}
\label{fig:sab_constclust}
\end{figure}

\begin{figure}[!ht]
\centering
\subcaptionbox[]{  No constraints
}[ 0.32\textwidth ]
{\includegraphics[width=0.35\textwidth] {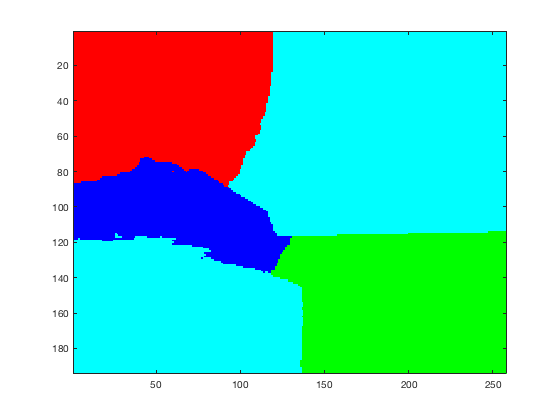} }
%
\subcaptionbox[]{  Must-links
}[ 0.32\textwidth ]
{\includegraphics[width=0.35\textwidth] {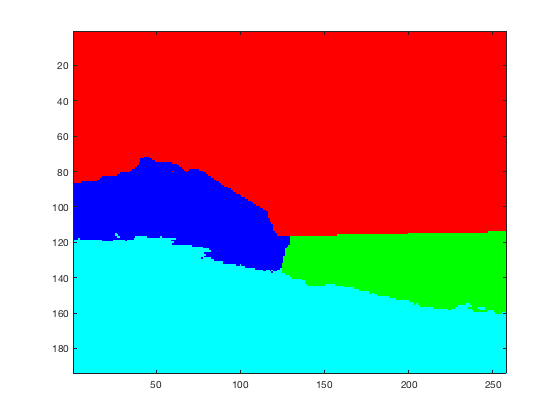} }
%
\subcaptionbox[]{  Cannot-links 
}[ 0.32\textwidth ]
{\includegraphics[width=0.35\textwidth] {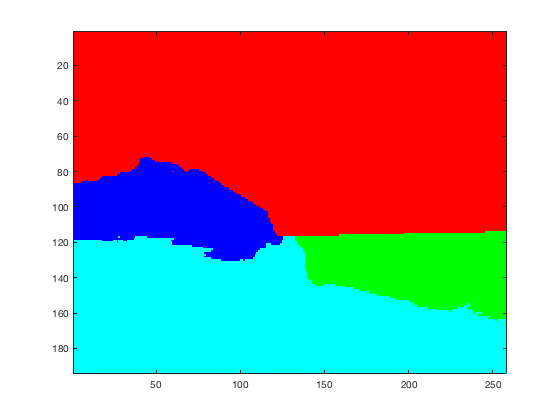} }
%
\captionsetup{width=0.95\linewidth}
\caption[Short Caption]{Segmentation results of the MBO scheme for the setting of (a) no constraints, (b) must-links, and (c) cannot-links, as chosen  in Figure \ref{fig:sab_constclust}.
} 
\label{fig:sab_constclust_res}
\end{figure}

In Figure \ref{fig:sab_fidelavoid}, we have taken into account fidelity and avoidance terms as highlighted in plot (a). The same information has been used for both the fidelity and the avoidance parts. For the former, the rows of the matrix $\hat{U}$ associated to the nodes in the colored patches have entries equal to $1$ at the corresponding cluster column (we have picked column $1$ for red, $2$ for blue, $3$ for cyan, and $4$ for green). All the remaining entries of $\hat{U}$ are equal to $0$. For the latter, we have considered an energy increase every time nodes are attributed to clusters to which they should not belong. This results in a constraint of the form $\tilde{U}=\hat{U}\left(1_{K}-I_{K}\right)$.
Starting from a random configuration of $U$, a single application of the MBO scheme using $\hat{U}$ yields the segmentation shown in Figure \ref{fig:sab_fidelavoid} (b) and a separate run of the algorithm using $\tilde{U}$ yields to (c). We note that the node-cluster association recovery results  agree  with the expected boundaries present within the image, with a minor uncertainty in the central area (which is a feature  present also in the constrained clustering case).

\begin{figure}[!ht]
\centering
\subcaptionbox[]{  Known information
}[ 0.32\textwidth ]
{\includegraphics[width=0.35\textwidth] {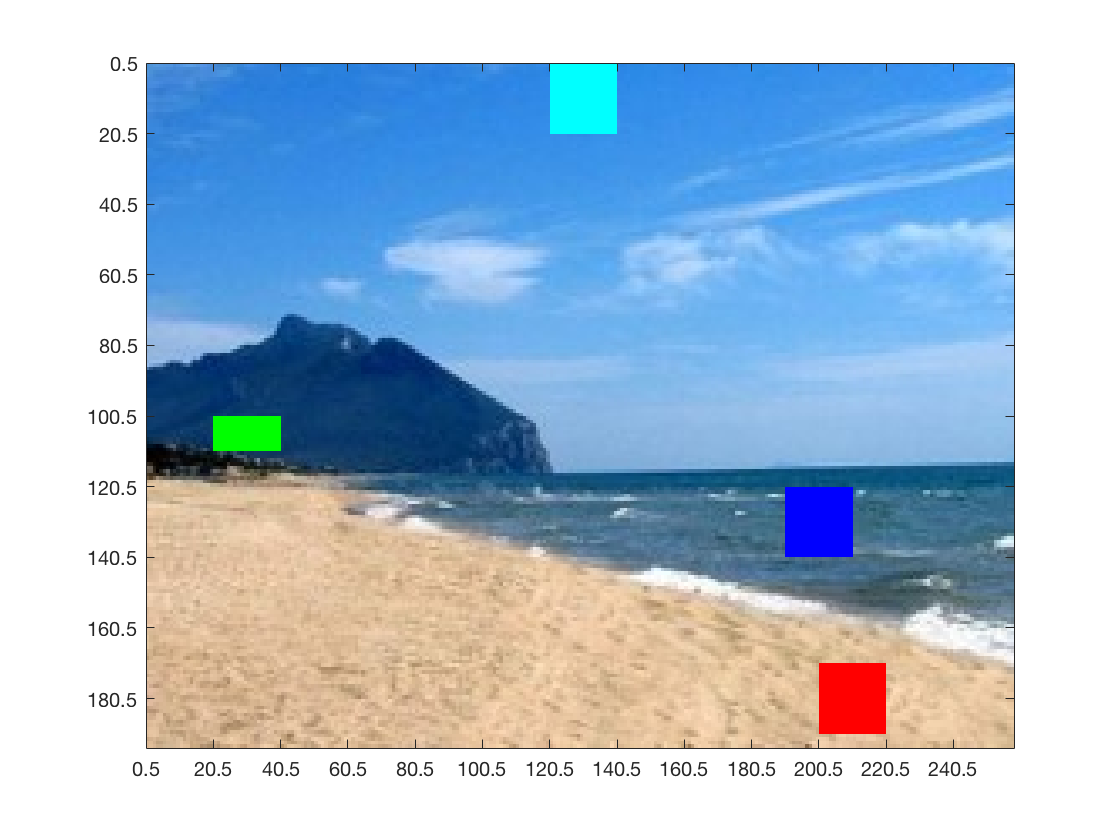} }
%
\subcaptionbox[]{  Fidelity 
}[ 0.32\textwidth ]
{\includegraphics[width=0.35\textwidth] {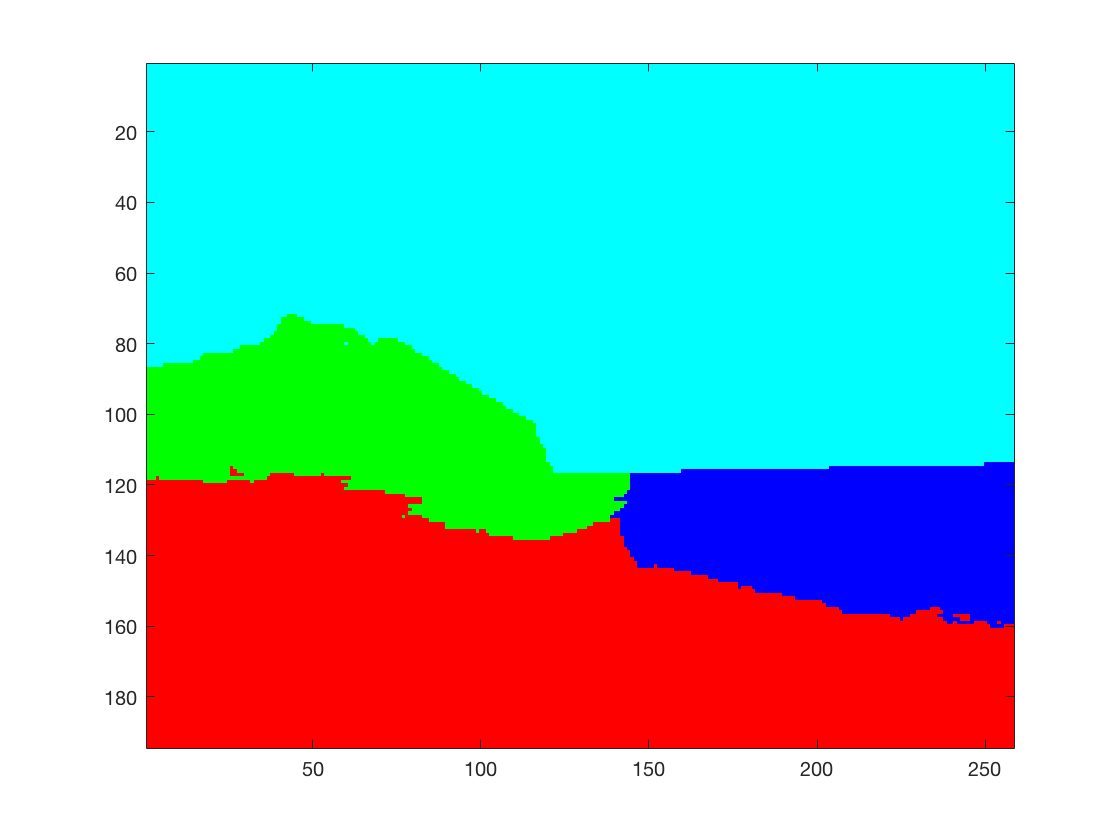} }
%
\subcaptionbox[]{  Avoidance 
}[ 0.32\textwidth ]
{\includegraphics[width=0.35\textwidth] {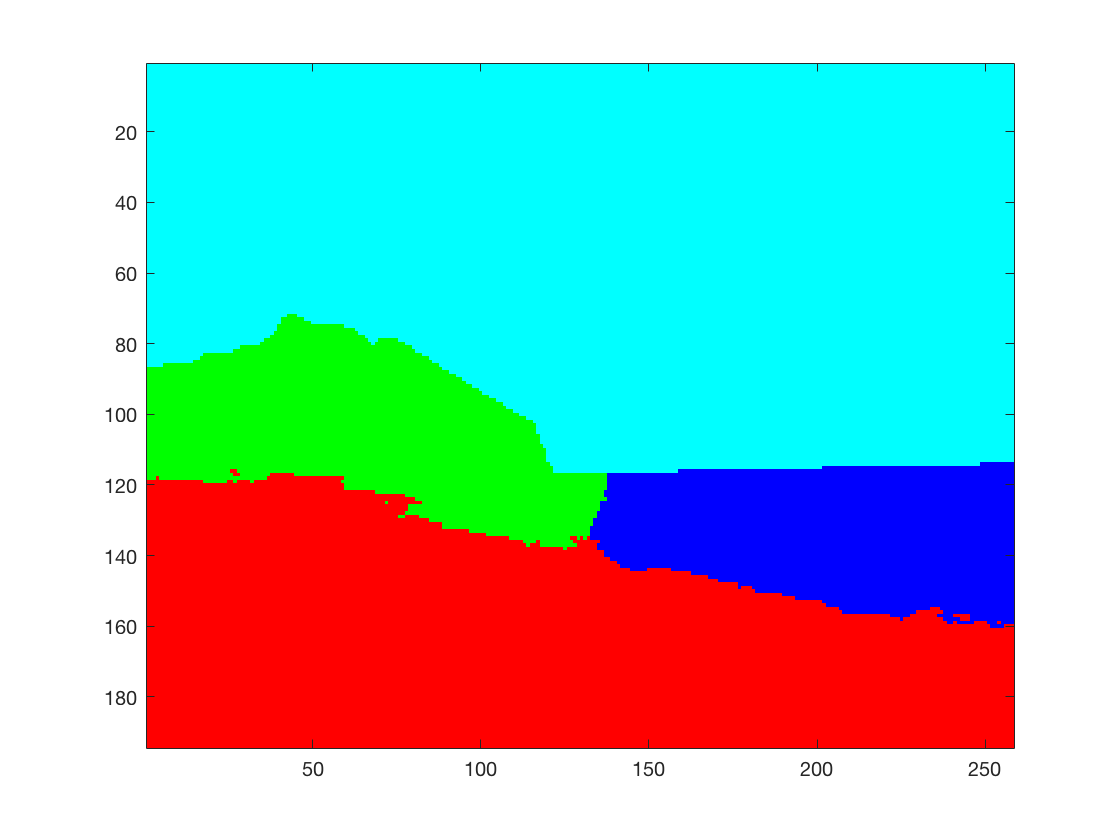} }
\captionsetup{width=0.95\linewidth}
\caption[Short Caption]{In plot (a), we color those pixels whose cluster membership is considered to be known a priori, with each color corresponding to a different cluster. Contrary to the constrained clustering case, the cluster assignment of these pixels is chosen a priori since it results from inserting the known information directly at the level of the characteristic matrix.
}
\label{fig:sab_fidelavoid}
\end{figure}
We have also considered the clusters highlighted in Figure \ref{fig:sab_fidelavoid} (a) as anchors, and the results are identical to those shown in Figure \ref{fig:sab_fidelavoid}  (b).

%
%
\section{Time series correlation clustering}  \label{sec:realData}

This section details the outcomes of our numerical experiments on two time series data sets. Since ground truth data is not available, we compare the output of the algorithms by plotting the network adjacency matrix sorted by cluster membership.

\subsection{S\&P 1500 Index}
We use time series daily prices for $n=1500$ financial instruments, constituents of the S\&P 1500 Composite Index, during 2003-2015, containing approximately $ n_d = 3000$ trading days \cite{sp500yahoo}. We  consider daily log returns of the prices, 
     \begin{equation}
        R_{i,t} = \log \frac{ P_{i,t} }{ P_{i,t-1}},
     \end{equation} 
where $P_{i,t}$ denotes the market close price of instrument $i$ on day $t$. 
In order to promote both positive and negative weights in the final correlation matrix we aim to cluster, we also compute the daily market excess return for each instrument,
 \begin{equation} 
    \tilde{R}_{i,t} = R_{i,t} -  R_{ \text{SPY},t},  \forall i=1, \ldots, n, \; t=1,\ldots,n_d,
\end{equation} 
where $R_{\text{SPY},t}$ denotes the daily log return of  \textsc{SPY}, the S\&P 500 index ETF which is used as a proxy (weighted average) for the market. 
Next, we calculate the Pearson correlation between historical returns for each pair of instruments $(i,j)$, which ultimately  provides the edge weight $A_{ij}$ in our signed graph.
Figure \ref{fig:SP1500} shows a segmentation of the US market into  $K \in \{ 10, 30 \}$ clusters. The cluster structure is more obvious to the naked eye for the BNC and MBO methods, as opposed to the plain Kmeans$++$ clustering. The large cluster recovered by MBO could potentially correspond to the 1000 instruments which are within the S\&P 1500 Index and  outside of the S\&P 500, but we defer this investigation to future work.

\begin{figure}[!ht]
\centering
\subcaptionbox[]{  Kmeans$++$ ; $K=10$
}[ 0.32\textwidth ]
{\includegraphics[width=0.35\textwidth] {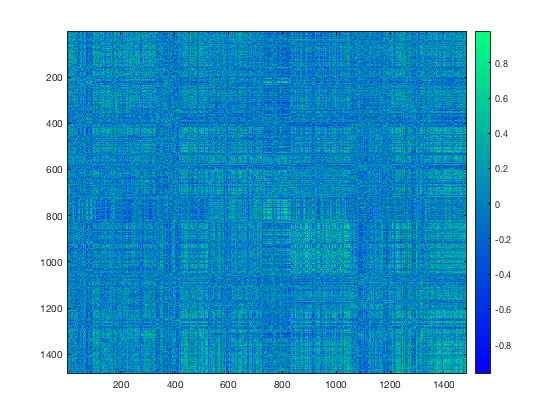} }
%
\subcaptionbox[]{  BNC ; $K=10$
}[ 0.32\textwidth ]
{\includegraphics[width=0.35\textwidth] {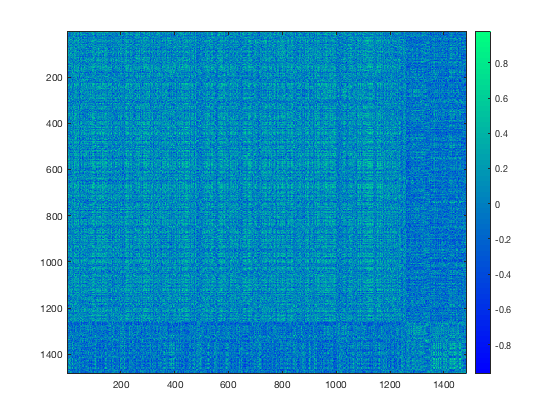} }
%
\subcaptionbox[]{  MBO ; $K=10$
}[ 0.32\textwidth ]
{\includegraphics[width=0.35\textwidth] {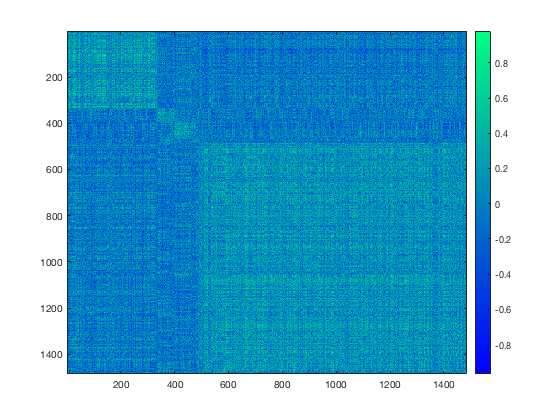} }
\subcaptionbox[]{  Kmeans$++$ ; $K=30$
}[ 0.32\textwidth ]
{\includegraphics[width=0.35\textwidth] {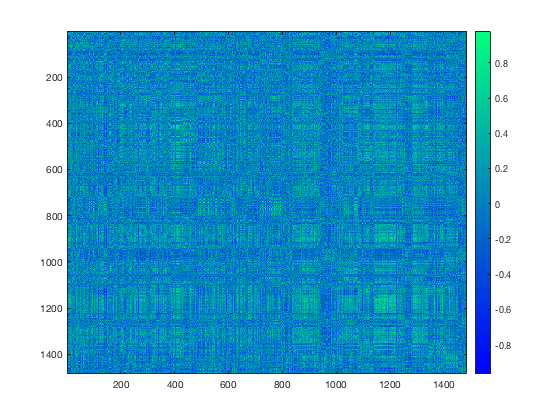} }
%
\subcaptionbox[]{  BNC ; $K=30$
}[ 0.32\textwidth ]
{\includegraphics[width=0.35\textwidth] {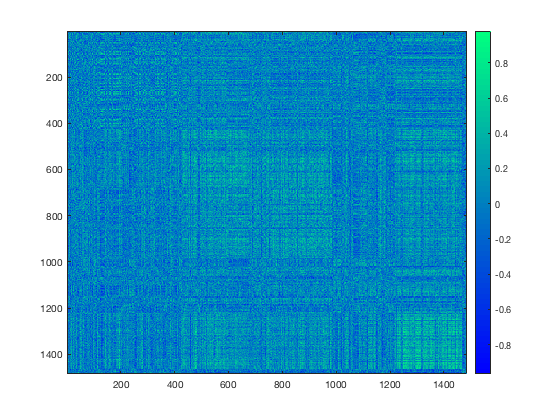} }
%
\subcaptionbox[]{  MBO ; $K=30$
}[ 0.32\textwidth ]
{\includegraphics[width=0.35\textwidth] {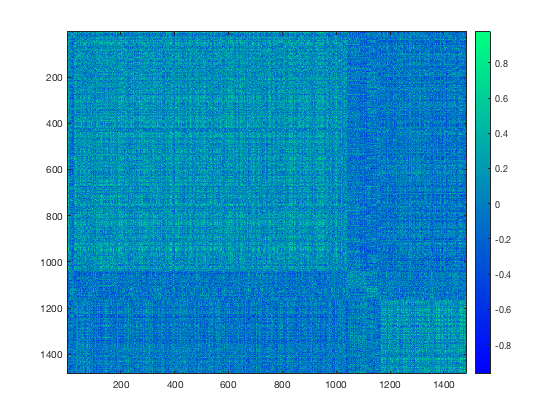} }
%
\captionsetup{width=0.95\linewidth}
\caption[Short Caption]{ The adjacency (correlation) matrix of the S\&P 1500 time series data set, sorted by cluster membership, for $K \in \{10,30\}$ clusters.
}
\label{fig:SP1500}
\end{figure}

\subsection{Foreign exchange}
This experiment pertains to clustering a foreign exchange 
FX currency correlation matrix, based on daily special drawing rights (SDR)  exchange value rates \cite{FX_IMF_SDR}. As shown in Figure \ref{fig:forex}, only BNC and MBO are able to recover four clusters, associated respectively to the EURO \euro, US Dollar \$, UK Pound Sterling \textsterling, and Japanese Yen  \yen. These are precisely the four currencies that, in certain percentage weights, define the value of the SDR reserve. To the naked eye, the adjacency matrix (sorted by cluster membership) recovered by MBO exhibits sharper transitions between clusters, when compared to BNC and especially Kmeans$++$. 

\begin{figure}[!ht]
\centering
\subcaptionbox[]{  Kmeans$++$
}[ 0.32\textwidth ]
{\includegraphics[width=0.35\textwidth] {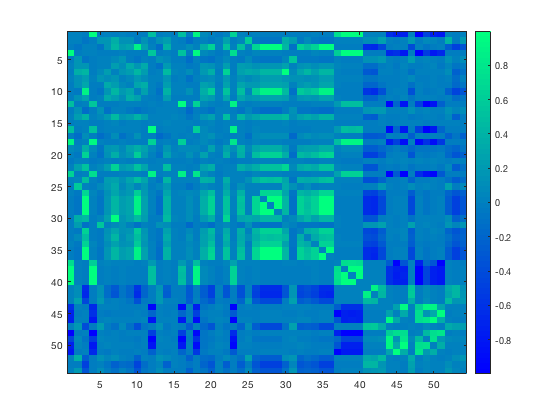} }
%
\subcaptionbox[]{  BNC 
}[ 0.32\textwidth ]
{\includegraphics[width=0.35\textwidth] {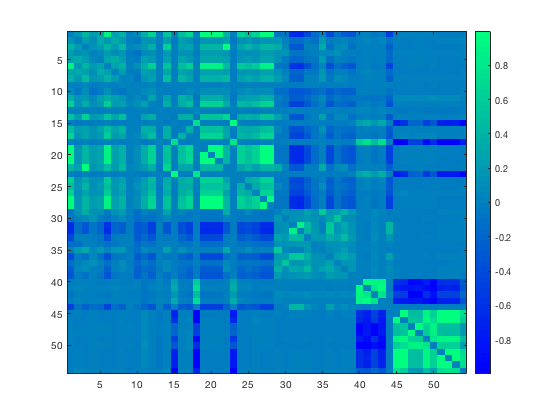} }
%
\subcaptionbox[]{  MBO 
}[ 0.32\textwidth ]
{\includegraphics[width=0.35\textwidth] {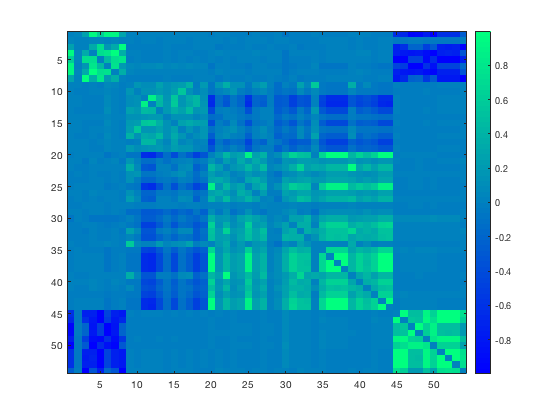} }
\captionsetup{width=0.95\linewidth}
\caption[Short Caption]{The adjacency (correlation) matrix of the Foreign Exchange time series data set, sorted by cluster membership, for $K = 4$ clusters. We remark that MBO exhibits sharper transitions between the four clusters (as highlighted by the green submatrices on the diagonal blocks), when compared to BNC and especially  Kmeans$++$.
}
\label{fig:forex}
\end{figure}

%
%
\section{US migration}  \label{sec:migration}
The final real world data set we consider is a migration network based on the 2000 US Census report, 
which captures migration patterns  between all pairs of counties in the US during the 1995-2000 time frame \cite{census,census_rep}. The network has  $V=3107$  nodes, where each node corresponds to a county in mainland US. The adjacency matrix of the given network is asymmetric, with $M_{ij} \neq M_{ji}, \; 1 \leq i,j \leq V$. 
Similar to the approach in \cite{belgium2010}, we first symmetrize the input matrix, and let $\tilde{M}_{ij} = M_{ij} + M_{ji}$ denote  the total number of people that migrated between county $i$ and county $j$ (thus $\tilde{M}_{ij} = \tilde{M}_{ji}$) during the five-year period.  Finally, we turn the network into a signed one by approximately centering the distribution of edge weights at 0. More specifically, from each nonzero entry we remove the median value $\tilde{m}$ of the distribution of nonzero edge weights in $\tilde{M}$, and denote the resulting symmetric  matrix by $A$, with $A_{ij} = \tilde{M}_{ij} -\tilde{m} $.  
In Figure \ref{fig:migration}, we compare the outputs of plain Kmeans++, BNC, and our MBO approach, for a varying number of clusters $K\in\{5,10,15,20,30,40,50\}$. It can be seen that the plain Kmeans++ approach is not able to return any meaningful results across all values of $k$, and a similar statement holds true for BNC for low values of $K\in\{5,10,15,20\}$: in both cases, a single cluster captures almost all the nodes in the network. BNC returns meaningful results ---in the sense that there are no clusters which contain only a few nodes--- for larger numbers of clusters  $K\in\{30,40,50\}$, aligning well with state boundaries. On the other hand, our proposed MBO algorithm returns meaningful results across all values of $K$, and the clusters align well with the true 
boundaries between the states, as previously observed in \cite{belgium2010}.

\newcommand{\wid}{1.7in}
\newcolumntype{C}{>{\centering\arraybackslash}m{\wid}}
\begin{table*}[!htp]\sffamily
\hspace{-1mm}
\begin{tabular}{l*3{C}@{}}
   & Kmeans++ & BNC & MBO  \\  
$k=5$
& \includegraphics[width=\wid]{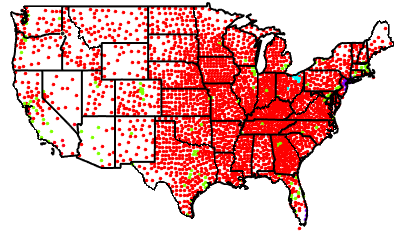} 
& \includegraphics[width=\wid]{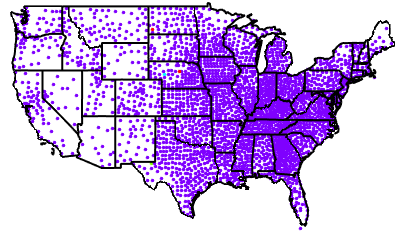}
& \includegraphics[width=\wid]{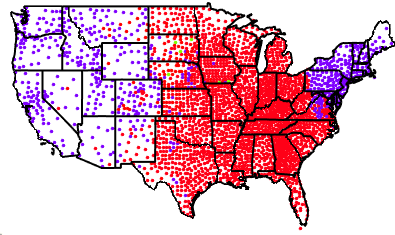} \\ 
$k=10$
& \includegraphics[width=\wid]{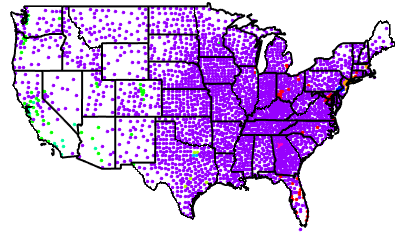} 
& \includegraphics[width=\wid]{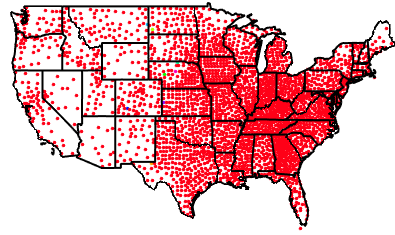}
& \includegraphics[width=\wid]{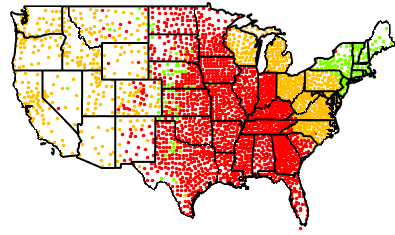}  \\
$k=15$
& \includegraphics[width=\wid]{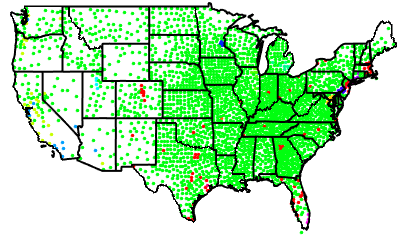} 
& \includegraphics[width=\wid]{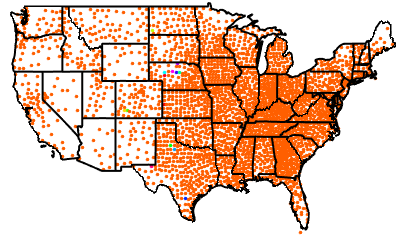}
& \includegraphics[width=\wid]{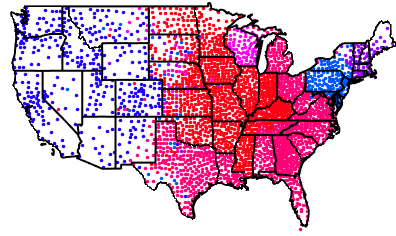}  \\
$k=20$
& \includegraphics[width=\wid]{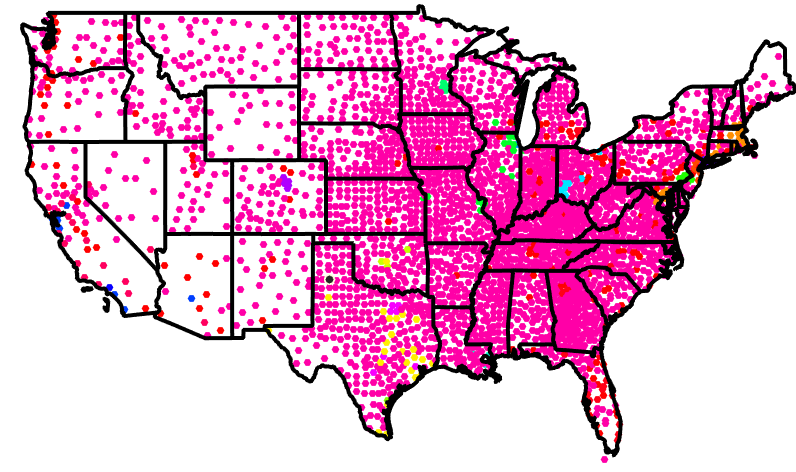} 
& \includegraphics[width=\wid]{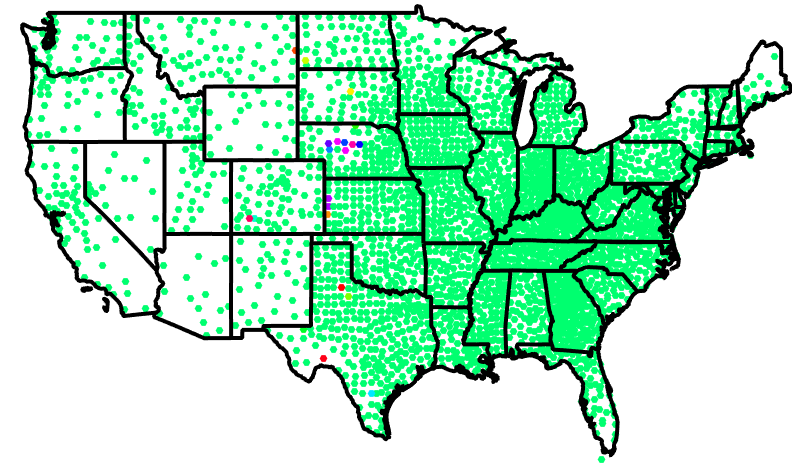}
& \includegraphics[width=\wid]{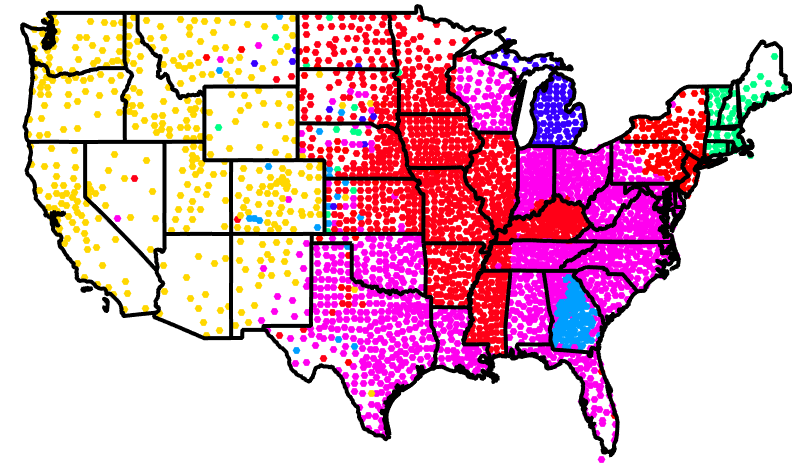}  \\
$k=30$
& \includegraphics[width=\wid]{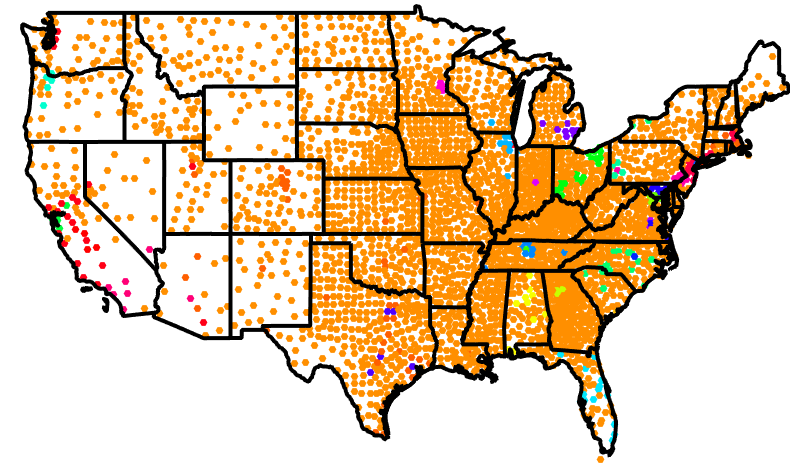} 
& \includegraphics[width=\wid]{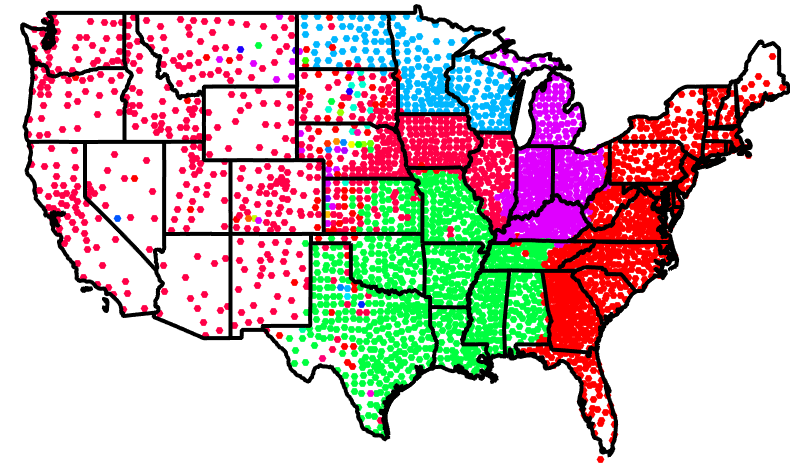}
& \includegraphics[width=\wid]{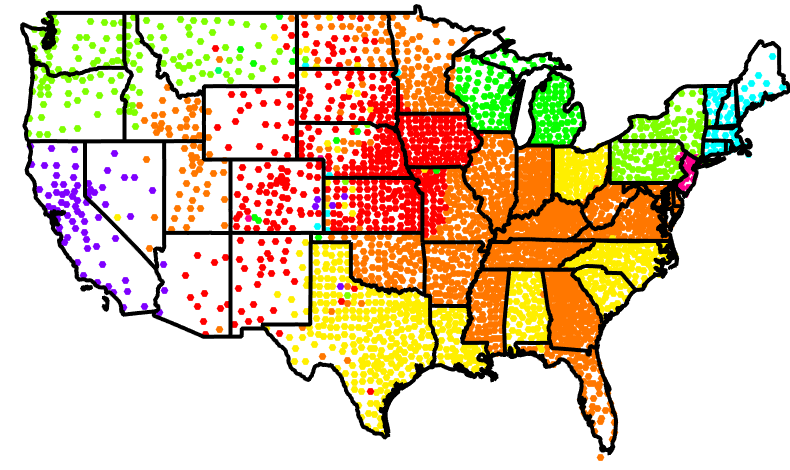}  \\
$k=40$
& \includegraphics[width=\wid]{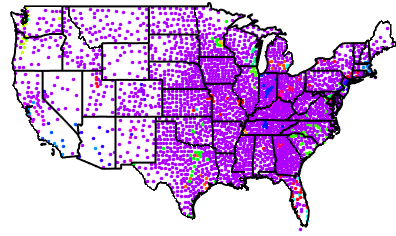} 
& \includegraphics[width=\wid]{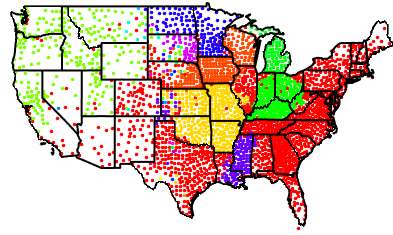}
& \includegraphics[width=\wid]{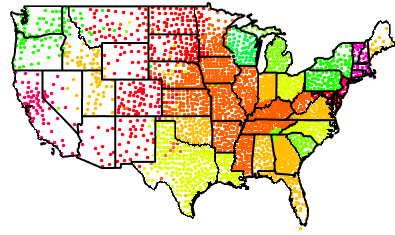}  \\
$k=50$
& \includegraphics[width=\wid]{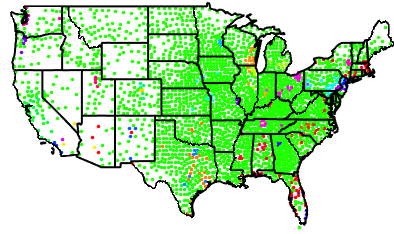} 
& \includegraphics[width=\wid]{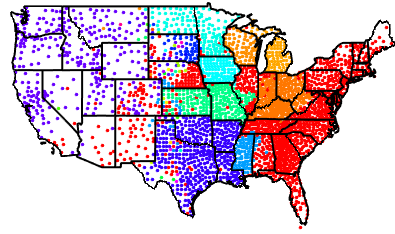}
& \includegraphics[width=\wid]{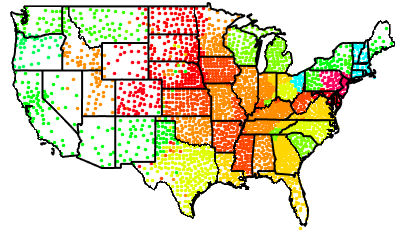}  
\end{tabular}
\captionsetup{width=0.99\columnwidth}
\vspace{-1mm}
\captionof{figure}{  Clusterings computed by Kmeans++, BNC, and MBO for the signed US Migration network, for varying number of clusters $k$. Each color, extracted from a continuous color map, corresponds to a different cluster.
}
\label{fig:migration}	
\end{table*}

\section{Summary and future directions}   \label{sec:conclusion}

We introduced graph-based diffuse interface models utilizing the Ginzburg--Landau functionals, using an adaptation of the classic numerical Merriman--Bence--Osher (MBO) scheme for minimizing such graph-based functionals. Our approach has the advantage that it can easily incorporate various types of labelled data available to the user, in the context of semi-supervised clustering.

We have provided extensive numerical experiments that demonstrate the robustness to noise and sampling sparsity of the MBO scheme, under a signed stochastic block model. In particular, for sparse graphs and larger number of clusters, we often obtain more accurate results compared to state-of-the-art methods such as the BNC and the signed Laplacian.

Our work opens up a number of interesting directions  for future investigation. It would be insightful to investigate the performance of the MBO scheme in the setting of very sparse graphs in the stochastic block model with $p = \Theta(\frac{1}{n})$, thus below the connectivity threshold, and potentially incorporate   regularization terms considered in the recent literature \cite{joseph2016impact, le2015sparse_Vershynin}.

As a by-product of our research, we expect that the low-dimensional embeddings of signed  networks  (e.g., correlation matrices or related similarity/dissimilarity measures) arising from high-dimensional time series data obtained via the MBO scheme, could be of independent interest  in the context of  robust dimensionality reduction in multivariate time series analysis. Such embeddings could be used to learn nonlinear mappings from multivariate time series data, which has broad applications ranging from financial data analysis to biomedical research.  Such mappings could prove useful for exploiting weak temporal correlations inherent in sequential data, and can be subsequently augmented with a regression or classification step, with the end goal of improving out-of-sample prediction. In a financial setting, this corresponds to building nonlinear risk/factor models for financial instruments, a task of paramount importance when it comes to understanding the risks that govern a financial system  \cite{lalouxBouchaud2000random,ArbitrageMultiFactor}.

An interesting potential future direction is the development of an MBO scheme for detecting structures in directed graphs. The vast majority of methods for structural analysis of graphs, for tasks such as clustering and community detection, have been developed for undirected graphs, and a number of challenges arise when considering directed graphs, complicating the immediate extension of existing algorithms for the undirected case. Various directed extensions typically consider a symmetrized version of the adjacency matrix \cite{pentney2005spectralMeila, gleich2006hierarchical,meila2007clusteringDirected,sussman2012consistent}. In ongoing work by a subset of the authors \cite{DirSpectralClustHerm}, we consider the problem of clustering directed graphs, where the clustering structure arises due to the existence of certain pairs of imbalanced cuts in the graph. More specifically, an imbalanced cut between two sets of vertices X and Y is one where the majority of the edges are directed from X to Y, with few edges being directed in the opposite direction from Y to X. The goal is to find $k$ clusters such that the sum of all (or many of) the ${k \choose 2}$ pairwise imbalanced cuts is maximized, while also penalizing the existence of small clusters in order to obtain a more balanced partition. In \cite{DirSpectralClustHerm}, we consider a Hermitian matrix derived directly from the skew-symmetric adjacency matrix, and recover an embedding from the spectrum of its corresponding (Hermitian) graph connection Laplacian, potentially after normalization \cite{asap2d, Singer_Hautieng_VDM}. Extending the Ginzburg--Landau functional and the MBO scheme to the setting of directed graphs and Hermitian Laplacians is an appealing future research direction.

%
%
%
%
\section*{Acknowledgements}
MC and AP acknowledge support from The Alan Turing Institute EPSRC grant EP/N510129/1 and seed funding  project SF029 \textit{``Predictive graph analytics and propagation of information in networks"}. AP also acknowledges support from the National Group of Mathematical Physics (GNFM-INdAM), by Imperial College together with the Data Science Institute and Thomson-Reuters Grant No. 4500902397-3408 and EPSRC grant EP/P002625/1. YvG did most of the work which has contributed to this paper when he was a lecturer at the University of Nottingham. YvG acknowledges that this project has received funding from the European Union's Horizon 2020 research and innovation programme under the Marie Sk\l{}odowska-Curie grant agreement No 777826. Francesco Cosentino is acknowledged for his contribution in the derivation of the proof of Lemma \ref{lm:sec}.

%
%
%
%

\appendix

\section*{Appendix}
\addcontentsline{toc}{section}{Appendix}

\setcounter{section}{1}

\subsection{Projection and Thresholding on $\Sigma_K$ and $\Sigma_K^{\pm}$}\label{app:projsimplex}
Let $K\in \mathbb{N}_+$. In this section, we show that computation of the projection on $\Sigma_{K}$ (or $\Sigma_{K}^{\pm}$) followed by the thresholding step in the MBO scheme in \eqref{eq:projOLD}--\eqref{eq:threOLD} is equivalent to determining a maximal component of every vector $\underbar{\ensuremath{u}}_i^{n+\frac{1}{2}}$. To fix notation, the vertices of $\Sigma_K$ are the elements of the canonical basis in $\mathbb{R^K}$ (for instance, for $K=3$ we have $\underbar{\ensuremath{e}}_1=\left(1,0,0\right)^T$, $\underbar{\ensuremath{e}}_2=\left(0,1,0\right)^T$ and $\underbar{\ensuremath{e}}_3=\left(0,0,1\right)^T$), while those of $\Sigma_{K}^{\pm}$ are defined as $\underbar{\ensuremath{e}}_{k}^{\pm}\coloneqq2\underbar{\ensuremath{e}}_{k}-\underbar{1}_K\in\left\{ -1;1\right\} ^{V}$ (that is, for the $K=3$ case, we have $\underbar{\ensuremath{e}}^{\pm}_1=\left(1,-1,-1\right)^T$, $\underbar{\ensuremath{e}}^{\pm}_2=\left(-1,1,-1\right)^T$ and $\underbar{\ensuremath{e}}^{\pm}_3=\left(-1,-1,1\right)^T$).

We first remind the reader of a useful equality for vectors $\underbar{\ensuremath{a}}$, $\underbar{\ensuremath{b}} \in \mathbb{R}^K$:
\begin{equation}\label{eq:normequality}
\|\underbar{\ensuremath{a}}-\underbar{\ensuremath{b}}\|_2^2 - \|\underbar{\ensuremath{a}}-\underbar{\ensuremath{c}}\|_2^2 = \|\underbar{\ensuremath{a}}\|_2^2 + \|\underbar{\ensuremath{b}}\|_2^2 - 2\langle \underbar{\ensuremath{a}}, \underbar{\ensuremath{b}}\rangle_2 - \|\underbar{\ensuremath{a}}\|_2^2 - \|\underbar{\ensuremath{c}}\|_2^2 + 2 \langle \underbar{\ensuremath{a}}, \underbar{\ensuremath{c}}\rangle_2 = \|\underbar{\ensuremath{b}}\|_2^2 - \|\underbar{\ensuremath{c}}\|_2^2 + 2 \langle \underbar{\ensuremath{a}}, \underbar{\ensuremath{c}}-\underbar{\ensuremath{b}}\rangle_2,
\end{equation}
where we have used the standard inner product in $\R^K$: $\langle \underbar{\ensuremath{a}}, \underbar{\ensuremath{b}}\rangle_2 := \sum_{k=1}^K a_k b_k$. The following lemma lists three direct consequences of \eqref{eq:normequality}.
\begin{lemma}\label{lm:consequences}
\begin{enumerate}
\item Given $\underbar{\ensuremath{a}}$, $\underbar{\ensuremath{b}} \in \mathbb{R}^K$ and $i, j\in \{1, \ldots, K\}$, define $\underbar{\ensuremath{b}}^*\in \mathbb{R}^K$ by
\begin{equation}
b_k^*:=\begin{cases}
b_k, &\text{if } i\neq k \neq j,\\
b_j, &\text{if } k=i,\\
b_i, &\text{if } k=j.
\end{cases}
\end{equation}
Then
\begin{equation}
\|\underbar{\ensuremath{a}}-\underbar{\ensuremath{b}}^*\|_2^2 = \|\underbar{\ensuremath{a}}-\underbar{\ensuremath{b}}\|_2^2 + 2 (a_i-a_j) (b_i-b_j).
\end{equation}

\item Let $\underbar{\ensuremath{x}}\in\mathbb{R}^{K}$ and let $\underbar{\ensuremath{e}}_{i},  \underbar{\ensuremath{e}}_{j} \in\mathbb{R}^{K}$ be two vertices of the simplex $\Sigma_K$. Then $\left\Vert \underbar{\ensuremath{x}}-\underbar{\ensuremath{e}}_{i}\right\Vert_2 \geqslant\left\Vert \underbar{\ensuremath{x}}-\underbar{\ensuremath{e}}_{j}\right\Vert_2 $ iff $x_{j}\geqslant x_{i}$, with equality between the norms iff $x_i=x_j$.

\item Let $\underbar{\ensuremath{x}}\in\mathbb{R}^{K}$ and let $\underbar{\ensuremath{e}}_{i}^\pm,  \underbar{\ensuremath{e}}_{j}^\pm \in\mathbb{R}^{K}$ be two vertices of the simplex $\Sigma_K^\pm$. Then $\left\Vert \underbar{\ensuremath{x}}-\underbar{\ensuremath{e}}_{i}^{\pm}\right\Vert_2 \geqslant\left\Vert \underbar{\ensuremath{x}}-\underbar{\ensuremath{e}}_{j}^{\pm}\right\Vert_2 $ iff $x_{j}\geqslant x_{i}$, with equality between the norms iff $x_i=x_j$.
\end{enumerate}
\end{lemma}
\begin{proof}
\begin{enumerate}
\item We apply \eqref{eq:normequality} with $\underbar{\ensuremath{c}} = \underbar{\ensuremath{b}}^*$. Since $\|\underbar{\ensuremath{b}}\|_2 = \|\underbar{\ensuremath{b}}^*\|_2$ and 
\begin{equation}
\langle \underbar{\ensuremath{a}}, \underbar{\ensuremath{c}}-\underbar{\ensuremath{b}}\rangle_2 = a_i (b_j-b_i) + a_j (b_i-b_j) = (a_i-a_j)(b_j-b_i),
\end{equation}
the result follows.

\item We apply \eqref{eq:normequality} with $\underbar{\ensuremath{a}} = \underbar{\ensuremath{x}}$, $\underbar{\ensuremath{b}} = \underbar{\ensuremath{e}}_i$, and $\underbar{\ensuremath{c}} = \underbar{\ensuremath{e}}_j$. Since $\|\underbar{\ensuremath{e}}_i\|_2 = 1 = \|\underbar{\ensuremath{e}}_j\|_2$ and $\langle \underbar{\ensuremath{x}}, \underbar{\ensuremath{e}}_j-\underbar{\ensuremath{e}}_i\rangle_2 = x_j-x_i$, it follows that
\begin{equation}
\|\underbar{\ensuremath{x}}-\underbar{\ensuremath{e}}_i\|_2^2 - \|\underbar{\ensuremath{x}}-\underbar{\ensuremath{e}}_j\|_2^2 = 2(x_j-x_i),
\end{equation}
and thus the result follows.

\item This follows from a proof completely analogous to the previous one, once we note that $\|\underbar{\ensuremath{e}}^\pm_i\|_2 = \sqrt K = \|\underbar{\ensuremath{e}}^\pm_j\|_2$ and $\langle \underbar{\ensuremath{x}}, \underbar{\ensuremath{e}}^ \pm_j-\underbar{\ensuremath{e}}^\pm_i\rangle_2 = 2(x_j-x_i)$ where $\underbar{\ensuremath{x}}\in\mathbb{R}^K$.
\end{enumerate}
\end{proof}

\begin{lemma}\label{lm:sec}
Let $\underbar{\ensuremath{x}}\in\mathbb{R}^{K}$ and 
\begin{equation}\label{eq:defproj}
\underbar{\ensuremath{x}}^{\perp}\coloneqq\argmin_{\underbar{\ensuremath{y}}\in\Sigma_K}\left\Vert \underbar{\ensuremath{x}}-\underbar{\ensuremath{y}}\right\Vert_2 
\end{equation}
be its projection on $\Sigma_K$\footnote{Existence and uniqueness of $\underbar{\ensuremath{x}}^{\perp}$ is a well-known result, see for instance Lemma 2.2 in \cite{bagirov2014introduction}}. Given two components $x_j$ and $x_i$ of $\underbar{\ensuremath{x}}$, if $x_{j}\geqslant x_{i}$ then $x_{j}^{\perp}\geqslant x_{i}^{\perp}$. The analogous result holds for the projection of $\underbar{\ensuremath{x}}$ on $\Sigma_K^\pm$.
\end{lemma}
\begin{proof}
We will prove the result for the projection on $\Sigma_K$. A completely analogous proof gives the result for $\Sigma_K^\pm$.

First note that the result is trivially true if $i=j$, hence we assume that $i\neq j$. For a proof by contradiction, assume that $x_j \geq x_i$ and $x_j^\perp < x_i^\perp$. Define $\underbar{\ensuremath{x}}^{\perp*}$ by 
\begin{equation}
x^{\perp*}_k := \begin{cases}
x^\perp_k, &\text{if } i\neq k \neq j,\\
x^\perp_j, &\text{if } k=i,\\
x^\perp_i, &\text{if } k=j.
\end{cases}
\end{equation}
Then, by the first result of Lemma~\ref{lm:consequences} (with $\underbar{\ensuremath{a}}=\underbar{\ensuremath{x}}$ and $\underbar{\ensuremath{b}}=\underbar{\ensuremath{x}}^\perp$), we have that
\begin{equation}
\|\underbar{\ensuremath{x}}-\underbar{\ensuremath{x}}^{\perp*}\|_2^2 - \|\underbar{\ensuremath{x}}-\underbar{\ensuremath{x}}^\perp\|_2^2 = 2 (x_i-x_j) (x_i^\perp - x_j^\perp) \leq 0,
\end{equation}
The inequality follows from our assumptions on $x_i$, $x_j$, $x_i^\perp$, and $x_j^\perp$. Since $\underbar{\ensuremath{x}}^\perp \in \Sigma_K$, we have $\sum_{k=1}^K x^{\perp*}_k = \sum_{k=1}^K x^\perp_k = 1$, hence $\underbar{\ensuremath{x}}^{\perp*} \in \Sigma_K$. Thus, per definition \eqref{eq:defproj}, we have $\|\underbar{\ensuremath{x}}-\underbar{\ensuremath{x}}^{\perp*}\|_2^2 - \|\underbar{\ensuremath{x}}-\underbar{\ensuremath{x}}^\perp\|_2^2 \geq 0$. It follows that $\|\underbar{\ensuremath{x}}-\underbar{\ensuremath{x}}^{\perp*}\|_2 = \|\underbar{\ensuremath{x}}-\underbar{\ensuremath{x}}^\perp\|_2$. By uniqueness of the minimizer in \eqref{eq:defproj} we deduce that $\underbar{\ensuremath{x}}^{\perp*} = \underbar{\ensuremath{x}}^\perp$. In particular that means that $x^\perp_j = x^{\perp*}_j = x^\perp_i$, which is a contradiction.
\end{proof}
Applying Lemma~\ref{lm:consequences} (result 2) to the thresholding step in \eqref{eq:threOLD}, we find that $k$ in \eqref{eq:threOLD} is such that the $k^{\text{th}}$ entry of $\underbar{v}_i$ is not smaller than all its other entries. Furthermore, Lemma~\ref{lm:sec} applied to \eqref{eq:projOLD} tells us that $\underbar{\ensuremath{v}}_i$ and $\left(\underbar{\ensuremath{u}}_{i}^{n+\frac{1}{2}}\right)$ have their largest entries at the same positions. Thus \eqref{eq:projOLD} and \eqref{eq:threOLD} are equivalent to
$\underbar{\ensuremath{u}}_{i}^{n+1}=\underbar{\ensuremath{e}}_{k}$, where $k$ is such that the $k^{\text{th}}$ entry of $\left(\underbar{\ensuremath{u}}_{i}^{n+\frac{1}{2}}\right)$ is not smaller than all its other entries.
%
%
%
%

\label{Bibliography}

\end{document}